\documentclass{article}

    \PassOptionsToPackage{numbers, sort&compress}{natbib}

 \usepackage[main, final]{neurips_2026}

\usepackage[utf8]{inputenc} 
\usepackage[T1]{fontenc}    
\usepackage{hyperref}       
\usepackage{url}            
\usepackage{booktabs}       
\usepackage{amsfonts}       
\usepackage{nicefrac}       
\usepackage{microtype}      
\usepackage{xcolor}         

\hypersetup{hidelinks}

\usepackage{amsmath,amsthm,amssymb}
\usepackage{enumerate}
\usepackage[shortlabels]{enumitem}
\usepackage{makecell}
\usepackage{subcaption}

\usepackage{algorithm}
\usepackage{algorithmic}

\usepackage{bm}

\usepackage{tikz}
\usepackage{pgfplots}
\usepackage{pifont}

\newcommand{\cmark}{\ding{51}}%
\newcommand{\xmark}{\ding{55}}%
\usepackage{multirow}

\theoremstyle{plain}
\newtheorem{theorem}{Theorem}[section]
\newtheorem{proposition}[theorem]{Proposition}
\newtheorem{lemma}[theorem]{Lemma}

\theoremstyle{definition}
\newtheorem{definition}[theorem]{Definition}
\newtheorem{assumption}[theorem]{Assumption}
\theoremstyle{remark}
\newtheorem{remark}[theorem]{Remark}

\newtheorem*{example}{Example}

\newcommand{\bz}{\mathbf 0}
\newcommand{\bP}{\mathbb P}

\newcommand{\R}{\mathbb R}

\DeclareMathOperator{\proj}{proj}

\DeclareMathOperator{\dist}{dist}

\DeclareMathOperator*{\argmin}{argmin}

\newcommand{\bx}{\bm{x}}

\newcommand{\bbeta}{\bm{\beta}}

\title{Beyond Maximum Likelihood: Variational Inequality Estimation for Generalized Linear Models}

\author{%
  Linglingzhi Zhu \\
  Georgia Tech \\
  \makebox[1.65in][c]{\texttt{llzzhu@gatech.edu}} \\
  \And
  Jonghyeok Lee \\
  Georgia Tech \\
  \makebox[1.65in][c]{\texttt{jlee4177@gatech.edu}} \\
  \And
  Yao Xie \\
  Georgia Tech \\
  \makebox[1.65in][c]{\texttt{yao.xie@isye.gatech.edu}}
}

\begin{document}

\maketitle

\begin{abstract}
Generalized linear models (GLMs) are fundamental tools for statistical modeling, with maximum likelihood estimation (MLE) serving as the classical approach for parameter inference. While MLE performs well for canonical GLMs, it can become computationally challenging in more general settings with non-canonical, non-smooth, or saturating link functions, where the resulting optimization landscape may be ill-conditioned, non-convex, or non-differentiable. In this paper, we study an alternative estimation framework based on variational inequalities (VIs), which formulates GLM estimation through an operator-based equilibrium condition rather than likelihood minimization. We analyze the VI estimator from a statistical perspective and establish finite-sample error bounds and asymptotic normality, together with convergence guarantees for fixed-point and stochastic approximation algorithms. The framework accommodates a broad class of link functions, including non-canonical, non-smooth, and non-monotone cases satisfying a strong Minty-type condition, and extends naturally to generalized additive models (GAMs) via basis expansion. Numerical experiments demonstrate that under a fixed iteration budget, the VI approach achieves competitive estimation accuracy and faster, more stable convergence than MLE, particularly in GLMs and GAMs with non-canonical or non-smooth link functions.
\end{abstract}

\section{Introduction}

Generalized linear models (GLMs) are a cornerstone of statistics and machine learning, providing a unified framework that extends linear regression to accommodate diverse response types through the exponential-family formulation \cite{nelder1972generalized,mccullagh1989generalized}. By linking the conditional mean of the response to a linear predictor via a specified link function, GLMs combine interpretability and flexibility, encompassing classical models such as linear, logistic, and Poisson regression as special cases. Owing to this balance, GLMs have been widely used across scientific, engineering, and data-driven decision-making domains.

The central inferential task in GLMs is to estimate the regression parameter characterizing the relationship between covariates and responses. Maximum likelihood estimation (MLE) has long served as the standard approach, offering strong asymptotic guarantees such as consistency and efficiency under correct model specification. However, in modern modeling scenarios, MLE can face substantial computational challenges. When the link function departs from the canonical form, for example by incorporating saturating or non-smooth structures designed for robustness and flexibility, the resulting likelihood-based objective may become ill-conditioned, non-convex, or non-differentiable. These issues often lead to slow convergence, sensitivity to initialization, and numerical instability, particularly in high-dimensional settings. Such difficulties motivate alternative formulations that retain the interpretability of GLMs while improving computational behavior and theoretical robustness.

In this paper, we study \emph{variational inequality (VI)-based estimation} as an operator-based alternative to MLE for GLMs. The VI framework, with roots in optimization and operator theory \cite{facchinei2003finite} and close ties on the statistical and algorithmic sides to the generalized method of moments \cite{hansen1982large} and the perceptron \cite{rosenblatt1958perceptron}, formulates estimation through an equilibrium condition rather than likelihood minimization, and has been applied to nonlinear least-squares and signal recovery problems \cite{juditsky2019signal}. This perspective naturally accommodates non-canonical and even non-monotone link functions, relaxing the smoothness and convexity requirements inherent in likelihood-based approaches. It also enables a unified treatment of statistical accuracy and algorithmic convergence within a single framework.

\paragraph{Contributions.}
Our analysis shows that the VI formulation recovers the MLE first-order optimality condition under canonical links and yields a novel moment-based estimator for non-canonical links. For the sample-average VI estimator, we establish finite-sample error bounds and asymptotic normality with a sandwich-type covariance structure under non-smooth link functions (Section~\ref{sec:estimator}), together with convergence guarantees for fixed-point and stochastic approximation algorithms (Section~\ref{sec:algorithm}). By operating under a strong Minty condition rather than strong monotonicity or likelihood convexity, our framework accommodates non-canonical, saturating, and non-monotone links and extends naturally to semiparametric generalized additive models (GAMs, Section~\ref{sec:gam}). Under correct specification, while converged MLE attains the Cram\'er--Rao lower bound when the global optimum is reachable, the VI estimator provides two key advantages in challenging non-canonical regimes: (i) convergence guarantees that MLE lacks when the likelihood is non-convex or has plateaus, and (ii) faster and more stable progress under a fixed first-order computational budget.

\section{Background and Problem Setup}\label{Background}

We briefly review the GLM framework and establish notation used throughout the paper. We then describe maximum likelihood estimation \eqref{Problem} and an alternative variational inequality formulation \eqref{problem_VI}, which forms the basis of our analysis. A summary of notation is provided in Appendix~\ref{sec:notation}.

Given a sample $(\bx,y)$, where $\bx\in\R^d$ denotes the covariates and $y\in\R$ the response, a GLM specifies the conditional distribution of $y\mid\bx$ using the exponential family. The conditional mean is linked to a linear predictor through a link function $g:\R\to\R$:
\begin{equation}\label{eq:glm_def}
\mathbb{E}[y\mid \bx]
=
g^{-1}\!\left(\beta_0+\sum_{j=1}^{d}\beta_j x_j\right),
\end{equation}
where $\bbeta=(\beta_0,\dots,\beta_d)\in\R^{d+1}$ is the coefficient vector including an intercept. Classical examples include logistic regression with a binomial distribution and logit link, and Poisson regression with a Poisson distribution and logarithmic link.

\paragraph{Maximum likelihood estimation.}
Given a conditional model with density (or mass) function $p(y\mid \bx;\bbeta)$ and i.i.d.\ observations $\{(\bx^i, y^i)\}_{i=1}^N$ (where $N$ denotes the sample size), the MLE minimizes the empirical negative log-likelihood (NLL),
\begin{equation}\tag{MLE}\label{Problem}
\min_{\bbeta\in\R^{d+1}} \mathcal{L}_N(\bbeta)
:=
\frac{1}{N}\sum_{i=1}^N
\ell\!\left(
g^{-1}\!\left(\beta_0+\sum_{j=1}^{d}\beta_j x_j^i\right),\, y^i
\right),
\end{equation}
where $\ell(u,y)$ denotes the per-observation NLL expressed as a function of the mean $u$.
Problem~\eqref{Problem} can be solved using second-order methods such as Fisher scoring, IRLS, or quasi-Newton methods (e.g., L-BFGS) \cite{green1984iteratively,mccullagh1989generalized}, as well as first-order or stochastic methods in large-scale settings \cite{nesterov2013introductory,kushner2003stochastic}.

When canonical link functions are used, \eqref{Problem} is convex and smooth, ensuring global optimality and efficient convergence. Some non-canonical links, such as the probit link, preserve this structure, but many links chosen in applications for physical or robustness reasons break it in one of three ways. First, saturating links create plateaus on which the NLL gradient vanishes. For example, photon-limited imaging constrains the Poisson intensity to a prescribed interval, which amounts to a clipped link~\cite{cao2015poisson}. Second, mixture and heavy-tailed links make the NLL non-convex, as in polychromatic X-ray tomography, where the measurements follow a mixture of exponential attenuation terms~\cite{lou2026accurate}, and in robust binary regression with robit~\cite{liu2004robit} or arctangent links. Third, links with linear rather than exponential growth make the NLL ill-conditioned: the softplus link, used in Poisson models of neural spike trains~\cite{mcfarland2013inferring}, yields an NLL whose curvature decays for large predictors, and the identity link in point-process models yields gradients that are not Lipschitz near zero intensity~\cite{he2020point}. Non-canonical links also appear in modern AI systems that fit GLMs on learned embeddings and model only the conditional mean, such as reward modeling from preference data, where the logistic link is not a necessary choice~\cite{sun2025rethinking}, and autoregressive sequence modeling~\cite{zhou2024nonlinear}.

\paragraph{Variational inequality formulation.}
We now formulate estimation as a variational inequality. Let $\mathbb{P}$ denote the joint distribution of $(\bx,y)$ 
and let $V:\R^{d+1}\to\R^{d+1}$ denote the population operator
\[
V(\bbeta)
:=
\mathbb{E}_{(\bx,y)\sim\mathbb{P}}
\!\left[
\left(
g^{-1}\!\left(\beta_0+\sum_{j=1}^{d}\beta_j x_j\right)-y
\right)\tilde{\bx}
\right],
\]
where $\tilde{\bx}:=[1;\bx]\in\R^{d+1}$. We define the target parameter $\bbeta^\star$ as the solution to the population variational inequality
\begin{equation}\label{problem_VI_popu}
\langle V(\bbeta^\star),\, \bbeta-\bbeta^\star\rangle \ge 0,
\qquad \forall\,\bbeta\in\R^{d+1}.
\end{equation}
If the univariate conditional mean relation~\eqref{eq:glm_def} holds at $\bbeta^\star$, the tower property of conditional expectation
directly implies the $(d+1)$-dimensional orthogonality condition
\[
V(\bbeta^\star) = \mathbb{E}_{(\bx,y)\sim\mathbb{P}}\!\left[\big(g^{-1}(\tilde{\bx}^\top\bbeta^\star)-y\big)\tilde{\bx}\right] = \mathbb{E}_{\bx}\!\left[\tilde{\bx}\,\mathbb{E}\!\left[g^{-1}(\tilde{\bx}^\top\bbeta^\star)-y \mid \bx\right]\right] = \bz.
\]
When a canonical link is used, \eqref{problem_VI_popu} coincides with the population first-order optimality condition of the MLE. For non-canonical links, \eqref{problem_VI_popu} defines a moment-based target parameter characterized solely by the conditional mean~\eqref{eq:glm_def}, removing the likelihood-specific derivative weight from the score equation. Given samples $\{(\bx^i,y^i)\}_{i=1}^N$, the empirical VI estimator $\hat{\bbeta}_N$ is defined as
\begin{equation}\tag{VI}\label{problem_VI}
\hat{\bbeta}_N \in
\left\{
\hat{\bbeta}\in\R^{d+1} :
\langle V_N(\hat{\bbeta}),\, \bbeta-\hat{\bbeta}\rangle \ge 0,
\ \forall\,\bbeta\in\R^{d+1}
\right\},
\end{equation}
with $V_N(\bbeta) := \frac{1}{N}\sum_{i=1}^N V_{(\bx^i,y^i)}(\bbeta)$, where $V_{(\bx, y)}(\bbeta):=(g^{-1}(\tilde{\bx}^{\top}\bbeta)-y)\tilde{\bx}$ denotes the per-sample vector field associated with a pair $(\bx,y)$. Note that $V(\bbeta)=\mathbb{E}_{(\bx,y)\sim\mathbb{P}}[V_{(\bx,y)}(\bbeta)]$.

\paragraph{Relation to prior work.}
VI-based estimation was introduced by Juditsky and Nemirovski \cite{juditsky2019signal}, building on perceptron and isotonic methods such as Isotron and GLMtron \cite{rosenblatt1958perceptron,kalai2009isotron,kakade2011efficient}, and was later extended to point processes, time series, and tomographic imaging \cite{juditsky2020convex,juditsky2023generalized,cheng2024point,zhou2024nonlinear,lou2026accurate}. Our work departs from this literature in both statistical scope and structural assumptions. First, whereas prior VI analyses focus primarily on algorithmic convergence, we study the  estimator $\hat{\bbeta}_N$ as a statistical object by establishing finite-sample bounds, asymptotic normality with a sandwich covariance under non-smooth links, and the explicit efficiency--stability trade-off relative to MLE. Second, while existing VI and matching-loss methods \cite{kalai2009isotron,kakade2011efficient,juditsky2019signal} rely on strong monotonicity, we operate under the strictly weaker \emph{strong Minty condition}. Requiring only an overall growth trend toward the true parameter rather than pairwise monotonicity, this relaxation extends statistical consistency and linear fixed-point convergence to accommodate saturating plateaus (e.g., clipped links) and non-monotone links. Finally, while $\hat{\bbeta}_N$ is a moment estimator \cite{hansen1982large}, our finite-sample bounds and convergence guarantees do not follow from classical  asymptotic generalized method of moments theory.

\section{From MLE to VI}
\label{sec:setup}

We consider the GLM defined in~\eqref{eq:glm_def}, where the response follows an exponential-family distribution with mean $g^{-1}(\tilde{\bx}^{\top}\bbeta)$.
The choice of the loss $\ell$ and inverse link $g^{-1}$ significantly impacts the optimization landscape.

\paragraph{Equivalence under canonical links.}
When the model employs a canonical link function (e.g., Logit for Logistic, Identity for Normal), the gradient of the NLL coincides exactly with the VI operator:
$\nabla \mathcal{L}_N(\bbeta) = V_N(\bbeta)$.
This implies that MLE and VI share identical first-order optimality conditions. 
Representative examples are summarized in Table~\ref{table:summary}.
However, for non-canonical links, this equivalence breaks, and MLE often suffers from optimization challenges. Detailed derivations are provided in Appendix~\ref{sec:setup_appendix}.

\begin{table}[h]
    \centering
    \caption{MLE and VI formulations under different links. 
    }
    \smallskip
    \footnotesize
    \label{table:summary}
    \renewcommand{\arraystretch}{1}
    \begin{tabular}{cccc} 
    \toprule
    \textbf{Model} & \textbf{Loss $\ell(u,y)$} & \textbf{Inverse Link $g^{-1}(z)$} & \textbf{Equivalent?} \\
    \midrule
    \multirow{2}{*}{Logistic} & \multirow{2}{*}{$-y \log u -(1-y)\log(1-u)$} & $e^z/(1+e^z)$ & \cmark \\
    & & $\frac{1}{2}+\frac{1}{\pi}\arctan(z)$  &  \xmark \\ 
    \midrule
    \multirow{2}{*}{Normal} & \multirow{2}{*}{$\frac{1}{2}(u-y)^2$} & $z$ & \cmark \\
    & & $\max\{0,z\}$  &  \xmark \\ 
    \midrule
    \multirow{2}{*}{Exponential} & \multirow{2}{*}{$\log u + y/u$} & $z^{-1}$ & \cmark \\
    & & $e^{-z}$  &  \xmark \\ 
    \midrule
    \multirow{3}{*}{Poisson} & \multirow{3}{*}{$-y\log u + u$} & $e^z$ & \cmark \\
    & & $\log(1+e^z)$  &  \xmark \\ 
    & & $\max\{c,\min\{e^z,C\}\}$  & \xmark \\
    \bottomrule
    \end{tabular}
\end{table}

\paragraph{Flatness of MLE landscape.}
It is well known that MLE optimization algorithms can converge slowly near the true parameter, especially in the non-canonical settings when the likelihood surface is nearly flat or ill-conditioned. Consider the softplus link $g^{-1}(z)=\log(1+e^z)$ in Poisson regression. In this case, the Hessian of the MLE objective decays rapidly in the right tail (as $z \to \infty$). In contrast, the Jacobian of the VI operator remains well-conditioned. This flatness in MLE leads to weaker local contraction and slower convergence compared to VI.

\paragraph{Non-convexity and non-smoothness.}
More severe issues arise with robust or truncated links. Consider the two-sided clipped exponential link
\[
g^{-1}(z)=\max\{c,\min\{e^z,C\}\}, \qquad 0<c<C,
\]
and the Poisson negative log-likelihood
\[
(\ell\circ g^{-1})(z) = -y\log(g^{-1}(z)) + g^{-1}(z).
\]
This composition is piecewise smooth with flat regions for $z \le \log c$ and $z \ge \log C$, and is non-differentiable at the two kink points. A direct inspection of the one-sided derivatives shows that the monotonicity conditions required for convexity cannot hold simultaneously when $0<c<C$. Hence, $\ell\circ g^{-1}$ is globally non-convex and non-smooth. 
For illustration, consider a single observation with $y=1$ and predictor $z=\beta$, which gives
\[
\mathcal{L}(\beta) = -\log g^{-1}(\beta) + g^{-1}(\beta).
\]
As shown in the first two panels of Figure~\ref{fig:VI_example}, clipping induces wide flat plateaus where gradient-based MLE updates can stall and fail to reach the global minimizer at $\beta=0$. 
In contrast, the third panel of Figure~\ref{fig:VI_example} shows that the corresponding VI vector field remains monotone, so simple fixed-point iterations can still converge to the global solution.

\begin{figure}[t]
    \centering
    \includegraphics[width=\linewidth]{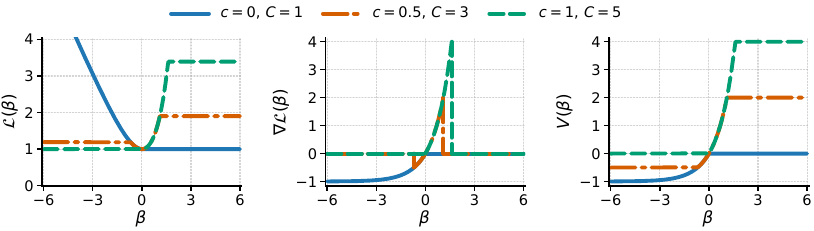}
    \caption{Clipped exponential link: MLE loss, its derivative, and the VI vector field.}
    \label{fig:VI_example}
\end{figure}

\section{Statistical Guarantees of VI Estimators}\label{sec:estimator}

Having established the formulation in~\eqref{problem_VI} and its motivation, we now turn to the theoretical analysis of the proposed estimator's statistical properties. 
This section derives finite-sample (Section~\ref{subsec:estimation_err}) and asymptotic (Section~\ref{subsec:asym}) guarantees for the empirical VI solution, defined as
\[
\hat{\bbeta}_{N}\in\operatorname{Sol}(V_N),
\]
i.e.,
$
V_N(\hat{\bbeta}_{N})
=\frac{1}{N}\sum_{i=1}^N \big(g^{-1}(\tilde{\bx}^{i\top} \hat{\bbeta}_{N})-y^i\big)\tilde{\bx}^i
=\bz.
$
We begin with the following assumptions, which are used throughout the paper:
\begin{assumption}\label{assump}
\leavevmode\\[-20pt]
\begin{enumerate}[label=\normalfont(\roman*)]
    \item The inverse link function $g^{-1}$ is $L$-Lipschitz continuous.
\item The true parameter $\bbeta^\star$ and $(\bx, y)\sim\mathbb{P}$ almost surely satisfy $|g^{-1}(\tilde{\bx}^{\top}\bbeta^\star)-y|\leq R$.
\item There exists $M\ge1$ such that $|x_j|\le M$
almost surely for all $j\in[d]$.
    \item The population operator $V$ and the empirical operator $V_N$ satisfy the strong Minty condition with modulus $\mu>0$ for any $\bbeta \in \mathbb{R}^{d+1}$:
    \[
    \langle V(\bbeta), \bbeta - \bbeta^\star \rangle \ge \mu \|\bbeta - \bbeta^\star\|^2,\quad 
    \langle V_N(\bbeta), \bbeta - \hat{\bbeta}_N \rangle \ge \mu \|\bbeta - \hat{\bbeta}_N\|^2. 
    \]
    Consequently, the solution sets $\operatorname{Sol}(V)$ and $\operatorname{Sol}(V_N)$ are nonempty and singleton.
\end{enumerate}
\end{assumption}

\begin{remark}[Strong Minty condition]
The strong Minty condition is a relatively mild regularity requirement compared to strong monotonicity. 
In Assumption~\ref{assump}, we adopt a slightly stronger version to ensure the uniqueness of both the true parameter $\bbeta^\star$ and the sample-averaged solution $\hat{\bbeta}_N$ for simplicity. 
Nevertheless, this assumption can be relaxed to the standard projection-based form in Definition~\ref{def:strong_minty}, 
by replacing $\bbeta^\star$ and $\hat{\bbeta}_N$ with $\proj_{\operatorname{Sol}(V)}(\bbeta)$ and $\proj_{\operatorname{Sol}(V_N)}(\bbeta)$, respectively. 
Intuitively, it requires that the vector field $V$ exhibits at least linear growth away from the solution, 
with a uniform positive coefficient that prevents flat or degenerate regions around the solution. 
Geometrically, in the one-dimensional case, 
this means that
$V(\beta)$ has the same sign as $\beta-\beta^\star$ and
satisfies $|V(\beta)|\geq\mu|\beta-\beta^\star|$,
thereby providing a linear error bound around the solution.
Further discussion and sufficient conditions are provided in Appendix~\ref{app:minty}; see also related treatments in~\cite{huang2023beyond}.
\end{remark}

\subsection{Finite-sample estimation error}\label{subsec:estimation_err}

We first derive the following estimation error bound for the variational inequality \eqref{problem_VI}.

\begin{theorem}[Estimation error]\label{thm:estimator} Suppose that Assumption \ref{assump} holds. Then for any $\epsilon \in(0,1)$, with probability at least $1-\epsilon$, we have the following estimation error  
\[
\|\hat{\bbeta}_N-\bbeta^{\star}\|
\leq  \frac{RM}{\mu}\sqrt{\frac{2(d+1)\ln(2(d+1)/\epsilon)}{N}}.
\]
\end{theorem}

\begin{remark}
For the MLE, a similar finite-sample error bound can be derived as 
$
\tilde{\mathcal{O}}(\sqrt{d}/(\mu' \sqrt{N})),
$
where $\mu'>0$ denotes the curvature modulus of the composite mapping 
$\ell\circ g^{-1}$ that satisfies the restricted secant inequality (RSI), i.e., for the maximum likelihood estimator $\bar{\bbeta}_N$ and any $\bbeta$
\[
\langle \nabla\mathcal{L}_N(\bbeta), \bbeta - \bar{\bbeta}_N \rangle \ge \mu' \|\bbeta - \bar{\bbeta}_N\|^2, 
\]
a condition stronger than the Polyak--\L{}ojasiewicz inequality and weaker than the strong convexity; see \cite{karimi2016linear}. Recall that the gradient of the empirical loss of \eqref{Problem} is given by
\begin{align*}
\nabla \mathcal{L}_N(\bbeta)
=\ &\mathbb{E}_{(\bx,y)\sim\mathbb{P}_N}
[
(g^{-1})'(\tilde{\bx}^{\top}\bbeta)
\ell'(
g^{-1}(\tilde{\bx}^{\top}\bbeta),y)\tilde{\bx}
].
\end{align*}
If the RSI holds with constant $\mu'$, then
$
\mu'\|\bar{\bbeta}_N-\bbeta^{\star}\|^2
\le
\langle\nabla\mathcal{L}_N(\bbeta^{\star}),\bbeta^{\star}-\bar{\bbeta}_N\rangle$,
which implies the bound
$
\|\bar{\bbeta}_N-\bbeta^{\star}\|
\le
\|\nabla\mathcal{L}_N(\bbeta^{\star})\|/\mu'.
$
Compared with the VI estimator in Theorem \ref{thm:estimator}, where by Lemma~\ref{lem:strongminty} the strong Minty condition reduces to a condition on $g^{-1}$ itself,
the MLE requires the composite function $z\mapsto(g^{-1})'(z)\ell'(g^{-1}(z),y)$ 
to satisfy an analogous Minty-type condition.
This composite structure makes the regularity requirement for MLE substantially more intricate.
\end{remark}

\subsection{Asymptotic normality of VI estimators}\label{subsec:asym}

We now proceed to establish the asymptotic normality of the VI estimator.
Our goal is to show that $\sqrt{N}(\hat\bbeta_N - \bbeta^\star)$ converges in distribution to a normal random vector with mean zero and a sandwich-type covariance matrix. 
To handle potential nonsmoothness of the link function $g^{-1}$, we employ tools from empirical process theory and a local linearization argument around the population solution $\bbeta^\star$. Since $g^{-1}$ is globally Lipschitz by Assumption~\ref{assump}, 
Rademacher's theorem \citep[Theorem 9.60]{rockafellar2009variational} implies that it is differentiable almost everywhere on $\R$, 
and hence its nondifferentiability set 
\[
\mathcal K:=\{t:(g^{-1})'(t)\text{ does not exist}\}
\]
has Lebesgue measure zero. We further assume the following condition.

\begin{assumption}\label{assump_asym}
The random variable $\langle\tilde{\bx},\bbeta^\star\rangle$ does not fall on the nondifferentiability set $\mathcal K$ of $g^{-1}$, that is,
\[
\mathbb P\big(\langle\tilde{\bx},\bbeta^\star\rangle\in \mathcal K\big)=0.
\]
\end{assumption}

This assumption is not restrictive in practice, since $\mathcal{K}$ has Lebesgue measure zero 
and thus the event $\langle \tilde{\bx},\bbeta^\star\rangle \in \mathcal{K}$ occurs with probability zero, for instance, whenever $\bx$ has a density and $\bbeta^\star\neq\bz$.
With this mild condition, the population VI operator $V$ is guaranteed to be differentiable at the true parameter 
$\bbeta^\star$.

\begin{lemma}\label{lem:V_diff}
Suppose that Assumptions~\ref{assump} and~\ref{assump_asym} hold. 
Then the function $V$ is Fr\'echet differentiable at $\bbeta^\star$ with Jacobian
\[
\nabla V(\bbeta^\star)
=\mathbb E_{(\bx,y)\sim\mathbb{P}}\left[(g^{-1})'(\langle\tilde{\bx},\bbeta^\star\rangle)\,\tilde{\bx}\tilde{\bx}^\top\right].
\]
\end{lemma}

Owing to the averaging-induced smoothing effect, we obtain the asymptotic normality result.

\begin{theorem}[Asymptotic normality of the VI estimator]\label{thm:VI-CLT}
Suppose Assumptions~\ref{assump} and \ref{assump_asym} hold, and in addition
the covariance matrix 
\begin{align*}
\Gamma
:\!&=
\mathbb{E}_{(\bx,y)\sim\mathbb{P}}[V_{(\bx,y)}(\bbeta^\star)V_{(\bx,y)}(\bbeta^\star)^\top] = \mathbb{E}_{(\bx,y)\sim\mathbb{P}}[(y - g^{-1}(\tilde{\bx}^\top\bbeta^\star))^2 \tilde{\bx}\tilde{\bx}^\top] 
\end{align*}
is finite and positive definite.
Then $\hat\bbeta_N$ is $\sqrt N$-consistent and
\[
\sqrt N(\hat\bbeta_N-\bbeta^\star)\ \xrightarrow{d}\ \mathcal N(\bz,\nabla V(\bbeta^{\star})^{-1}\Gamma \nabla V(\bbeta^{\star})^{-\top}).
\]
\end{theorem}

The resulting asymptotic covariance exhibits a sandwich structure analogous to that of the MLE. Under correct likelihood specification and standard regularity conditions, MLE is asymptotically efficient, whereas the VI estimator may incur an efficiency loss. We provide a detailed efficiency comparison in Appendix~\ref{app:eff}. 
This potential loss can be viewed as a trade-off for superior computational performance with global convergence guarantees.

\section{Convergence of VI-Based Algorithms}\label{sec:algorithm}

This section investigates the computational efficiency of solving the empirical VI problem using iterative algorithms. We begin by analyzing the convergence behavior of a fixed-point iterative method under a fixed sample size $N$.
We then extend the analysis to the stochastic approximation setting, where the estimation is performed using a streaming data scheme.

For the following deterministic fixed-point iterative scheme under a fixed sample size, we establish its linear convergence rate under the strong Minty condition.

\begin{algorithm}[H]
  \caption{Fixed-Point Iterative Method}
  \label{alg-det}
  \begin{algorithmic}
    \STATE {\bfseries Input:} initialization $\bbeta^{0}$, step size $\eta>0$, data distribution $(\bx,y)\sim \bP_{N}$ 
    \FOR{$t=0$ {\bfseries to} $T-1$}
      \STATE $\bbeta^{t+1} \gets \bbeta^{t} - \eta\, V_N(\bbeta^{t})$
    \ENDFOR
  \end{algorithmic}
\end{algorithm}

\begin{theorem}[Linear convergence of fixed-point method]\label{thm:determine}
Suppose that Assumption \ref{assump} holds and set the step size
$
\eta=\frac{\mu}{L^2(1+dM^2)^2},
$
where $\mu>0$ is the strong Minty modulus in Assumption~\ref{assump}. Then the sequence of estimates $\bbeta^{t}$ given by the Algorithm \ref{alg-det} satisfies for every $t=0,1, 2,\ldots,$ 
$$
\|\bbeta^t-\hat{\bbeta}_N\|^2\leq \left(1-\frac{\mu^2}{L^2(1+dM^2)^2}\right)^{t}\cdot \|\bbeta^{0}-\hat{\bbeta}_N\|^2.
$$
Furthermore, for any $\epsilon \in(0,1)$, with probability at least $1-\epsilon$, we have the following estimation error  
\begin{align*}
\|\bbeta^t-\bbeta^{\star}\|
\leq\ & \left(1-\frac{\mu^2}{L^2(1+dM^2)^2}\right)^{t/2}\cdot \|\bbeta^{0}-\hat{\bbeta}_N\|+ \frac{RM}{\mu}\sqrt{\frac{2(d+1)\ln(2(d+1)/\epsilon)}{N}}.
\end{align*}
\end{theorem}

Next, we consider the stochastic approximation setting, in which the estimator is updated from streaming data.

\begin{algorithm}[H]
  \caption{Stochastic Approximation}
  \label{alg-sgd}
  \begin{algorithmic}
    \STATE {\bfseries Input:} initialization $\bbeta^{0}$, step sizes $\{\eta^{t}\}_{t=0}^{T-1}$ with $\eta^{t}>0$, streaming data $(\bx^{t},y^{t})\sim \bP$
    \FOR{$t=0$ {\bfseries to} $T-1$}
      \STATE $\bbeta^{t+1} \gets \bbeta^{t} - \eta^{t}\, V_{(\bx^{t},y^{t})}(\bbeta^{t})$
    \ENDFOR
  \end{algorithmic}
\end{algorithm}

\begin{theorem}[Sublinear estimation error of stochastic approximation]
\label{thm:sgd}
Suppose that Assumption~\ref{assump} holds. Then for any deterministic initialization $\bbeta^0$, the sequence $\{\bbeta^t\}_{t\geq0}$ generated by Algorithm~\ref{alg-sgd} from i.i.d. streaming data with step size
$
\eta^t=\frac{1}{\mu(t+1)},
$
where $\mu>0$ is the strong Minty modulus in Assumption~\ref{assump}, satisfies
\[
\mathbb{E}\|\bbeta^t-\bbeta^\star\|^2
\leq
\frac{c_0(d+1)R^2M^2}{\mu^2(t+1)},
\qquad t=0,1,2,\ldots,
\]
for some constant $c_0>0$ independent of $t$.
\end{theorem}

\section{Extension to Generalized Additive Models}\label{sec:gam}

Generalized additive models (GAMs) \cite{hastie1986generalized} extend GLMs by replacing the linear predictor
with an additive (component-wise) structure, while retaining the exponential-family modeling of the response.
A GAM links the conditional mean to the covariates via
\begin{equation*}
\mathbb{E}[y\mid \bx] = g^{-1}\!\left( \beta_0 + \sum_{j=1}^d f_j(x_j) \right),
\end{equation*}
where each additive component
$f_j:\mathbb{R}\to\mathbb{R}$ describes the nonlinear effect
of the $j$-th feature.
To remove the ambiguity between the intercept and the additive
constants, we impose an identifying constraint on each component,
such as $f_j(0)=0$ or centering with respect to a fixed reference
measure.

\paragraph{Basis representation.}
To render the estimation problem finite-dimensional, we approximate
each component function $f_j$ using a truncated basis expansion.
Let $\{\psi_k\}_{k=1}^K$ be a set of basis functions
(e.g., B-splines or polynomials) satisfying the chosen identifying
constraint.
For example, the Legendre basis functions
$\psi_k=P_k$, $k=1,\ldots,K$, satisfy
$\int_{-1}^{1}\psi_k(u)\,du=0$.
We parameterize
\[
f_j(x_j)\approx\sum_{k=1}^K\beta_{jk}\psi_k(x_j),
\quad j=1,\ldots,d.
\]
Define the stacked feature mapping $\Phi:\mathbb{R}^d\to\mathbb{R}^{dK}$ by
\[
\Phi(\bx) := [\psi_1(x_1),\ldots,\psi_K(x_1),\ \ldots,\ \psi_1(x_d),\ldots,\psi_K(x_d)]^\top,
\]
and
$\tilde{\Phi}(\bx):=[1;\Phi(\bx)]\in\mathbb{R}^{dK+1}$. Let $\bbeta\in\mathbb{R}^{dK+1}$ denote the augmented parameter vector collecting the intercept and all basis coefficients,
i.e., $\bbeta=(\beta_0,\{\beta_{jk}\}_{j\in[d],k\in[K]})$. Under this formulation, the additive predictor becomes linear in the
transformed feature space:
\[
\beta_0+\sum_{j=1}^d\sum_{k=1}^K \beta_{jk}\psi_k(x_j)=\langle \bbeta,\tilde{\Phi}(\bx)\rangle,
\]
and hence the basis-expanded GAM reduces to a GLM problem in $\mathbb{R}^{dK+1}$.

The choice of basis functions impacts the theoretical properties of the estimator. We require the basis functions to be orthogonal (or at least linearly independent) with respect to the data distribution. 
If the basis functions are linearly dependent with respect to the data distribution, the second-moment matrix $\mathbb{E}[\tilde{\Phi}(\bx)\tilde{\Phi}(\bx)^\top]$ is singular, $V$ vanishes on a whole line through $\bbeta^\star$, and the strong Minty condition in Assumption~\ref{assump}(iv) fails for every $\mu>0$, so the uniqueness and stability guarantees of the VI estimator are lost.

\paragraph{VI estimator for GAM.}
Following the framework established for GLMs, we define the VI estimator for GAMs by substituting the original covariates
with the basis features. The empirical vector field $V_N:\mathbb{R}^{dK+1}\to\mathbb{R}^{dK+1}$ is
\begin{equation*}
V_N(\bbeta)
:= \mathbb{E}_{(\bx,y)\sim\mathbb{P}_N}
\left[\left( g^{-1}\!\big(\langle \bbeta,\tilde{\Phi}(\bx)\rangle\big) - y \right)\cdot \tilde{\Phi}(\bx)\right].
\end{equation*}
The VI estimator $\hat{\bbeta}_N$ is then defined as the solution of $V_N(\bbeta)=\bz$, which is unique under Assumption~\ref{assump}(iv).
This formulation allows us to extend the theoretical guarantees of the VI estimator to the GAM setting.
In particular, when the basis size $K$ is fixed, the results in Sections~\ref{sec:estimator} and~\ref{sec:algorithm} carry over with $d+1$ replaced by $dK+1$ and with the constants $M$ and $\mu$ in Assumption~\ref{assump} referring to the transformed features $\tilde{\Phi}(\bx)$.

\section{Numerical Experiments}
\label{sec:exp}

We empirically evaluate the VI estimator from two complementary perspectives. First, we use controlled synthetic Poisson GLMs to isolate the optimization phenomena predicted by our analysis, such as flatness, truncation, and non-convexity induced by non-canonical links. Second, we evaluate the method on a real count-prediction benchmark to examine whether these stability gains can translate into improved predictive performance in practical GLM/GAM settings. Additional experimental details and supplementary results are provided in Appendix~\ref{app_exp}. Code and instructions for reproducing the experiments are available at \url{https://github.com/jonghyeoki/vi-glm-estimation}.

\subsection{Synthetic Poisson GLMs}\label{subsec:synthetic}

We begin with synthetic Poisson regression experiments under several link functions.
For each experiment, the true parameter is fixed as $\bbeta^\star = d^{-1/2}(1,\dots,1)\in\R^d$, and the covariates $\bx^i\in\R^d$, $i=1,\dots,N$, are drawn i.i.d.\ from $\mathcal{N}(0,I_d)$. 
Given $\bx^i$, the response is generated according to
$
y^i \sim \mathrm{Poisson}(g^{-1}({\bbeta^\star}^\top \bx^i))
$,
with no intercept term. 
We consider the following inverse link functions: 
(i) the log link $g^{-1}(z)=e^z$; 
(ii) the softplus link $g^{-1}(z)=\log(1+e^z)$; 
(iii) the clipped exponential link $g^{-1}(z)=\min\{e^z,2\}$; and 
(iv) a scaled Gaussian-mixture CDF link
\(
g^{-1}(z)=4.5\Phi\!\left(\frac{z+1.0}{0.7}\right)+5.5\Phi\!\left(\tfrac{z-1.2}{0.5}\right),
\)
where $\Phi$ denotes
the standard Gaussian CDF. 
The log link corresponds to the canonical Poisson model, for which the VI and MLE updates coincide. 
The softplus and clipped exponential links introduce bounded curvature and truncation in the mean function, respectively, while the Gaussian-mixture CDF provides a non-convex example designed to test robustness in more challenging settings.

For each link, we compute both estimators by first-order iterative updates
\[
    \bbeta^{k+1} = \bbeta^k - \eta^k V_N(\bbeta^k),\quad \bbeta^{k+1} = \bbeta^k - \eta^k \nabla \mathcal{L}_N(\bbeta^k),\]
using the same exponentially decaying step sizes. We vary the dimension $d \in \{10, 20, 50, 100\}$ and sample size $N \in \{100, 200, 500, 1000\}$. 

Table~\ref{tab:poisson_table_softplus} reports the average squared error $\|\bbeta^k - \bbeta^\star\|^2$ after $k \in \{20, 50, 100, 200\}$ iterations for the softplus link. The VI estimator consistently attains lower error than MLE across iteration budgets, often by a larger margin at the early stage. A more comprehensive set of results with other link functions and a sparse parameter setting can be found in Appendix~\ref{app:GLM_exp}, where we observe that VI generally dominates MLE, except for the Gaussian-mixture CDF link in the ``easy'' regime of large sample size, low dimension, and large iterations.

\begin{table}[t]
\centering
\caption{Mean squared error of the VI estimator and MLE with softplus link across iterations $k$. For each $(k, d, N)$ combination, the smaller error between the two estimators is highlighted in bold. The values in parentheses are standard deviations across 1000 independent repetitions.}
\smallskip
\scriptsize
\setlength{\tabcolsep}{3.5pt} %
\renewcommand{\arraystretch}{0.975}
\begin{tabular}{cccccccccc}
\toprule
\multirow{2}{*}{$k$} & \multirow{2}{*}{$d$} &
\multicolumn{2}{c}{$N=100$} & \multicolumn{2}{c}{$N=200$} &
\multicolumn{2}{c}{$N=500$} & \multicolumn{2}{c}{$N=1000$} \\
 &  & VI & MLE & VI & MLE & VI & MLE & VI & MLE \\
\midrule
\multirow{4}{*}{\makecell[c]{20}} &
10  & \textbf{.63} {\scriptsize (.08)} & .71 {\scriptsize (.06)} &
      \textbf{.51} {\scriptsize (.07)} & .62 {\scriptsize (.06)} &
      \textbf{.34} {\scriptsize (.06)} & .47 {\scriptsize (.05)} &
      \textbf{.21} {\scriptsize (.04)} & .34 {\scriptsize (.03)} \\
& 20  & \textbf{.73} {\scriptsize (.07)} & .79 {\scriptsize (.05)} &
       \textbf{.63} {\scriptsize (.06)} & .71 {\scriptsize (.05)} &
       \textbf{.47} {\scriptsize (.05)} & .58 {\scriptsize (.04)} &
       \textbf{.33} {\scriptsize (.04)} & .46 {\scriptsize (.03)} \\
& 50  & \textbf{.82} {\scriptsize (.05)} & .87 {\scriptsize (.03)} &
       \textbf{.75} {\scriptsize (.04)} & .81 {\scriptsize (.03)} &
       \textbf{.63} {\scriptsize (.04)} & .71 {\scriptsize (.03)} &
       \textbf{.51} {\scriptsize (.03)} & .62 {\scriptsize (.03)} \\
& 100 & \textbf{.88} {\scriptsize (.03)} & .91 {\scriptsize (.02)} &
       \textbf{.83} {\scriptsize (.03)} & .87 {\scriptsize (.02)} &
       \textbf{.73} {\scriptsize (.03)} & .79 {\scriptsize (.02)} &
       \textbf{.63} {\scriptsize (.03)} & .71 {\scriptsize (.02)} \\
\midrule
\multirow{4}{*}{\makecell[c]{50}} &
10  & \textbf{.47} {\scriptsize (.10)} & .57 {\scriptsize (.09)} &
      \textbf{.32} {\scriptsize (.08)} & .44 {\scriptsize (.07)} &
      \textbf{.16} {\scriptsize (.05)} & .27 {\scriptsize (.05)} &
      \textbf{.08} {\scriptsize (.03)} & .16 {\scriptsize (.03)} \\
& 20  & \textbf{.59} {\scriptsize (.09)} & .67 {\scriptsize (.07)} &
       \textbf{.46} {\scriptsize (.07)} & .57 {\scriptsize (.06)} &
       \textbf{.28} {\scriptsize (.05)} & .40 {\scriptsize (.05)} &
       \textbf{.16} {\scriptsize (.03)} & .27 {\scriptsize (.03)} \\
& 50  & \textbf{.75} {\scriptsize (.06)} & .80 {\scriptsize (.05)} &
       \textbf{.64} {\scriptsize (.06)} & .71 {\scriptsize (.05)} &
       \textbf{.46} {\scriptsize (.05)} & .57 {\scriptsize (.04)} &
       \textbf{.32} {\scriptsize (.04)} & .44 {\scriptsize (.03)} \\
& 100 & \textbf{.83} {\scriptsize (.05)} & .86 {\scriptsize (.04)} &
       \textbf{.74} {\scriptsize (.05)} & .79 {\scriptsize (.04)} &
       \textbf{.59} {\scriptsize (.04)} & .68 {\scriptsize (.03)} &
       \textbf{.46} {\scriptsize (.03)} & .56 {\scriptsize (.03)} \\
\midrule
\multirow{4}{*}{\makecell[c]{100}} &
10  & \textbf{.40} {\scriptsize (.11)} & .50 {\scriptsize (.10)} &
      \textbf{.26} {\scriptsize (.08)} & .36 {\scriptsize (.08)} &
      \textbf{.12} {\scriptsize (.04)} & .20 {\scriptsize (.04)} &
      \textbf{.05} {\scriptsize (.02)} & .11 {\scriptsize (.03)} \\
& 20  & \textbf{.54} {\scriptsize (.10)} & .62 {\scriptsize (.08)} &
       \textbf{.40} {\scriptsize (.08)} & .50 {\scriptsize (.07)} &
       \textbf{.22} {\scriptsize (.05)} & .32 {\scriptsize (.05)} &
       \textbf{.12} {\scriptsize (.03)} & .20 {\scriptsize (.03)} \\
& 50  & \textbf{.71} {\scriptsize (.07)} & .76 {\scriptsize (.06)} &
       \textbf{.59} {\scriptsize (.06)} & .66 {\scriptsize (.05)} &
       \textbf{.40} {\scriptsize (.05)} & .50 {\scriptsize (.04)} &
       \textbf{.26} {\scriptsize (.04)} & .36 {\scriptsize (.03)} \\
& 100 & \textbf{.82} {\scriptsize (.05)} & .84 {\scriptsize (.04)} &
       \textbf{.72} {\scriptsize (.05)} & .76 {\scriptsize (.04)} &
       \textbf{.54} {\scriptsize (.05)} & .63 {\scriptsize (.04)} &
       \textbf{.40} {\scriptsize (.04)} & .50 {\scriptsize (.03)} \\
\midrule
\multirow{4}{*}{\makecell[c]{200}} &
10  & \textbf{.38} {\scriptsize (.12)} & .48 {\scriptsize (.10)} &
      \textbf{.24} {\scriptsize (.08)} & .34 {\scriptsize (.08)} &
      \textbf{.10} {\scriptsize (.04)} & .18 {\scriptsize (.04)} &
      \textbf{.05} {\scriptsize (.02)} & .09 {\scriptsize (.02)} \\
& 20  & \textbf{.53} {\scriptsize (.11)} & .61 {\scriptsize (.09)} &
       \textbf{.38} {\scriptsize (.08)} & .47 {\scriptsize (.07)} &
       \textbf{.20} {\scriptsize (.04)} & .30 {\scriptsize (.04)} &
       \textbf{.10} {\scriptsize (.03)} & .18 {\scriptsize (.03)} \\
& 50  & \textbf{.71} {\scriptsize (.07)} & .75 {\scriptsize (.06)} &
       \textbf{.57} {\scriptsize (.07)} & .64 {\scriptsize (.06)} &
       \textbf{.38} {\scriptsize (.05)} & .48 {\scriptsize (.04)} &
       \textbf{.24} {\scriptsize (.04)} & .34 {\scriptsize (.03)} \\
& 100 & \textbf{.82} {\scriptsize (.06)} & .84 {\scriptsize (.04)} &
       \textbf{.70} {\scriptsize (.06)} & .75 {\scriptsize (.04)} &
       \textbf{.52} {\scriptsize (.05)} & .61 {\scriptsize (.04)} &
       \textbf{.38} {\scriptsize (.04)} & .48 {\scriptsize (.03)} \\
\bottomrule
\end{tabular}
\label{tab:poisson_table_softplus}

\end{table}

Figure~\ref{fig:trajectories} and Figure~\ref{fig:Kcurve} illustrate representative convergence behavior for $(d,N)= (20, 400)$. Figure~\ref{fig:trajectories} shows the sample convergence trajectories of the VI and MLE iterates for all four link functions. The trajectories for the log link coincide exactly, as expected from their equivalence as in Appendix \ref{app:examples}. For non-canonical links, the VI update exhibits noticeably faster convergence in the early iteration phase. In the clipped exponential case, the VI and MLE trajectories coincide during the initial iterations because the iterates have not yet reached the truncation boundary, and the model effectively behaves like the log link in that region. For the Gaussian-mixture CDF link, which induces a highly non-convex regime, VI avoids the oscillations and occasional divergence observed in MLE.

Figure~\ref{fig:Kcurve} reports the average squared error against the iteration count $k$. The gap between VI and MLE is particularly pronounced for the softplus link, where VI converges significantly faster. For the clipped exponential link, the gap becomes more pronounced in the later iterations as the iterations approach the truncation boundary. For the Gaussian-mixture CDF link, the gap is largest during the early iterations and gradually narrows as both methods stabilize.

Overall, under the reported first-order iteration budgets, the results demonstrate that the proposed VI framework achieves comparable estimation accuracy to MLE under finite sample sizes, while often exhibiting faster convergence and greater numerical stability for non-canonical link functions. With canonical links, both methods coincide, but in more general GLMs, the VI formulation effectively mitigates the flat-landscape and conditioning issues that slow gradient-based likelihood estimation. 
Appendix~\ref{app:converged} also compares numerically converged softplus fits: at convergence, MLE attains slightly lower parameter error than VI, as expected from its asymptotic efficiency under correct specification (Appendix~\ref{app:eff}).

\begin{figure}[!t]
    \centering
    \begin{subfigure}[t]{0.22\linewidth}
        \centering
        \includegraphics[width=\linewidth]{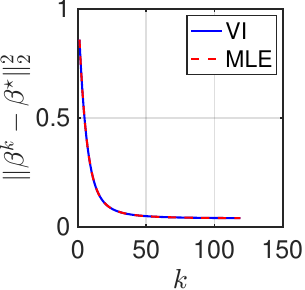}
        \caption{log}
    \end{subfigure}\hfill
    \begin{subfigure}[t]{0.22\linewidth}
        \centering
        \includegraphics[width=\linewidth]{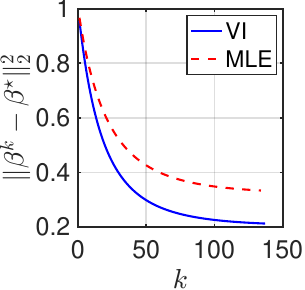}
        \caption{softplus}
    \end{subfigure}\hfill
    \begin{subfigure}[t]{0.22\linewidth}
        \centering
        \includegraphics[width=\linewidth]{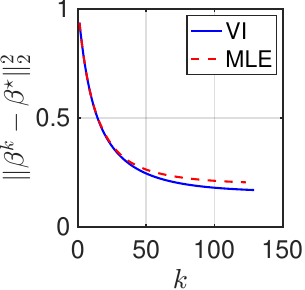}
        \caption{clipped exp.}
    \end{subfigure}\hfill
    \begin{subfigure}[t]{0.22\linewidth}
        \centering
        \includegraphics[width=\linewidth]{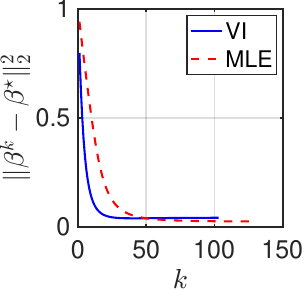}
        \caption{GMM CDF}
    \end{subfigure}
    \caption{Convergence trajectories of VI and MLE for Poisson regression with different link functions ($d=20, N=400$).}
    \label{fig:trajectories}
\end{figure}

\begin{figure}[!t]
    \centering
    \begin{subfigure}[t]{0.31\linewidth}
        \centering
        \includegraphics[width=\linewidth]{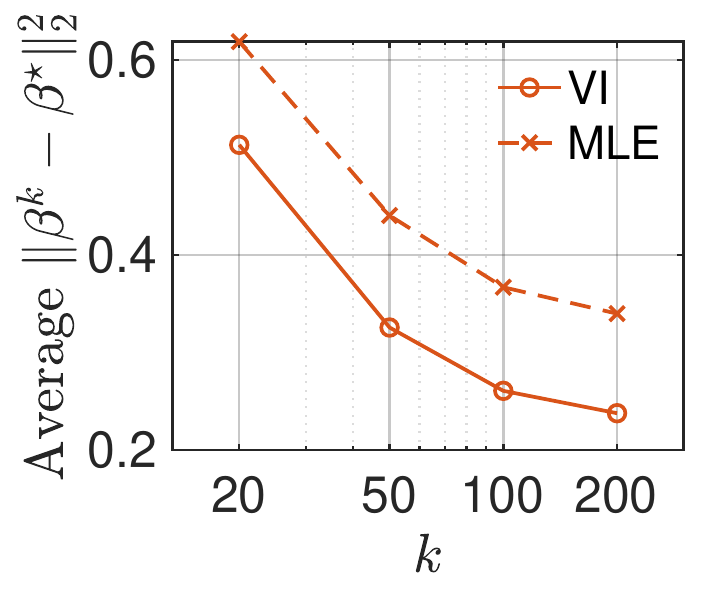}
        \caption{softplus}
    \end{subfigure}\hfill
    \begin{subfigure}[t]{0.31\linewidth}
        \centering
        \includegraphics[width=\linewidth]{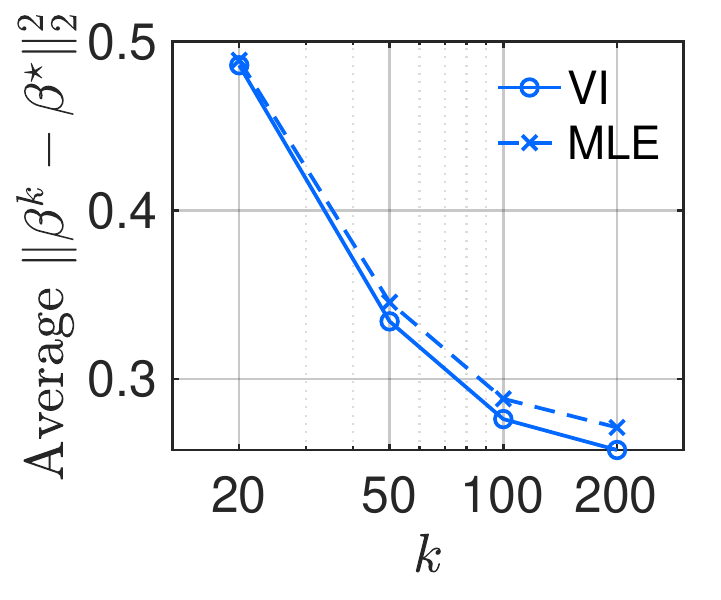}
        \caption{clipped exp.}
    \end{subfigure}\hfill
    \begin{subfigure}[t]{0.31\linewidth}
        \centering
        \includegraphics[width=\linewidth]{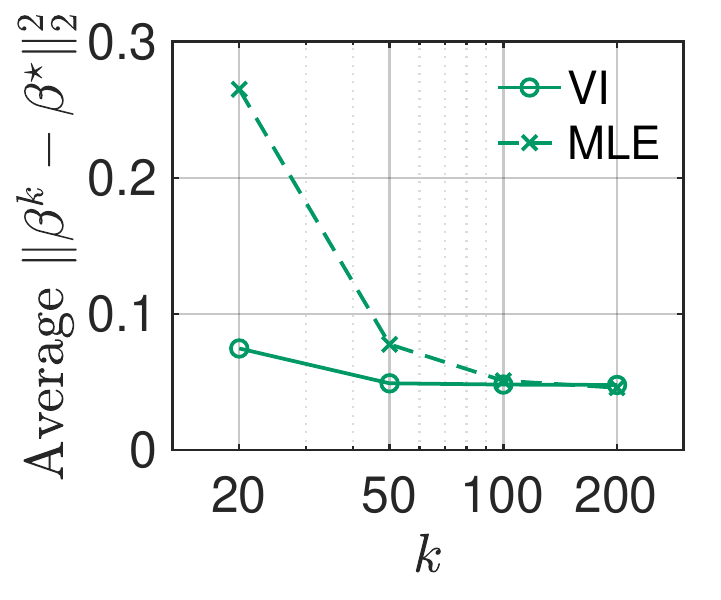}
        \caption{GMM CDF}
    \end{subfigure}
    \caption{Average squared error for VI and MLE against iteration budgets for Poisson regression with different link functions under 500 replications ($d=20, N=400$).}
    \label{fig:Kcurve}
\end{figure}

\subsection{Synthetic GAMs}\label{subsec:gam_exp}

We examine the reconstruction of a GAM using the VI and MLE
estimators in a toy experiment.
Consider a GAM of the form
$
\mathbb{E}[y\mid\bx]
=
g^{-1}(\sum_{j=1}^d f_j(x_j))$.
Each additive component $f_j$ is expanded using a truncated
Legendre basis with the constant basis function $P_0$ omitted:
$
f_j(x_j)=\sum_{k=1}^K\beta_{jk}P_k(x_j)$,
where $P_k$ denotes the $k$-th standard Legendre polynomial.
With this basis expansion, the GAM reduces to a finite-dimensional
GLM, to which we apply both the VI and MLE updates.
Detailed experimental settings and results are deferred to
Appendix~\ref{app:GAM_exp}. Here we briefly summarize the main observations.
For non-canonical links, VI typically exhibits more stable optimization behavior and slightly improved reconstruction of the underlying additive components  under the tested first-order iteration budgets.
These results demonstrate that the VI estimator naturally extends to GAMs and retains its key advantages beyond linear predictors.

\subsection{Real data: Bike-sharing count prediction}\label{subsec:realdata}

We next evaluate whether the stability behavior observed in the synthetic experiments appears in a practical count-prediction task. We use the UCI hourly bike-sharing dataset \citep{bike_sharing_275}, where the goal is to predict hourly rental counts from calendar and weather covariates. We use a chronological 3:1:1 train/validation/test split and compare GLM/GAM feature maps under the log, softplus, and clipped exponential links. For the GAM, continuous covariates are standardized using the training split and expanded using degree-3 Legendre bases, while categorical and time variables are one-hot encoded. For both VI and MLE, the initial step size is tuned by validation Poisson deviance.

The detailed results including runtime curves, calibration plots, and test metrics are provided in Appendix~\ref{app:realdata}. 
The main pattern is consistent with the synthetic results, that under the canonical log link, VI and MLE are essentially indistinguishable. Under the softplus link, VI obtains substantially better held-out performance and reaches a lower validation deviance at comparable runtime. The test-set calibration plot for the GAM design also shows a more stable predicted-versus-observed relationship for VI in the softplus case. The clipped exponential link gives a more mixed comparison that VI improves over MLE in the GLM design, but MLE is slightly better in the GAM design. These results again support that VI can provide substantial stability and predictive gains in harder non-canonical regimes under this fixed iteration budget. Appendix~\ref{app:converged} also compares converged softplus fits. Converged MLE substantially narrows the fixed-budget gap, but converged VI retains lower test deviance.

\section{Discussion}
\label{sec:discussion}

This work shows that the behavior of the VI estimator is governed by the strong Minty modulus and the Lipschitz constant. These quantities determine convergence, stability, and statistical error bounds, suggesting a guideline for link-function design: with other quantities fixed, larger Minty moduli and smaller Lipschitz constants yield tighter guarantees. Compared to MLE, whose optimization depends on the composite loss $\ell\circ g^{-1}$, the VI framework offers greater flexibility in challenging regimes. In practice, VI can be particularly useful when likelihood objectives exhibit flat or non-smooth regions, while canonical models remain well served by MLE. The VI formulation decouples estimation from likelihood curvature, relying on the conditional mean rather than a fully specified likelihood. Nevertheless, under correct likelihood specification and standard regularity conditions, VI may incur an asymptotic efficiency loss relative to MLE (Appendix~\ref{app:eff}). 
Future work includes extending the framework to regularized
high-dimensional GLMs through variational inclusions and proximal
updates, as well as analyzing penalized VI estimators under
conditional-mean misspecification.

\newpage

\begin{ack}
This work is partially supported by NSF grants DMS-2220495, CNS-2220387, and DMS-2134037, and the Coca-Cola Foundation. The authors declare no competing interests.
\end{ack}

\bibliographystyle{plainnat}
\bibliography{references}

\clearpage
\appendix

\section{Notation}\label{sec:notation}
The notation in the paper is standard. For a positive integer $d$, we write $[d]:=\{1,\ldots,d\}$. Let $\mathbb{R}^d$ denote the $d$-dimensional Euclidean space equipped with the inner product 
$\langle \bm{u}, \bm{v}\rangle := \bm{u}^{\top}\bm{v}$ and the induced norm 
$\|\bm{u}\| := \sqrt{\langle \bm{u}, \bm{u}\rangle}$. 
We write $\|\cdot\|_{\infty}$ for the infinity norm. For a closed convex set $\mathcal{B} \subseteq \mathbb{R}^{d}$, the Euclidean projection of a point 
$\bbeta$ onto $\mathcal{B}$ is defined as
$
\proj_{\mathcal{B}}(\bbeta) := \argmin_{\bbeta' \in \mathcal{B}} \|\bbeta' - \bbeta\|$,
and the corresponding distance is given by
$
\operatorname{dist}(\bbeta, \mathcal{B}) := \|\bbeta - \proj_{\mathcal{B}}(\bbeta)\|$.
For a general vector field $F:\mathbb{R}^d \to \mathbb{R}^d$, its solution set is defined as
\begin{align*}
\operatorname{Sol}(F)
:\!&= \left\{\hat{\bbeta} \in \mathbb{R}^d :
\langle F(\hat{\bbeta}), \bbeta - \hat{\bbeta} \rangle \ge 0,\ 
\forall \bbeta \in \mathbb{R}^d \right\}= \left\{\hat{\bbeta} \in \mathbb{R}^d : F(\hat{\bbeta}) = \bz\right\}.
\end{align*}
We use $\nabla F(\bbeta)$ to denote the Jacobian of $F$ at $\bbeta$. 
We consider i.i.d.\ samples $\{(\bx^i, y^i)\}_{i=1}^N$ drawn from an unknown distribution 
$\mathbb{P}$, where each covariate $\bx^i \in \mathbb{R}^{d}$ and response $y^i \in \mathbb{R}$. 
The empirical measure is denoted by
$
\mathbb{P}_N := \frac{1}{N}\sum_{i=1}^N \delta_{(\bx^i, y^i)}.
$
Denote by $V_{(\bx^i, y^i)}$ the per-sample VI operator and by $\mathcal{L}_{(\bx^i, y^i)}$ the per-sample MLE negative log-likelihood function.
We use $\xrightarrow{d}$ to denote convergence in distribution. The notation $\mathcal{O}(\cdot)$ denotes the standard asymptotic order, and $\tilde{\mathcal{O}}(\cdot)$ further hides logarithmic factors. The notation $\mathcal{O}_p(\cdot)$ denotes stochastic order in probability, that is, $X_N = \mathcal{O}_p(a_N)$ means that the sequence $\{X_N / a_N\}$ is bounded in probability. Likewise, the notations $o(\cdot)$ and $o_p(\cdot)$ denote deterministic and stochastic convergence to zero, respectively.

\section{Detailed MLE and VI Optimization Landscape Derivations}
\label{sec:setup_appendix}

In this appendix, we supplement the motivation provided in Section~\ref{sec:setup} with detailed mathematical derivations and a rigorous analysis of the optimization landscapes. 
First, we explicitly derive the gradients for several classical GLMs—including logistic, linear, exponential, and Poisson regression—to demonstrate the precise equivalence between MLE and VI formulation under canonical link functions (Appendix~\ref{app:examples}). 
Second, we investigate the pathological behaviors of MLE in non-canonical settings. We provide an analytical calculation of the Hessian for the softplus link to quantify the ill-conditioning (flatness) of the MLE objective. 

\subsection{Examples of vector fields for canonical MLE and VI}
\label{app:examples}

This section provides several representative examples of empirical vector fields arising from
the MLE and the corresponding VI
formulations. For each GLM, we derive the gradient of the empirical
NLL function $\mathcal{L}_N(\bbeta)$ and verify that it coincides
with the VI operator $V_N(\bbeta)$ under the canonical link function.

\subsubsection{Logistic regression}
For the binary logistic regression,  the empirical NLL is given by
\begin{align*}
\mathcal{L}_N(\bbeta)
= -\frac{1}{N}\sum_{i=1}^N \Bigl(
&y^i\log g^{-1}(\bbeta^\top \tilde{\bx}^i)+ (1-y^i)\log\bigl(1-g^{-1}(\bbeta^\top \tilde{\bx}^i)\bigr)
\Bigr),
\end{align*}
with gradient given by
\begin{align*}
\nabla\mathcal{L}_N(\bbeta)
= \frac{1}{N}\sum_{i=1}^N\;&
\frac{(g^{-1})'(\bbeta^\top \tilde{\bx}^i)}
{g^{-1}(\bbeta^\top \tilde{\bx}^i)\bigl(1-g^{-1}(\bbeta^\top \tilde{\bx}^i)\bigr)}\cdot\Bigl(g^{-1}(\bbeta^\top \tilde{\bx}^i)-y^i\Bigr)\,\tilde{\bx}^i.
\end{align*}
For the logistic (sigmoid) link function
$
g^{-1}(z) = 1/(1+e^{-z}),
$
one can show that $(g^{-1})'(z)=g^{-1}(z)(1-g^{-1}(z))$, which simplifies the gradient to
\[
\nabla\mathcal{L}_N(\bbeta)
= \frac{1}{N}\sum_{i=1}^N
\left(g^{-1}(\bbeta^\top \tilde{\bx}^i)-y^i\right)\tilde{\bx}^i,
\]
and the last term above is identical to $V_N(\bbeta)$ in this case.

\subsubsection{Linear regression}

For the Gaussian (linear) model with an identity link, the empirical NLL is
\[
\mathcal{L}_N(\bbeta)
= \frac{1}{2N}\sum_{i=1}^N \left(g^{-1}(\bbeta^\top \tilde{\bx}^i)-y^i\right)^2,
\]
whose gradient is
\[
\nabla\mathcal{L}_N(\bbeta)
= \frac{1}{N}\sum_{i=1}^N (g^{-1})'(\bbeta^\top \tilde{\bx}^i)
\left(g^{-1}(\bbeta^\top \tilde{\bx}^i)-y^i\right)\tilde{\bx}^i.
\]
With the identity link $g^{-1}(z)=z$, we have $(g^{-1})'(z)=1$, leading to
\[
\nabla\mathcal{L}_N(\bbeta)
= \frac{1}{N}\sum_{i=1}^N \left(g^{-1}(\bbeta^\top \tilde{\bx}^i)-y^i\right)\tilde{\bx}^i
= V_N(\bbeta).
\]

\subsubsection{Exponential regression}

For the exponential model, the empirical NLL with $\bbeta^\top \tilde{\bx}^i>0$ is
\[
\mathcal{L}_N(\bbeta)
= \frac{1}{N}\sum_{i=1}^N
\left(\log g^{-1}(\bbeta^\top \tilde{\bx}^i) + y^i/g^{-1}(\bbeta^\top \tilde{\bx}^i)\right),
\]
and the corresponding gradient is
\[
\nabla\mathcal{L}_N(\bbeta)
= \frac{1}{N}\sum_{i=1}^N (g^{-1})'(\bbeta^\top \tilde{\bx}^i)
\left(\frac{1}{g^{-1}(\bbeta^\top \tilde{\bx}^i)}
- \frac{y^i}{g^{-1}(\bbeta^\top \tilde{\bx}^i)^2}\right)\tilde{\bx}^i.
\]
For the inverse link $g^{-1}(z)=z^{-1}$, we have $(g^{-1})'(z)=-z^{-2}$, yielding
\[
\nabla\mathcal{L}_N(\bbeta)
= -\frac{1}{N}\sum_{i=1}^N
\left(g^{-1}(\bbeta^\top \tilde{\bx}^i)-y^i\right)\tilde{\bx}^i
= -V_N(\bbeta).
\]
The gradient $\nabla\mathcal{L}_N(\bbeta)$ differs from $V_N(\bbeta)$ only by a sign, so minimizing the NLL is equivalent to solving $V_N(\bbeta)=\bz$, yielding the same optimal solution.

\subsubsection{Poisson regression}

For the Poisson model, the empirical NLL is
\[
\mathcal{L}_N(\bbeta)
= \frac{1}{N}\sum_{i=1}^N
\left(-y^i\log g^{-1}(\bbeta^\top \tilde{\bx}^i) + g^{-1}(\bbeta^\top \tilde{\bx}^i)\right).
\]
The gradient is
\[
\nabla\mathcal{L}_N(\bbeta)
= \frac{1}{N}\sum_{i=1}^N (g^{-1})'(\bbeta^\top \tilde{\bx}^i)
\left(-\frac{y^i}{g^{-1}(\bbeta^\top \tilde{\bx}^i)}+1\right)\tilde{\bx}^i.
\]
With the exponential link $g^{-1}(z)=e^z$, we have $(g^{-1})'(z)=g^{-1}(z)$, so
\[
\nabla\mathcal{L}_N(\bbeta)
= \frac{1}{N}\sum_{i=1}^N
\left(g^{-1}(\bbeta^\top \tilde{\bx}^i)-y^i\right)\tilde{\bx}^i
= V_N(\bbeta).
\]

\subsection{Flatness of MLE optimization landscape}

 We show that even when standard smooth link functions are applied to improve curvature, the resulting optimization landscape of the MLE can still exhibit considerably slower convergence than the VI formulation. 
This phenomenon is illustrated using the softplus link function, under which the MLE score and the VI operator define different estimating equations. 

Recall the softplus inverse link function:
\[
g^{-1}(z)=\log(1+e^z),
\qquad
(g^{-1})'(z)=\sigma(z):=\frac{1}{1+e^{-z}},
\]
and
\[
(g^{-1})''(z)=\sigma(z)(1-\sigma(z)).
\]
For a single sample $(\bx,y)$, the Poisson negative
log-likelihood is
\[
\mathcal{L}_{(\bx,y)}(\bbeta)
=
-y\log(g^{-1}(z))+g^{-1}(z),
\quad
z=\tilde{\bx}^{\top}\bbeta,
\]
with gradient and Hessian given by, respectively,
\[
\nabla\mathcal{L}_{(\bx,y)}(\bbeta)
=
\sigma(z)
\left(1-\frac{y}{\log(1+e^z)}\right)\tilde{\bx},
\]
and
\[
\nabla^2\mathcal{L}_{(\bx,y)}(\bbeta)
=
\left\{
\sigma(z)(1-\sigma(z))
\left(1-\frac{y}{\log(1+e^z)}\right)
+
\frac{y\sigma(z)^2}{\log(1+e^z)^2}
\right\}
\tilde{\bx}\tilde{\bx}^{\top}.
\]
On the other hand, the VI estimating operator is
$
V_{(\bx,y)}(\bbeta)
=
(g^{-1}(z)-y)\tilde{\bx}$. At the population level, suppose that the conditional mean
is correctly specified, and write
$z^\star:=\tilde{\bx}^{\top}\bbeta^\star$, so that
$\mathbb{E}[y\mid\bx]=g^{-1}(z^\star)$.
Evaluating at $\bbeta^\star$ gives
\[
\nabla V(\bbeta^\star)
=
\mathbb{E}_{\bx\sim\mathbb{P}_{\bx}}
\left[
\sigma(z^\star)\tilde{\bx}\tilde{\bx}^{\top}
\right].
\]
Taking conditional expectations in the preceding Hessian
formula and evaluating at $\bbeta^\star$, we know that the first term
vanishes, and hence
\[
\nabla^2\mathcal{L}(\bbeta^\star)
=
\mathbb{E}_{\bx\sim\mathbb{P}_{\bx}}
\left[
\frac{\sigma(z^\star)^2}{g^{-1}(z^\star)}
\tilde{\bx}\tilde{\bx}^{\top}
\right].
\]
Since $g^{-1}(z)\geq\sigma(z)$ for all $z$,
we have
$\sigma(z)\geq\sigma(z)^2/g^{-1}(z)$ pointwise.
Consequently,
\[
\begin{aligned}
&\nabla V(\bbeta^\star)
-\nabla^2\mathcal{L}(\bbeta^\star)=
\mathbb{E}_{\bx\sim\mathbb{P}_{\bx}}
\left[
\left(
\sigma(z^\star)
-\frac{\sigma(z^\star)^2}{g^{-1}(z^\star)}
\right)
\tilde{\bx}\tilde{\bx}^{\top}
\right]
\succeq\bz.
\end{aligned}
\]
Moreover, the corresponding scalar weights satisfy
\[
\frac{\sigma(z)}{\sigma(z)^2/g^{-1}(z)}
=
\frac{g^{-1}(z)}{\sigma(z)}
=
\frac{\log(1+e^z)(1+e^z)}{e^z}
\longrightarrow+\infty
\qquad\text{as }z\to+\infty.
\]
In particular, $\sigma(z)\to1$, whereas
$\sigma(z)^2/g^{-1}(z)\sim1/z$.
Thus, in the right tail, the likelihood curvature weight
decays relative to the VI Jacobian weight.
This comparison illustrates a stronger local response of
the VI operator in large-predictor regimes.

\section{Characterization of Minty Conditions}\label{app:minty}

To start with, we introduce the following definition, which is widely used for solving variational inequalities, e.g., \cite{huang2023beyond}.
\begin{definition}[Strong Minty condition]\label{def:strong_minty}
The vector field $F:\R^{d}\rightarrow \R^{d}$ satisfies the strong Minty condition with constant $\mu \geq 0$ if $\operatorname{Sol}(F)\neq\emptyset$ and
$$
\langle F(\bbeta), \bbeta-\proj_{\operatorname{Sol}(F)}(\bbeta)\rangle \geq \mu\cdot\dist^2(\bbeta,\operatorname{Sol}(F)), \quad \forall \bbeta\in \R^{d}.
$$
\end{definition}

The strong Minty condition plays a central role in the VI estimator. In this section, we first establish sufficient conditions under which this property holds and illustrate them through representative examples. 
We then discuss guiding principles for checking link functions that satisfy the strong Minty condition in practice.

\begin{definition}
\label{def:weak_minty_eb}
Given constants $\rho,\mu_{\mathrm{EB}}>0$, we introduce the following properties:
\begin{itemize}
\item   
$F$ is said to be $\rho$-weakly monotone if
\begin{equation*}
\langle F(\bbeta_2)-F(\bbeta_1),\, \bbeta_2-\bbeta_1\rangle 
\ge
-\rho\|\bbeta_2-\bbeta_1\|^2,
\qquad \forall\,\bbeta_1,\bbeta_2\in\mathbb{R}^d.
\end{equation*}
\item  
$F$ is said to satisfy the error bound property with constant $\mu_{\mathrm{EB}}>0$ if
\begin{equation*}
\|F(\bbeta)-F(\proj_{\operatorname{Sol}(F)}(\bbeta))\|
\ge
\mu_{\mathrm{EB}}\cdot\mathrm{dist}(\bbeta,\,\operatorname{Sol}(F)),
\qquad \forall\,\bbeta\in\mathbb{R}^d.
\end{equation*}
\end{itemize}
\end{definition}

Weak monotonicity is a mild regularity condition that allows certain non-monotone vector fields while still ensuring local stability of the VI solution.
In contrast, the error bound condition provides a geometric separation property: it quantifies how far a point $\bbeta$ is from the solution set in terms of the magnitude of the operator, implying that the vector field 
$F(\bbeta)$ is uniformly separated from zero along the direction of its projection onto the solution set.
Intuitively, this means that $F$ grows at least linearly with the distance to the solution. With these two conditions, we have the following implication.

\begin{proposition}\label{cor:weakcvx}
Suppose that the vector field
$F=\nabla f:\mathbb{R}^{d}\rightarrow\mathbb{R}^{d}$
is the gradient of a continuously differentiable function
$f:\mathbb{R}^{d}\rightarrow\mathbb{R}$,
that $\operatorname{Sol}(F)\neq\emptyset$, and that
$F$ satisfies the $L$-Lipschitz condition, i.e.,
\begin{equation*}
\|F(\bbeta_2)-F(\bbeta_1)\|
\leq L\|\bbeta_2-\bbeta_1\|,
\qquad
\forall\,\bbeta_1,\bbeta_2\in\mathbb{R}^d.
\end{equation*}
Additionally, we assume that $F$ satisfies the
$\rho$-weak monotonicity condition in
Definition~\ref{def:weak_minty_eb} with $0<\rho<L$,
and the error bound property holds with constant
$\mu_{\rm EB}>0$.
If $\rho L<\mu_{\rm EB}^2$, then the strong Minty condition
holds with
$
\mu=\frac{\mu_{\rm EB}^2-\rho L}{L-\rho}$.
\end{proposition}

\begin{proof}
Since $F$ satisfies the $\rho$-weak monotonicity and
$L$-Lipschitz conditions, we know that
$\bbeta\mapsto F(\bbeta)+\rho\bbeta$ is monotone and
$(L+\rho)$-Lipschitz.
Moreover,
\[
F(\bbeta)+\rho\bbeta
=
\nabla\left(f(\bbeta)+\frac{\rho}{2}\|\bbeta\|^2\right),
\]
and the function
$\bbeta\mapsto f(\bbeta)+\frac{\rho}{2}\|\bbeta\|^2$
is convex.
Fix $\bbeta\in\mathbb{R}^d$ and choose any Euclidean
projection $\proj_{\operatorname{Sol}(F)}(\bbeta)$.
By the Baillon--Haddad theorem
\citep[Corollaire~10]{baillon1977quelques}, we have
\begin{align*}
&\left\langle
F(\bbeta)-F(\proj_{\operatorname{Sol}(F)}(\bbeta))
+\rho(\bbeta-\proj_{\operatorname{Sol}(F)}(\bbeta)),
\bbeta-\proj_{\operatorname{Sol}(F)}(\bbeta)
\right\rangle\\
\geq\;&
\frac{1}{L+\rho}
\left\|
F(\bbeta)-F(\proj_{\operatorname{Sol}(F)}(\bbeta))
+\rho(\bbeta-\proj_{\operatorname{Sol}(F)}(\bbeta))
\right\|^2,
\end{align*}
which implies that
\begin{align*}
&\frac{L-\rho}{L+\rho}
\left\langle
F(\bbeta)-F(\proj_{\operatorname{Sol}(F)}(\bbeta)),
\bbeta-\proj_{\operatorname{Sol}(F)}(\bbeta)
\right\rangle\\
\geq\;&
\frac{1}{L+\rho}
\|F(\bbeta)-F(\proj_{\operatorname{Sol}(F)}(\bbeta))\|^2
-\frac{\rho L}{L+\rho}
\|\bbeta-\proj_{\operatorname{Sol}(F)}(\bbeta)\|^2.
\end{align*}
Combining this with the error bound property and using
$L-\rho>0$, we obtain
\begin{align*}
&\left\langle
F(\bbeta)-F(\proj_{\operatorname{Sol}(F)}(\bbeta)),
\bbeta-\proj_{\operatorname{Sol}(F)}(\bbeta)
\right\rangle\\
\geq\;&
\frac{\mu_{\rm EB}^2}{L-\rho}
\|\bbeta-\proj_{\operatorname{Sol}(F)}(\bbeta)\|^2
-\frac{\rho L}{L-\rho}
\|\bbeta-\proj_{\operatorname{Sol}(F)}(\bbeta)\|^2\\
=\;&
\frac{\mu_{\rm EB}^2-\rho L}{L-\rho}
\|\bbeta-\proj_{\operatorname{Sol}(F)}(\bbeta)\|^2.
\end{align*}
Since $F(\proj_{\operatorname{Sol}(F)}(\bbeta))=\bz$,
this is the desired strong Minty inequality.
The proof is complete.
\end{proof}

It is straightforward to verify that if $F$ is strongly monotone, then it implies the error bound property, which characterizes the local growth behavior around the solution set and implies the strong Minty condition. Moreover, we provide an example showing that $F$ need not be monotone for the strong Minty condition to hold.

\begin{example}
In many practical systems, such as those encountered in signal processing and communications, the ideal input-output response is approximately linear. However, due to hardware imperfections or nonlinear circuit effects, small nonlinear distortions may arise. A simple model capturing such behavior is given by
\[
g^{-1}(z) = z + 2\sin(z)\cos(z).
\]
The derivative of this function is
\[
(g^{-1})'(z) = 1 + 2 (\cos^2(z)-\sin^2(z))=3\cos^2(z)-\sin^2(z),
\]
which is not always positive; hence, \(g^{-1}\) is not a monotone function.  
Nevertheless, let \(z^\star = 0\) denote the solution to \(g^{-1}(z)=0\). Then
\[
\langle g^{-1}(z)-g^{-1}(z^\star),\, z - z^\star \rangle = g^{-1}(z)z\ge z^2 + 2 z\sin(z)\cos(z)=z(z + \sin(2z))\ge \frac{1}{2}z^2.
\]
This inequality confirms that, despite the lack of monotonicity, the mapping \(g^{-1}\) satisfies the strong Minty condition with modulus $\mu =1/2$.
\end{example}

The following lemma shows that the vector field $V_N$ satisfies the strong Minty condition on a given parameter set whenever a weighted design matrix is positive definite, where the weights are the slopes of secants of the inverse link function $g^{-1}$. With a fixed sample size $N$, we collect the covariates into the data matrix
$\bm{X}_N \in \mathbb{R}^{N \times d}$ defined as
\[
\bm{X}_N :=
\begin{bmatrix}
x^1_1 & \cdots & x^1_{d} \\
\vdots & \ddots & \vdots \\
x^N_1 & \cdots & x^N_{d}
\end{bmatrix},
\]
where $x^i_j$ denotes the $j$-th covariate of the $i$-th observation, and write $\tilde{\bm{X}}_N:=[\bm{1}\; \bm{X}_N]$.

\begin{lemma}[Sufficient condition for the strong Minty condition]\label{lem:strongminty}
Suppose that Assumption~\ref{assump}(i) holds, that $V_N(\hat{\bbeta}_N)=\bz$, and let $\mathcal{B}\subset\R^{d+1}$ be a set containing $\hat{\bbeta}_N$. For $i\in[N]$, write $z_i:=\tilde{\bx}^{i\top}\hat{\bbeta}_N$ and let
\[
\gamma_i:=\inf\left\{\frac{g^{-1}(z)-g^{-1}(z_i)}{z-z_i}:\ z=\tilde{\bx}^{i\top}\bbeta,\ \bbeta\in\mathcal{B},\ z\neq z_i\right\}
\]
be the infimum of the slopes of the secants of $g^{-1}$ from $z_i$ to the predictor values attained on $\mathcal{B}$ (and $\gamma_i:=0$ if this set is empty). Then, for any $\bbeta\in\mathcal{B}$,
\[
\langle V_N(\bbeta),\bbeta-\hat{\bbeta}_N\rangle
\ \ge\ \lambda_{\min}\left(\frac1N\sum_{i=1}^N\gamma_i\,\tilde{\bx}^i\tilde{\bx}^{i\top}\right)\|\bbeta-\hat{\bbeta}_N\|^2.
\]
In particular, if $g^{-1}$ is $\mu_g$-strongly monotone with $\mu_g>0$, then for $\mathcal{B}=\R^{d+1}$ every $\gamma_i\ge\mu_g$, and the vector field $V_N$ satisfies the strong Minty condition on $\R^{d+1}$ with modulus $\mu_g\sigma_N^2/N$, where $\sigma_N$ is the minimal singular value of $\tilde{\bm{X}}_N$.
\end{lemma}

\begin{proof}
Fix $\bbeta\in\mathcal{B}$. Recall the definition that $V_N(\bbeta)=\frac{1}{N}\sum_{i=1}^N (g^{-1}(\tilde{\bx}^{i\top}\bbeta)-y^i)\cdot\tilde{\bx}^i$. Since $V_N(\hat{\bbeta}_N)=\bz$, we have
\begin{align*}
\langle V_N(\bbeta),\bbeta-\hat{\bbeta}_N\rangle
&=\frac{1}{N}\sum_{i=1}^N\left(g^{-1}(\tilde{\bx}^{i\top}\bbeta)-g^{-1}(\tilde{\bx}^{i\top}\hat{\bbeta}_N)\right)\cdot\tilde{\bx}^{i\top}(\bbeta-\hat{\bbeta}_N)\\
&\ge\frac{1}{N}\sum_{i=1}^N\gamma_i\left|\tilde{\bx}^{i\top}(\bbeta-\hat{\bbeta}_N)\right|^2\\
&=(\bbeta-\hat{\bbeta}_N)^\top\left(\frac{1}{N}\sum_{i=1}^N\gamma_i\,\tilde{\bx}^i\tilde{\bx}^{i\top}\right)(\bbeta-\hat{\bbeta}_N)\\
&\ge\lambda_{\min}\left(\frac{1}{N}\sum_{i=1}^N\gamma_i\,\tilde{\bx}^i\tilde{\bx}^{i\top}\right)\|\bbeta-\hat{\bbeta}_N\|^2.
\end{align*}
This establishes the first claim. 
It remains to consider the case where $g^{-1}$ is $\mu_g$-strongly monotone and $\mathcal{B}=\R^{d+1}$. In this case, the set defining $\gamma_i$ is nonempty since $\tilde{\bx}^i\neq\bz$, and every secant slope of $g^{-1}$ is at least $\mu_g$ by the strong monotonicity, so that $\gamma_i\ge\mu_g$ for any $i\in[N]$. Hence
\[
\frac{1}{N}\sum_{i=1}^N\gamma_i\,\tilde{\bx}^i\tilde{\bx}^{i\top}\succeq\frac{\mu_g}{N}\tilde{\bm{X}}_N^\top\tilde{\bm{X}}_N,
\]
and the minimal eigenvalue of the right-hand side is $\mu_g\sigma_N^2/N$. Then the strong Minty condition on $\R^{d+1}$ follows from the first claim. The proof is complete.
\end{proof}

A population version holds under analogous conditions. If $V(\bbeta^\star)=\bz$ and $\mathbb{E}\|\tilde{\bx}\|^2<\infty$, then the same argument, with $\hat{\bbeta}_N$ replaced by $\bbeta^\star$ and the average replaced by the expectation, shows that for any set $\mathcal{B}$ containing $\bbeta^\star$ and any $\bbeta\in\mathcal{B}$,
\[
\langle V(\bbeta),\bbeta-\bbeta^\star\rangle \ge \lambda_{\min}(\mathbb{E}[\gamma(\bx)\tilde{\bx}\tilde{\bx}^\top])\|\bbeta-\bbeta^\star\|^2,
\]
where $\gamma(\bx)$ is defined in the same way as $\gamma_i$, with $\tilde{\bx}$ in place of $\tilde{\bx}^i$ and $\bbeta^\star$ in place of $\hat{\bbeta}_N$. The expectation is finite since $|\gamma(\bx)|\le L$ by Assumption~\ref{assump}(i).

\begin{remark}\label{rem:secant_minty}
Each $\gamma_i$ is a one-dimensional quantity and can be computed once $\mathcal{B}$ is fixed. For $\mathcal{B}=\{\bbeta:\|\bbeta-\hat{\bbeta}_N\|\le r\}$ with $r>0$, the predictor values $\tilde{\bx}^{i\top}\bbeta$ range over $[z_i-r\|\tilde{\bx}^i\|,\,z_i+r\|\tilde{\bx}^i\|]$. If $(g^{-1})'\ge c$ almost everywhere on this interval, then $\gamma_i\ge c$, since each secant slope is an average of $(g^{-1})'$. 
For the softplus, Gaussian-mixture CDF, and arctangent links, $(g^{-1})'$ is positive and continuous, so $\gamma_i>0$ for every $i$ on such a ball, and the modulus in Lemma~\ref{lem:strongminty} is positive whenever $\tilde{\bm{X}}_N$ has full column rank. This is how Assumption~\ref{assump}(iv) can hold on a bounded set for links that are not strongly monotone.
For a clipped link and the same ball, $\gamma_i=0$ only for the observations whose fitted predictor $z_i$ is saturated, while an observation whose predictors enter the flat region for some $\bbeta\in\mathcal{B}$ keeps $\gamma_i>0$, since the secant from an unsaturated $z_i$ into the flat region still has positive slope. Hence the strong Minty condition can hold on $\mathcal{B}$ whenever the unsaturated observations span $\R^{d+1}$, although $g^{-1}$ is not strongly monotone on any interval that overlaps a flat region.
\end{remark}

\section{Proofs for Section \ref{sec:estimator}}

\begin{lemma}\label{lem:upperbd}
Suppose that Assumptions~\ref{assump} {\rm(i)}--{\rm(iii)}
hold.
Then, for almost every $(\bx,y)$, we have
\[
\|V_{(\bx,y)}(\bbeta^\star)\|_{\infty}\leq RM,
\qquad
\|V_{(\bx,y)}(\bbeta^\star)\|
\leq\sqrt{d+1}RM,
\]
and for any $\bbeta_1,\bbeta_2\in\mathbb{R}^{d+1}$,
\[
\|V_{(\bx,y)}(\bbeta_2)-V_{(\bx,y)}(\bbeta_1)\|
\leq
(LdM^2+L)\|\bbeta_2-\bbeta_1\|.
\]
Consequently, for any $\bbeta\in\mathbb{R}^{d+1}$,
\[
\|V_{(\bx,y)}(\bbeta)\|^2
\leq
2(d+1)R^2M^2
+
2L^2(1+dM^2)^2\|\bbeta-\bbeta^\star\|^2.
\]
\end{lemma}

\begin{proof}
First, from Assumptions~\ref{assump} {\rm(ii)}--{\rm(iii)}
it follows that, almost surely,
$\|\bx\|_{\infty}\leq M$,
$\|\bx\|\leq\sqrt{d}M$, and
\[
\left|
g^{-1}\left(\beta_0^\star+\sum_{j=1}^{d}\beta_j^\star x_j\right)
-y
\right|
\leq R.
\]
Then we have the upper bound (assuming $M\geq1$) that
\[
\begin{aligned}
\|V_{(\bx,y)}(\bbeta^\star)\|_{\infty}
&\leq
\left|
g^{-1}\left(\beta_0^\star+\sum_{j=1}^{d}\beta_j^\star x_j\right)
-y
\right|
\|\tilde{\bx}\|_{\infty}
\leq RM,\\
\|V_{(\bx,y)}(\bbeta^\star)\|
&\leq
\left|
g^{-1}\left(\beta_0^\star+\sum_{j=1}^{d}\beta_j^\star x_j\right)
-y
\right|
\|\tilde{\bx}\|
\leq\sqrt{d+1}\,RM.
\end{aligned}
\]
Next, for the Lipschitz constant of $V_{(\bx,y)}$,
one has that
\[
\begin{aligned}
\|V_{(\bx,y)}(\bbeta')-V_{(\bx,y)}(\bbeta)\|
&=
\left\|
\left(
g^{-1}\left(\beta_0'+\sum_{j=1}^{d}\beta_j'x_j\right)
-
g^{-1}\left(\beta_0+\sum_{j=1}^{d}\beta_jx_j\right)
\right)\tilde{\bx}
\right\|\\
&\leq
L\|\tilde{\bx}\|\cdot
\left|
\beta_0'-\beta_0+
\sum_{j=1}^{d}(\beta_j'-\beta_j)x_j
\right|\\
&\leq
L\|\tilde{\bx}\|^2\|\bbeta'-\bbeta\|\\
&\leq
(LdM^2+L)\|\bbeta'-\bbeta\|.
\end{aligned}
\]
Finally, combining these bounds gives
\[
\begin{aligned}
\|V_{(\bx,y)}(\bbeta)\|^2
&\leq
2\|V_{(\bx,y)}(\bbeta^\star)\|^2
+
2\|V_{(\bx,y)}(\bbeta)-V_{(\bx,y)}(\bbeta^\star)\|^2\\
&\leq
2(d+1)R^2M^2
+
2L^2(1+dM^2)^2\|\bbeta-\bbeta^\star\|^2.
\end{aligned}
\]
The proof is complete.
\end{proof}

\subsection*{Proof of Theorem \ref{thm:estimator}}

First, from the strong Minty condition in Assumption \ref{assump} (iv), we know that
\[
\mu\|\bbeta^{\star}-\hat{\bbeta}_N\|^2
\leq\langle V_N(\bbeta^{\star}),\bbeta^{\star}-\hat{\bbeta}_N\rangle\leq\|V_N(\bbeta^{\star})\|\|\bbeta^{\star}-\hat{\bbeta}_N\|,
\]
which implies 
\begin{equation}\label{eq:estimate_key}
\mu\|\bbeta^{\star}-\hat{\bbeta}_N\|\leq \|V_N(\bbeta^{\star})\|.
\end{equation}
Recall the definition that for any $\bbeta\in\R^{d+1}$
\[
V_N(\bbeta)=\mathbb{E}_{(\bx,y)\sim\mathbb{P}_N} \left[\left(g^{-1}\left(\beta_0+\sum_{j=1}^{d} \beta_{j}x_{j}\right)-y\right)\cdot\tilde{\bx}\right]=\mathbb{E}_{(\bx,y)\sim\mathbb{P}_N}V_{(\bx, y)}(\bbeta).
\]
Then we know from Lemma \ref{lem:upperbd} that, almost surely,
\[
\|V_{(\bx^i,y^i)}(\bbeta^{\star})\|_{\infty}\le RM,\qquad i=1,\dots,N.
\]
As the samples are i.i.d.\ and $V(\bbeta^{\star})=\bz$, by applying Hoeffding's inequality \citep[Theorem 2.2.6]{vershynin2018high} to the coordinates of $V_{(\bx^i,y^i)}(\bbeta^{\star})$, it follows that for any $t>0$ and $j\in\{0,\dots,d\}$,
\begin{equation*}
\begin{aligned}
\mathbb{P}\left(|V_N(\bbeta^{\star})_{j}|\ge t\right)
&=\mathbb{P}\left(|V_N(\bbeta^{\star})_{j}-V(\bbeta^{\star})_{j}|\ge t\right)\\
&=\mathbb{P}\left(V_N(\bbeta^{\star})_{j}-V(\bbeta^{\star})_{j}\ge t\right)+\mathbb{P}\left(-V_N(\bbeta^{\star})_{j}+V(\bbeta^{\star})_{j}\ge t\right)\\
&\leq 2\exp\left(-\frac{N t^2}{2R^2M^2}\right).
\end{aligned}
\end{equation*}
Using the union bound (i.e., Boole's inequality),
\begin{equation*}
\begin{aligned}
\mathbb{P}\left(\|V_N(\bbeta^{\star})\|\ge t\right)
&\leq \mathbb{P}\left(\sqrt{d+1}\|V_N(\bbeta^{\star})\|_{\infty}\ge t\right)\\
&\leq \sum_{j=0}^d\mathbb{P}\left(|V_N(\bbeta^{\star})_{j}|\ge \frac{t}{\sqrt{d+1}}\right)\le 2(d+1)\exp\left(-\frac{Nt^2}{2 (d+1)R^2 M^2}\right).
\end{aligned}
\end{equation*}
Equivalently, for any $\epsilon\in (0,1)$, set $t=RM\sqrt{2(d+1)\ln(2(d+1)/\epsilon)/N}$. Then with probability at least $1-\epsilon$,
\[
\|V_N(\bbeta^{\star})\|\le RM\sqrt{\frac{2(d+1)\ln(2(d+1)/\epsilon)}{N}}.
\]
Combining this with \eqref{eq:estimate_key}, we derive the desired results.

\subsection*{Proof of Lemma \ref{lem:V_diff}}

Let $u:=\langle\tilde{\bx},\bbeta^\star\rangle$ and 
$\delta_t:=\langle\tilde{\bx},\bm{h}_t\rangle$ for any sequence $\bm{h}_t\to \bm{h}$ as $t\downarrow0$. 
Then
\begin{equation}\label{eq:smooth_key1}
\frac{V(\bbeta^\star+t \bm{h}_t)-V(\bbeta^\star)}{t}
=\mathbb E_{(\bx,y)\sim\mathbb{P}}\left[\frac{g^{-1}(u+t\delta_t)-g^{-1}(u)}{t}\,\tilde{\bx}\right].
\end{equation}
Since $\bP(u\in\mathcal K)=0$ by Assumption~\ref{assump_asym}, 
$g^{-1}$ is differentiable at $u$ almost surely, and therefore
\[
g^{-1}(u+t\delta_t)
= g^{-1}(u) + (g^{-1})'(u)\,t\delta_t + o(t|\delta_t|) \quad\text{a.s.}
\]
Hence, pointwise almost surely,
\begin{equation}\label{eq:smooth_key2}
\frac{g^{-1}(u+t\delta_t)-g^{-1}(u)}{t}\,\tilde{\bx}\
\rightarrow\
(g^{-1})'(u)\,\langle\tilde{\bx},\bm{h}\rangle\,\tilde{\bx}.
\end{equation}
Moreover, by the global Lipschitz property of $g^{-1}$ with constant $L$,
\[
\left\|\frac{g^{-1}(u+t\delta_t)-g^{-1}(u)}{t}\,\tilde{\bx}\right\|
\le L\,|\delta_t|\,\|\tilde{\bx}\|
\le L\,\|\tilde{\bx}\|^2\|\bm{h}_t\|.
\]
Under Assumption~\ref{assump}, we have  
$\|\tilde{\bx}\|^2\le 1+dM^2$, 
so the right-hand side is bounded by $L(1+dM^2)\|\bm{h}_t\|$, 
which is integrable and independent of $t$ since $\|\bm{h}_t\|$ is bounded. 
From \eqref{eq:smooth_key1} and \eqref{eq:smooth_key2} with the dominated convergence theorem, it follows that
\[
\frac{V(\bbeta^\star+t \bm{h}_t)-V(\bbeta^\star)}{t}
\ \rightarrow\
\mathbb E_{(\bx,y)\sim\mathbb{P}}\left[(g^{-1})'(\langle \tilde{\bx},\bbeta^{\star}\rangle)\tilde{\bx}\tilde{\bx}^\top\right]\bm{h}.
\]
Thus $V$ is Fr\'echet differentiable at $\bbeta^\star$ with 
Jacobian $\nabla V(\bbeta^\star)$ as claimed.

\subsection*{Proof of Theorem \ref{thm:VI-CLT}}

By Assumption~\ref{assump},
$|g^{-1}(\langle\tilde{\bx},\bbeta^\star\rangle)-y|\leq R$,
$\|\bx\|^2\leq dM^2$, and $g^{-1}$ is Lipschitz.
Hence,
\[
\begin{aligned}
\mathbb{E}_{(\bx,y)\sim\mathbb{P}}
\left[\|V_{(\bx,y)}(\bbeta^\star)\|^2\right]
&=
\mathbb{E}_{(\bx,y)\sim\mathbb{P}}
\left[
|g^{-1}(\langle\tilde{\bx},\bbeta^\star\rangle)-y|^2
\|\tilde{\bx}\|^2
\right]\leq R^2(1+dM^2)<\infty.
\end{aligned}
\]
Therefore, when $(\bx^i,y^i)_{i=1}^N$ are i.i.d.\ samples,
the random vectors $V_{(\bx^i,y^i)}(\bbeta^\star)$
are i.i.d.\ with mean zero and covariance matrix $\Gamma$
finite and positive definite.
Applying the multivariate central limit theorem yields
\begin{equation}\label{eq:score-clt}
\sqrt{N}\big(V_N(\bbeta^\star)-V(\bbeta^\star)\big)
\xrightarrow{d}\mathcal{N}(\bz,\Gamma).
\end{equation}

Adding and subtracting $V(\hat{\bbeta}_N)$ and
$V(\bbeta^\star)$ in $V_N(\hat{\bbeta}_N)$ yields
\begin{equation}\label{eq:add-subtract}
\begin{aligned}
\bz=V_N(\hat{\bbeta}_N)
&=
V_N(\bbeta^\star)
+
\underbrace{
V(\hat{\bbeta}_N)-V(\bbeta^\star)
}_{\text{population increment}}+
\underbrace{
(V_N(\hat{\bbeta}_N)-V(\hat{\bbeta}_N))
-
(V_N(\bbeta^\star)-V(\bbeta^\star))
}_{=:R_N}.
\end{aligned}
\end{equation}

By Assumption~\ref{assump}(iv) and Lemma~\ref{lem:V_diff}, for any $\bm{h}\in\R^{d+1}$,
\[
\bm{h}^\top\nabla V(\bbeta^\star)\bm{h}
=\lim_{t\downarrow0}\frac{\langle V(\bbeta^\star+t\bm{h}),t\bm{h}\rangle}{t^2}
\ge\mu\|\bm{h}\|^2,
\]
so $\nabla V(\bbeta^\star)$ is non-singular.
By Theorem~\ref{thm:estimator},
$\|\hat{\bbeta}_N-\bbeta^\star\|=O_p(N^{-1/2})$.
The Fr\'echet differentiability of $V$ at $\bbeta^\star$
(Lemma~\ref{lem:V_diff}) therefore gives the local linearization
\begin{equation}\label{eq:linearization}
\begin{aligned}
V(\hat{\bbeta}_N)-V(\bbeta^\star)
=
\nabla V(\bbeta^\star)(\hat{\bbeta}_N-\bbeta^\star)+r_N,\quad
\|r_N\|=o_p(\|\hat{\bbeta}_N-\bbeta^\star\|).
\end{aligned}
\end{equation}

From Lemma~\ref{lem:upperbd}, the sample operator
$V_{(\bx,y)}$ is $(LdM^2+L)$-Lipschitz.
Write $L_V:=LdM^2+L$ and note that
\[
R_N
=
[\mathbb{E}_N-\mathbb{E}]
\left(
V_{(\cdot)}(\hat{\bbeta}_N)-V_{(\cdot)}(\bbeta^\star)
\right),
\]
where $\mathbb{E}_N$ denotes averaging over the sample.
For a fixed $0<r\leq1$, define the local class
\[
\mathcal{F}_r
:=
\left\{
f_{\bbeta}(\cdot)
:=
V_{(\cdot)}(\bbeta)-V_{(\cdot)}(\bbeta^\star)
:\;
\|\bbeta-\bbeta^\star\|\leq r
\right\}.
\]
For every coordinate $j\in\{0,\ldots,d\}$, let
$
\mathcal{F}_{r,j}
:=\{
f_{\bbeta,j}:\|\bbeta-\bbeta^\star\|\leq r\}$,
where $f_{\bbeta,j}$ is the $j$-th coordinate of
$f_{\bbeta}$.
Almost surely,
\[
|f_{\bbeta,j}(\bx,y)|\leq L_Vr,
\qquad
|f_{\bbeta,j}(\bx,y)-f_{\bbeta',j}(\bx,y)|
\leq L_V\|\bbeta-\bbeta'\|.
\]
By \citet[Theorem~2.7.11]{van1996weak}, the bracketing numbers of $\mathcal{F}_{r,j}$ in $L_2(\mathbb{P})$ are bounded by the covering numbers of the parameter ball, which are polynomial in $1/\epsilon$. Consequently, each $\mathcal{F}_{r,j}$ is $\mathbb{P}$-Donsker by \citet[Theorem~2.5.6]{van1996weak}.
Moreover, $f_{\bbeta^\star}=\bz$ and
$
(\mathbb{E}\|f_{\bbeta}(\bx,y)\|^2)^{1/2}
\leq L_V\|\bbeta-\bbeta^\star\|$.
The asymptotic equicontinuity of the coordinate empirical
processes therefore implies that, for every $\epsilon>0$,
\begin{equation}\label{eq:covering}
\lim_{\delta\downarrow0}\limsup_{N\to\infty}
\mathbb{P}\left(
\sup_{\|\bbeta-\bbeta^\star\|\leq\delta}
\left\|
\sqrt{N}[\mathbb{E}_N-\mathbb{E}]f_{\bbeta}
\right\|
>\epsilon
\right)=0,
\end{equation}
where $0<\delta\leq r$.
Since $\hat{\bbeta}_N$ is consistent, for every such
$\delta$ we have
\[
\begin{aligned}
\mathbb{P}(\sqrt{N}\|R_N\|>\epsilon)
&\leq
\mathbb{P}(
\|\hat{\bbeta}_N-\bbeta^\star\|>\delta)+
\mathbb{P}\left(
\sup_{\|\bbeta-\bbeta^\star\|\leq\delta}
\left\|
\sqrt{N}[\mathbb{E}_N-\mathbb{E}]f_{\bbeta}
\right\|
>\epsilon
\right).
\end{aligned}
\]
Taking the limit superior as $N\to\infty$ and then letting
$\delta\downarrow0$, we conclude from 
\eqref{eq:covering} that
$
\sqrt{N}\|R_N\|=o_p(1)$,
i.e., $\|R_N\|=o_p(N^{-1/2})$.
Substituting~\eqref{eq:linearization}
into~\eqref{eq:add-subtract} gives
\[
\bz
=
V_N(\bbeta^\star)
+
\nabla V(\bbeta^\star)(\hat{\bbeta}_N-\bbeta^\star)
+r_N+R_N.
\]
Rearranging terms yields
\[
\begin{aligned}
\sqrt{N}(\hat{\bbeta}_N-\bbeta^\star)
&=
-\nabla V(\bbeta^\star)^{-1}
\sqrt{N}\big(V_N(\bbeta^\star)-V(\bbeta^\star)\big)-\nabla V(\bbeta^\star)^{-1}\sqrt{N}(r_N+R_N).
\end{aligned}
\]
The first term converges in distribution to
$\mathcal{N}(\bz,
\nabla V(\bbeta^\star)^{-1}
\Gamma\nabla V(\bbeta^\star)^{-\top})$
by~\eqref{eq:score-clt}.
On the other hand,
$\sqrt{N}\|r_N\|=o_p(1)$
by~\eqref{eq:linearization} and
$\|\hat{\bbeta}_N-\bbeta^\star\|=O_p(N^{-1/2})$,
while $\sqrt{N}\|R_N\|=o_p(1)$ by the preceding argument.
Thus, the second term vanishes in probability.
Applying Slutsky's theorem gives
\[
\sqrt{N}(\hat{\bbeta}_N-\bbeta^\star)
\xrightarrow{d}
\mathcal{N}(
\bz,
\nabla V(\bbeta^\star)^{-1}
\Gamma\nabla V(\bbeta^\star)^{-\top}).
\]
The proof is complete.

\section{Comparison of Asymptotic Efficiency between VI and MLE}
\label{app:eff}

A classical result in asymptotic statistics states that the MLE is asymptotically efficient, namely, it is consistent, asymptotically unbiased, and attains the Cram\'er--Rao lower bound (CRB) under correct model specification.
The asymptotic covariance of the proposed VI estimator,
\[
\Sigma_{\mathrm{VI}}
= \nabla V(\bbeta^\star)^{-1}\Gamma\nabla V(\bbeta^\star)^{-\top},
\]
shares the same ``sandwich'' structure as that of the MLE,
\[
\Sigma_{\mathrm{MLE}}
= \big(\nabla^2 \mathcal{L}(\bbeta^\star)\big)^{-1}
\operatorname{Cov}(\nabla \mathcal{L}_{(\bx,y)}(\bbeta^\star))
\big(\nabla^2 \mathcal{L}(\bbeta^\star)\big)^{-\top},
\]
where the expected Hessian and the score covariance are given by
\[
\nabla^2 \mathcal{L}(\bbeta^\star)
= \mathbb{E}_{(\bx,y)\sim\mathbb{P}}
  \!\left[
    \frac{(g^{-1})'(\tilde{\bx}^\top\bbeta^\star)^2}
         {\operatorname{Var}(y\mid\tilde{\bx})}
    \, \tilde{\bx}\tilde{\bx}^\top
  \right],
\]
and
\begin{align*}
\operatorname{Cov}(\nabla \mathcal{L}_{(\bx,y)}(\bbeta^\star))
&:=   \mathbb{E}_{(\bx,y)\sim\mathbb{P}}
  \!\left[
    \nabla \mathcal{L}_{(\bx,y)}(\bbeta^\star)
    \nabla \mathcal{L}_{(\bx,y)}(\bbeta^\star)^\top
  \right]\\
&=  \mathbb{E}_{(\bx,y)\sim\mathbb{P}}
  \!\left[
    \frac{(y - g^{-1}(\tilde{\bx}^\top\bbeta^\star))^2
    (g^{-1})'(\tilde{\bx}^\top\bbeta^\star)^2}{
    \operatorname{Var}(y\mid\tilde{\bx})^2
    }
    \, \tilde{\bx}\tilde{\bx}^\top
  \right].
\end{align*}
For correctly specified models, the information identity
\[
\operatorname{Cov}\!\big(\nabla \mathcal{L}_{(\bx,y)}(\bbeta^\star)\big)
= \nabla^2 \mathcal{L}(\bbeta^\star)
= I(\bbeta^\star)
\]
holds, and the MLE achieves the Fisher information bound
\(I(\bbeta^\star)^{-1}\), attaining the CRB.
In the following, we show that the MLE can indeed be asymptotically more efficient than the VI estimator under correct model specification. Their asymptotic covariances coincide only when the link function is canonical. This observation does not contradict our earlier discussion, as the VI estimator may still outperform the MLE in the presence of model or variance misspecification.

For correctly specified models with general  link functions, 
the VI estimator generally does not satisfy the corresponding identity
\(\Gamma = \nabla V(\bbeta^\star)\),
since the residual covariance 
and the Jacobian
\(\nabla V(\bbeta^\star)\)
coincide only for canonical links. 
Let
\[
\Sigma_{\mathrm{MLE}}^{-1}
=\mathbb{E}\!\left[\frac{(g^{-1})'(\tilde{\bx}^\top\bbeta^\star)^{2}}
{\operatorname{Var}(y\mid \tilde{\bx})}\,
\tilde{\bx}\tilde{\bx}^\top\right],\]
\[
\bm{r}_1=\sqrt{\operatorname{Var}(y\mid \tilde{\bx})}\,\tilde{\bx},\quad
\bm{r}_2=\frac{(g^{-1})'(\tilde{\bx}^\top\bbeta^\star)}{\sqrt{\operatorname{Var}(y\mid \tilde{\bx})}}\,\tilde{\bx}.
\]
Then $\mathbb{E}[\bm{r}_1\bm{r}_1^\top]=\Gamma$, $\mathbb{E}[\bm{r}_2\bm{r}_2^\top]=\Sigma_{\mathrm{MLE}}^{-1}$, $\mathbb{E}[\bm{r}_1\bm{r}_2^\top]=\nabla V(\bbeta^\star)$,
and hence the block moment matrix
\[
\begin{bmatrix}
\Gamma & \nabla V(\bbeta^\star)\\[2pt]
\nabla V(\bbeta^\star) & \Sigma_{\mathrm{MLE}}^{-1}
\end{bmatrix}
\succeq \mathbf{0}.
\]
By the Schur complement, this implies the MLE is never less efficient than the VI estimator, i.e.,
\[
\Sigma_{\mathrm{MLE}} \ \preceq\ \Sigma_{\mathrm{VI}}.
\]
The equality holds if and only if $\bm{r}_2=A\bm{r}_1$ almost surely for some fixed matrix $A$. In particular, it holds when, for a constant $c>0$,
\[
\mathbb{E}\big[(y - g^{-1}(\tilde{\bx}^\top\bbeta^\star))^2 \mid \tilde{\bx}\big]
= c\cdot(g^{-1})'(\tilde{\bx}^\top\bbeta^\star)
\quad \text{a.s. in } \tilde{\bx},
\]
which is the case for canonical links (up to a constant scale), and for nondegenerate designs this condition is also necessary.
Hence, under correct specification, the MLE dominates the VI estimator in efficiency,
with equality for canonical links. As a result, the VI estimator may exhibit a loss of statistical efficiency relative to the MLE, for which the MLE is asymptotically optimal in the sense of attaining the CRB.
Under misspecification, the comparison changes. If the conditional mean~\eqref{eq:glm_def} is correct but the variance or the full conditional distribution is not, both estimators remain consistent for $\bbeta^\star$, and neither has uniformly smaller asymptotic variance. If the conditional mean or the link is itself misspecified, the VI estimator converges to the solution of the moment equation $V(\bbeta)=\bz$, whereas the MLE converges to the Kullback--Leibler projection induced by the working likelihood.

\section{Proofs for Section \ref{sec:algorithm}}

\subsection*{Proof of Theorem \ref{thm:determine}}

From the update rule of the fixed-point iterative method, the iterates satisfy
\[
\bbeta^{t+1}
= \bbeta^t - \eta\,V_N(\bbeta^t),
\qquad
\hat{\bbeta}_N
= \hat{\bbeta}_N - \eta\,V_N(\hat{\bbeta}_N),
\]
where $\hat{\bbeta}_N$ is a fixed point of the map $\bbeta\mapsto\bbeta-\eta V_N(\bbeta)$ since $V_N(\hat{\bbeta}_N)=\bz$.
It then follows that
\[
\|\bbeta^{t+1} - \hat{\bbeta}_N\|^2
=
\|\bbeta^t - \eta\,V_N(\bbeta^t)
      - \hat{\bbeta}_N + \eta\,V_N(\hat{\bbeta}_N)\|^2.
\]
Expanding the square gives
\[
\|\bbeta^{t+1}-\hat\bbeta_N\|^2
=
\|\bbeta^t-\hat\bbeta_N\|^2
-2\eta\big\langle\bbeta^t-\hat\bbeta_N,\,
  V_N(\bbeta^t)-V_N(\hat\bbeta_N)\big\rangle
+\eta^2\|V_N(\bbeta^t)-V_N(\hat\bbeta_N)\|^2.
\]
By the strong Minty condition in Assumption~\ref{assump} and the Lipschitz continuity in Lemma~\ref{lem:upperbd}, we have
\[
\langle
  V_N(\bbeta^t)-V_N(\hat\bbeta_N),\bbeta^t-\hat\bbeta_N\rangle
\ge \mu\|\bbeta^t-\hat\bbeta_N\|^2,
\quad
\|V_N(\bbeta^t)-V_N(\hat\bbeta_N)\|
\le (L d M^2+L)\|\bbeta^t-\hat\bbeta_N\|.
\]
Substituting these bounds gives
\[
\|\bbeta^{t+1}-\hat\bbeta_N\|^2
\le
\bigl(1 - 2\eta\mu + \eta^2 L^2 (1+dM^2)^2 \bigr)
\|\bbeta^t-\hat\bbeta_N\|^2.
\]
Choosing $\eta = \mu/(L^2 (1+dM^2)^2)$ yields
\[
\|\bbeta^{t+1}-\hat\bbeta_N\|^2
\le
\left(1 - \frac{\mu^2}{L^2 (1+dM^2)^2}\right)
\|\bbeta^t-\hat\bbeta_N\|^2.
\]
By applying the above bound recursively, we obtain the claimed linear convergence rate of Algorithm \ref{alg-det}. 
Combining this result with Theorem~\ref{thm:estimator} further yields with probability at least $1-\epsilon$,
\begin{align*}
\|\bbeta^{t} - \bbeta^{\star}\|
&\le
\|\bbeta^{t} - \hat{\bbeta}_N\| + \|\hat{\bbeta}_N - \bbeta^{\star}\|\\
&\le
\left(1 - \frac{\mu^{2}}{L^{2}(1+dM^2)^2}\right)^{t/2} \|\bbeta^{0} - \hat{\bbeta}_N\|
+ \frac{RM}{\mu} \sqrt{\frac{2(d+1)\ln(2(d+1)/\epsilon)}{N}}.
\end{align*}
The proof is complete.

\subsection*{Proof of Theorem \ref{thm:sgd}}


Since $\bbeta^0$ is deterministic, the growth bound in Lemma~\ref{lem:upperbd} implies by induction that all iterates have finite second moments.
From the update rule, we know that
\[
\begin{aligned}
\|\bbeta^{t+1}-\bbeta^\star\|^2
&=
\|
\bbeta^t-\eta^tV_{(\bx^t,y^t)}(\bbeta^t)-\bbeta^\star\|^2\\
&=
\|\bbeta^t-\bbeta^\star\|^2
-2\eta^t
\langle V_{(\bx^t,y^t)}(\bbeta^t),
\bbeta^t-\bbeta^\star\rangle
+(\eta^t)^2\|V_{(\bx^t,y^t)}(\bbeta^t)\|^2.
\end{aligned}
\]
Taking expectations and using the law of total expectation
yields
\begin{equation}\label{eq:sgd_expansion}
\begin{aligned}
\frac12\mathbb{E}\|\bbeta^{t+1}-\bbeta^\star\|^2
&\leq
\frac12\mathbb{E}\|\bbeta^t-\bbeta^\star\|^2
-\eta^t\mathbb{E}\left[
\langle V_{(\bx^t,y^t)}(\bbeta^t),
\bbeta^t-\bbeta^\star\rangle
\right]
+\frac12(\eta^t)^2
\mathbb{E}\|V_{(\bx^t,y^t)}(\bbeta^t)\|^2.
\end{aligned}
\end{equation}
Since $(\bx^t,y^t)$ is independent of $\bbeta^t$,
the strong Minty condition in Assumption~\ref{assump}
gives
\[
\begin{aligned}
\mathbb{E}\left[
\langle V_{(\bx^t,y^t)}(\bbeta^t),
\bbeta^t-\bbeta^\star\rangle
\right]
&=
\mathbb{E}\left[
\langle V(\bbeta^t),\bbeta^t-\bbeta^\star\rangle
\right]\geq
\mu\mathbb{E}\|\bbeta^t-\bbeta^\star\|^2.
\end{aligned}
\]
In addition, Lemma~\ref{lem:upperbd} implies that
\[
\mathbb{E}\|V_{(\bx^t,y^t)}(\bbeta^t)\|^2
\leq
2(d+1)R^2M^2
+
2L^2(1+dM^2)^2
\mathbb{E}\|\bbeta^t-\bbeta^\star\|^2.
\]
Substituting these bounds into~\eqref{eq:sgd_expansion}
yields
\[
\begin{aligned}
\frac12\mathbb{E}\|\bbeta^{t+1}-\bbeta^\star\|^2
&\leq
\frac12
\left(
1-2\mu\eta^t
+2(\eta^t)^2L^2(1+dM^2)^2
\right)
\mathbb{E}\|\bbeta^t-\bbeta^\star\|^2+(\eta^t)^2(d+1)R^2M^2.
\end{aligned}
\]
Let $\eta^t=1/[\mu(t+1)]$.
Then
\[
\begin{aligned}
\frac12\mathbb{E}\|\bbeta^{t+1}-\bbeta^\star\|^2
&\leq
\frac12
\left(
1-\frac{2}{t+1}
+\frac{2L^2(1+dM^2)^2}{\mu^2(t+1)^2}
\right)
\mathbb{E}\|\bbeta^t-\bbeta^\star\|^2
+\frac{(d+1)R^2M^2}{\mu^2(t+1)^2}.
\end{aligned}
\]
Choose an integer
\[
t_0\geq
\max\left\{
1,\frac{4L^2(1+dM^2)^2}{\mu^2}
\right\}.
\]
For every $t\geq t_0$, we have
\[
\frac{2L^2(1+dM^2)^2}{\mu^2(t+1)^2}
\leq\frac{1}{2(t+1)},
\]
and hence
\[
\begin{aligned}
\frac12\mathbb{E}\|\bbeta^{t+1}-\bbeta^\star\|^2
&\leq
\frac12\left(1-\frac{3}{2(t+1)}\right)
\mathbb{E}\|\bbeta^t-\bbeta^\star\|^2
+\frac{(d+1)R^2M^2}{\mu^2(t+1)^2}.
\end{aligned}
\]

Since the first $t_0+1$ iterates have finite second moments,
we can choose $c_0\geq4$ sufficiently large such that
\[
\frac12\mathbb{E}\|\bbeta^t-\bbeta^\star\|^2
\leq
\frac{c_0(d+1)R^2M^2}{2\mu^2(t+1)},
\qquad t=0,1,\ldots,t_0.
\]
The preceding recursion shows that $c_0$ can be chosen
in terms of the normalized initial error
$\mu^2\|\bbeta^0-\bbeta^\star\|^2/
[(d+1)R^2M^2]$
and $L(1+dM^2)/\mu$.

We now prove by induction that this bound holds for all
$t\geq0$.
The cases $0\leq t\leq t_0$ hold by the choice of $c_0$.
For the induction step, suppose that the bound holds
at some $t\geq t_0$.
Since $1-3/[2(t+1)]\geq0$, we obtain
\[
\begin{aligned}
\frac12\mathbb{E}\|\bbeta^{t+1}-\bbeta^\star\|^2
&\leq
\frac{c_0(d+1)R^2M^2}{2\mu^2(t+1)}
\left(1-\frac{3}{2(t+1)}\right)
+\frac{(d+1)R^2M^2}{\mu^2(t+1)^2}\\
&=
\frac{c_0(d+1)R^2M^2}{2\mu^2}
\left(
\frac{1}{t+1}
-\frac{3/2-2/c_0}{(t+1)^2}
\right)\\
&\leq
\frac{c_0(d+1)R^2M^2}{2\mu^2}
\left(
\frac{1}{t+1}-\frac{1}{(t+1)^2}
\right)\\
&\leq
\frac{c_0(d+1)R^2M^2}{2\mu^2(t+2)},
\end{aligned}
\]
where the second inequality uses $c_0\geq4$.
This completes the induction and thus the proof.

\section{Additional Experimental Details and Results}\label{app_exp}
\subsection{Additional synthetic GLM results}\label{app:GLM_exp}
We provide additional results supplementing the synthetic Poisson GLM experiment in Section \ref{subsec:synthetic}. Tables~\ref{tab:poisson_table_dense} and \ref{tab:poisson_table_sparse} present additional numerical results for the finite-sample performance of VI and MLE under different parameter structures. Table~\ref{tab:poisson_table_dense} corresponds to the dense parameter setting $\bbeta^\star = d^{-1/2}(1,\dots,1) \in \R^d$ as in Section \ref{sec:exp}. The results for the softplus link were already discussed in Table~\ref{tab:poisson_table_softplus}, so here we report the outcomes for the log (sanity check), clipped exponential, and Gaussian-mixture CDF links. Table~\ref{tab:poisson_table_sparse} corresponds to the sparse parameter setting, where $\bbeta^\star = (2/\sqrt{5}, 1/\sqrt5, 0, \dots, 0) \in \R^d$, and summarizes the results for the softplus, clipped exponential, and Gaussian-mixture CDF links (the log link is omitted).

Overall, the results under the sparse parameter setting are qualitatively similar to those under the dense parameter setting. In both cases, the VI estimator consistently outperforms MLE across most configurations, with particularly large improvements for the softplus link.

\begin{table}[!t]
\centering
\caption{Mean squared error of the VI estimator and MLE across link functions and iterations $k$ with dense parameter $\bbeta^\star = d^{-1/2}(1,\dots,1)$. For each $(\text{Link}, k, d, N)$ combination, the smaller error between the two estimators is highlighted in bold. The values in parentheses are standard deviations across 1000 independent repetitions.}
\resizebox{\textwidth}{!}{%
\begin{tabular}{cccccccccc}
\toprule
\multirow{2}{*}{Link, $k$} & \multirow{2}{*}{$d$} & \multicolumn{2}{c}{$N=100$} & \multicolumn{2}{c}{$N=200$} & \multicolumn{2}{c}{$N=500$} & \multicolumn{2}{c}{$N=1000$} \\
 &  & VI & MLE & VI & MLE & VI & MLE & VI & MLE \\
\midrule
\multirow{4}{*}{\makecell[c]{log \\ $k=20$}} & 10 & {.241} (.105) & .241 (.105) & {.108} (.057) & .108 (.057) & {.025} (.014) & .025 (.014) & {.007} (.004) & .007 (.004) \\
 & 20 & {.384} (.106) & .384 (.106) & {.232} (.074) & .232 (.074) & {.079} (.030) & .079 (.030) & {.023} (.009) & .023 (.009) \\
 & 50 & {.589} (.082) & .589 (.082) & {.432} (.071) & .432 (.071) & {.226} (.046) & .226 (.046) & {.100} (.025) & .100 (.025) \\
 & 100 & {.719} (.061) & .719 (.061) & {.586} (.062) & .586 (.062) & {.379} (.049) & .379 (.049) & {.221} (.032) & .221 (.032) \\
\midrule
\multirow{4}{*}{\makecell[c]{log \\ $k=50$}} & 10 & {.128} (.069) & .128 (.069) & {.048} (.027) & .048 (.027) & {.013} (.007) & .013 (.007) & {.006} (.003) & .006 (.003) \\
 & 20 & {.258} (.088) & .258 (.088) & {.124} (.049) & .124 (.049) & {.033} (.012) & .033 (.012) & {.013} (.004) & .013 (.004) \\
 & 50 & {.485} (.080) & .485 (.080) & {.307} (.063) & .307 (.063) & {.120} (.030) & .120 (.030) & {.045} (.011) & .045 (.011) \\
 & 100 & {.655} (.065) & .655 (.065) & {.486} (.060) & .486 (.060) & {.249} (.039) & .249 (.039) & {.118} (.021) & .118 (.021) \\
\midrule
\multirow{4}{*}{\makecell[c]{log \\ $k=100$}} & 10 & {.104} (.059) & .104 (.059) & {.039} (.021) & .039 (.021) & {.012} (.006) & .012 (.006) & {.006} (.003) & .006 (.003) \\
 & 20 & {.221} (.078) & .221 (.078) & {.101} (.039) & .101 (.039) & {.028} (.010) & .028 (.010) & {.012} (.004) & .012 (.004) \\
 & 50 & {.452} (.084) & .452 (.084) & {.269} (.059) & .269 (.059) & {.096} (.025) & .096 (.025) & {.037} (.008) & .037 (.008) \\
 & 100 & {.633} (.064) & .633 (.064) & {.448} (.059) & .448 (.059) & {.214} (.035) & .214 (.035) & {.095} (.017) & .095 (.017) \\
\midrule
\multirow{4}{*}{\makecell[c]{log \\ $k=200$}} & 10 & {.096} (.053) & .096 (.053) & {.037} (.020) & .037 (.020) & {.012} (.006) & .012 (.006) & {.006} (.003) & .006 (.003) \\
 & 20 & {.207} (.075) & .207 (.075) & {.093} (.037) & .093 (.037) & {.027} (.010) & .027 (.010) & {.013} (.004) & .013 (.004) \\
 & 50 & {.442} (.081) & .442 (.081) & {.259} (.057) & .259 (.057) & {.089} (.022) & .089 (.022) & {.036} (.008) & .036 (.008) \\
 & 100 & {.626} (.069) & .626 (.069) & {.439} (.059) & .439 (.059) & {.207} (.035) & .207 (.035) & {.090} (.016) & .090 (.016) \\
\midrule
\multirow{4}{*}{\makecell[c]{clipped exp. \\ $k=20$}} & 10 & \textbf{.593} (.068) & .595 (.066) & \textbf{.484} (.054) & .488 (.051) & \textbf{.342} (.042) & .346 (.036) & \textbf{.241} (.030) & .245 (.024) \\
 & 20 & \textbf{.696} (.052) & .697 (.051) & \textbf{.593} (.047) & .595 (.045) & \textbf{.449} (.037) & .452 (.034) & \textbf{.340} (.028) & .344 (.024) \\
 & 50 & \textbf{.800} (.040) & .801 (.039) & \textbf{.722} (.035) & .722 (.034) & \textbf{.593} (.030) & .595 (.029) & \textbf{.484} (.024) & .486 (.023) \\
 & 100 & \textbf{.865} (.027) & .865 (.027) & \textbf{.802} (.026) & .803 (.025) & \textbf{.693} (.024) & .694 (.023) & \textbf{.594} (.021) & .596 (.020) \\
\midrule
\multirow{4}{*}{\makecell[c]{clipped exp. \\ $k=50$}} & 10 & \textbf{.444} (.084) & .456 (.076) & \textbf{.334} (.066) & .344 (.056) & \textbf{.201} (.048) & .208 (.036) & \textbf{.125} (.033) & .128 (.022) \\
 & 20 & \textbf{.571} (.074) & .580 (.069) & \textbf{.447} (.061) & .457 (.055) & \textbf{.297} (.044) & .306 (.037) & \textbf{.199} (.032) & .206 (.025) \\
 & 50 & \textbf{.724} (.055) & .729 (.051) & \textbf{.608} (.046) & .616 (.042) & \textbf{.448} (.038) & .458 (.034) & \textbf{.332} (.031) & .342 (.026) \\
 & 100 & \textbf{.814} (.039) & .818 (.036) & \textbf{.722} (.038) & .727 (.035) & \textbf{.569} (.033) & .578 (.030) & \textbf{.450} (.026) & .459 (.024) \\
\midrule
\multirow{4}{*}{\makecell[c]{clipped exp. \\ $k=100$}} & 10 & \textbf{.388} (.090) & .403 (.081) & \textbf{.273} (.068) & .285 (.057) & \textbf{.156} (.048) & .161 (.034) & \textbf{.089} (.033) & .090 (.021) \\
 & 20 & \textbf{.519} (.080) & .532 (.072) & \textbf{.394} (.065) & .409 (.057) & \textbf{.239} (.044) & .251 (.037) & \textbf{.154} (.033) & .160 (.024) \\
 & 50 & \textbf{.691} (.060) & .698 (.055) & \textbf{.566} (.057) & .578 (.052) & \textbf{.392} (.041) & .406 (.036) & \textbf{.274} (.032) & .287 (.026) \\
 & 100 & \textbf{.803} (.047) & .804 (.043) & \textbf{.691} (.043) & .698 (.040) & \textbf{.519} (.034) & .532 (.031) & \textbf{.390} (.029) & .404 (.025) \\
\midrule
\multirow{4}{*}{\makecell[c]{clipped exp. \\ $k=200$}} & 10 & \textbf{.377} (.094) & .393 (.083) & \textbf{.255} (.070) & .270 (.058) & \textbf{.141} (.046) & .145 (.033) & \textbf{.080} (.032) & .080 (.020) \\
 & 20 & \textbf{.502} (.084) & .516 (.075) & \textbf{.375} (.068) & .391 (.058) & \textbf{.228} (.047) & .239 (.038) & \textbf{.142} (.034) & .147 (.024) \\
 & 50 & \textbf{.685} (.066) & .692 (.060) & \textbf{.547} (.055) & .560 (.050) & \textbf{.375} (.042) & .391 (.037) & \textbf{.259} (.033) & .272 (.027) \\
 & 100 & .800 (.048) & \textbf{.799} (.043) & \textbf{.680} (.046) & .688 (.041) & \textbf{.507} (.038) & .521 (.034) & \textbf{.376} (.030) & .391 (.026) \\
\midrule
\multirow{4}{*}{\makecell[c]{GMM CDF \\ $k=20$}} & 10 & \textbf{.167} (.061) & .459 (.094) & \textbf{.076} (.030) & .271 (.069) & \textbf{.023} (.010) & .074 (.026) & \textbf{.010} (.004) & .017 (.007) \\
 & 20 & \textbf{.302} (.070) & .610 (.070) & \textbf{.166} (.043) & .456 (.068) & \textbf{.057} (.017) & .209 (.043) & \textbf{.022} (.007) & .073 (.017) \\
 & 50 & \textbf{.513} (.067) & .754 (.045) & \textbf{.348} (.050) & .648 (.044) & \textbf{.164} (.026) & .452 (.043) & \textbf{.076} (.013) & .267 (.033) \\
 & 100 & \textbf{.663} (.053) & .827 (.031) & \textbf{.511} (.048) & .753 (.032) & \textbf{.299} (.031) & .609 (.032) & \textbf{.165} (.019) & .451 (.029) \\
\midrule
\multirow{4}{*}{\makecell[c]{GMM CDF \\ $k=50$}} & 10 & \textbf{.103} (.045) & .200 (.079) & \textbf{.049} (.023) & .079 (.036) & .019 (.009) & \textbf{.017} (.008) & .010 (.004) & \textbf{.007} (.003) \\
 & 20 & \textbf{.210} (.066) & .371 (.086) & \textbf{.105} (.035) & .201 (.053) & \textbf{.039} (.013) & .055 (.018) & .019 (.006) & \textbf{.017} (.005) \\
 & 50 & \textbf{.435} (.074) & .588 (.068) & \textbf{.260} (.051) & .426 (.059) & \textbf{.104} (.022) & .199 (.035) & \textbf{.049} (.010) & .079 (.015) \\
 & 100 & \textbf{.624} (.063) & .715 (.048) & \textbf{.435} (.055) & .586 (.050) & \textbf{.213} (.028) & .370 (.038) & \textbf{.105} (.015) & .199 (.025) \\
\midrule
\multirow{4}{*}{\makecell[c]{GMM CDF \\ $k=100$}} & 10 & \textbf{.098} (.046) & .145 (.066) & \textbf{.048} (.023) & .052 (.025) & .019 (.009) & \textbf{.014} (.006) & .010 (.005) & \textbf{.007} (.003) \\
 & 20 & \textbf{.200} (.065) & .297 (.082) & \textbf{.100} (.033) & .143 (.045) & .039 (.013) & \textbf{.037} (.013) & .019 (.006) & \textbf{.014} (.005) \\
 & 50 & \textbf{.424} (.079) & .525 (.072) & \textbf{.240} (.049) & .344 (.058) & \textbf{.099} (.021) & .141 (.028) & \textbf{.048} (.010) & .051 (.011) \\
 & 100 & \textbf{.625} (.070) & .671 (.056) & \textbf{.419} (.056) & .520 (.051) & \textbf{.197} (.027) & .292 (.036) & \textbf{.098} (.015) & .140 (.020) \\
\midrule
\multirow{4}{*}{\makecell[c]{GMM CDF \\ $k=200$}} & 10 & \textbf{.099} (.050) & .131 (.067) & .049 (.023) & \textbf{.047} (.023) & .020 (.010) & \textbf{.014} (.007) & .010 (.004) & \textbf{.007} (.003) \\
 & 20 & \textbf{.196} (.066) & .273 (.079) & \textbf{.097} (.032) & .127 (.042) & .038 (.013) & \textbf{.034} (.012) & .019 (.006) & \textbf{.014} (.004) \\
 & 50 & \textbf{.421} (.083) & .507 (.077) & \textbf{.241} (.049) & .328 (.057) & \textbf{.098} (.020) & .125 (.027) & .049 (.010) & \textbf{.047} (.010) \\
 & 100 & \textbf{.629} (.075) & .660 (.057) & \textbf{.421} (.056) & .503 (.052) & \textbf{.195} (.029) & .271 (.035) & \textbf{.097} (.015) & .127 (.020) \\
\midrule
\bottomrule
\end{tabular}
}
\label{tab:poisson_table_dense}
\end{table}

\begin{table}[!t]
\centering
\caption{Mean squared error of the VI estimator and MLE across link functions and iterations $k$ with sparse parameter $\bbeta^\star = (2/\sqrt{5}, 1/\sqrt5, 0, \dots, 0)$.}
\resizebox{\textwidth}{!}{%
\begin{tabular}{cccccccccc}
\toprule
\multirow{2}{*}{Link, $k$} & \multirow{2}{*}{$d$} & \multicolumn{2}{c}{$N=100$} & \multicolumn{2}{c}{$N=200$} & \multicolumn{2}{c}{$N=500$} & \multicolumn{2}{c}{$N=1000$} \\
 &  & VI & MLE & VI & MLE & VI & MLE & VI & MLE \\
\midrule
\multirow{4}{*}{\makecell[c]{softplus \\ $k=20$}} & 10 & \textbf{.632} (.078) & .717 (.061) & \textbf{.513} (.069) & .619 (.055) & \textbf{.339} (.054) & .464 (.046) & \textbf{.216} (.037) & .342 (.034) \\
 & 20 & \textbf{.731} (.064) & .795 (.049) & \textbf{.625} (.058) & .711 (.046) & \textbf{.463} (.047) & .577 (.038) & \textbf{.334} (.036) & .461 (.030) \\
 & 50 & \textbf{.832} (.043) & .872 (.032) & \textbf{.757} (.042) & .815 (.032) & \textbf{.631} (.037) & .716 (.028) & \textbf{.508} (.030) & .615 (.024) \\
 & 100 & \textbf{.882} (.031) & .909 (.023) & \textbf{.823} (.029) & .865 (.022) & \textbf{.726} (.028) & .791 (.021) & \textbf{.626} (.027) & .712 (.021) \\
\midrule
\multirow{4}{*}{\makecell[c]{softplus \\ $k=50$}} & 10 & \textbf{.469} (.106) & .570 (.088) & \textbf{.341} (.080) & .455 (.069) & \textbf{.159} (.048) & .268 (.048) & \textbf{.078} (.028) & .161 (.032) \\
 & 20 & \textbf{.585} (.082) & .668 (.066) & \textbf{.459} (.078) & .562 (.066) & \textbf{.276} (.053) & .393 (.048) & \textbf{.164} (.033) & .273 (.033) \\
 & 50 & \textbf{.748} (.058) & .798 (.046) & \textbf{.637} (.062) & .710 (.050) & \textbf{.460} (.052) & .563 (.043) & \textbf{.324} (.034) & .439 (.030) \\
 & 100 & \textbf{.827} (.046) & .858 (.036) & \textbf{.741} (.047) & .792 (.037) & \textbf{.596} (.039) & .677 (.032) & \textbf{.465} (.032) & .567 (.027) \\
\midrule
\multirow{4}{*}{\makecell[c]{softplus \\ $k=100$}} & 10 & \textbf{.400} (.117) & .498 (.102) & \textbf{.258} (.080) & .365 (.075) & \textbf{.115} (.042) & .201 (.046) & \textbf{.048} (.021) & .105 (.027) \\
 & 20 & \textbf{.535} (.103) & .617 (.085) & \textbf{.401} (.083) & .501 (.070) & \textbf{.215} (.051) & .321 (.049) & \textbf{.116} (.035) & .204 (.038) \\
 & 50 & \textbf{.717} (.064) & .766 (.050) & \textbf{.584} (.063) & .659 (.052) & \textbf{.397} (.048) & .498 (.041) & \textbf{.258} (.037) & .365 (.034) \\
 & 100 & \textbf{.818} (.053) & .843 (.042) & \textbf{.706} (.053) & .756 (.043) & \textbf{.539} (.042) & .621 (.035) & \textbf{.394} (.031) & .495 (.027) \\
\midrule
\multirow{4}{*}{\makecell[c]{softplus \\ $k=200$}} & 10 & \textbf{.378} (.111) & .476 (.099) & \textbf{.232} (.079) & .335 (.077) & \textbf{.100} (.037) & .179 (.043) & \textbf{.045} (.019) & .093 (.025) \\
 & 20 & \textbf{.527} (.106) & .607 (.089) & \textbf{.382} (.080) & .479 (.071) & \textbf{.200} (.051) & .300 (.050) & \textbf{.102} (.029) & .182 (.032) \\
 & 50 & \textbf{.707} (.075) & .754 (.061) & \textbf{.574} (.069) & .647 (.057) & \textbf{.375} (.049) & .475 (.043) & \textbf{.241} (.034) & .344 (.032) \\
 & 100 & \textbf{.817} (.057) & .839 (.045) & \textbf{.701} (.058) & .750 (.046) & \textbf{.521} (.045) & .603 (.038) & \textbf{.376} (.036) & .475 (.032) \\
\midrule
\multirow{4}{*}{\makecell[c]{clipped exp. \\ $k=20$}} & 10 & \textbf{.596} (.064) & .598 (.062) & \textbf{.487} (.057) & .491 (.054) & \textbf{.339} (.039) & .343 (.035) & \textbf{.241} (.030) & .244 (.024) \\
 & 20 & \textbf{.694} (.053) & .695 (.051) & \textbf{.595} (.042) & .597 (.040) & \textbf{.444} (.038) & .448 (.035) & \textbf{.341} (.027) & .345 (.024) \\
 & 50 & \textbf{.800} (.037) & .801 (.036) & \textbf{.721} (.036) & .722 (.035) & \textbf{.590} (.028) & .592 (.027) & \textbf{.485} (.025) & .488 (.023) \\
 & 100 & \textbf{.867} (.024) & .868 (.024) & \textbf{.799} (.025) & .799 (.025) & \textbf{.692} (.025) & .693 (.024) & \textbf{.595} (.019) & .597 (.019) \\
\midrule
\multirow{4}{*}{\makecell[c]{clipped exp. \\ $k=50$}} & 10 & \textbf{.447} (.085) & .457 (.077) & \textbf{.325} (.061) & .337 (.051) & \textbf{.203} (.045) & .210 (.034) & \textbf{.126} (.036) & .130 (.023) \\
 & 20 & \textbf{.572} (.067) & .581 (.061) & \textbf{.445} (.063) & .459 (.056) & \textbf{.298} (.045) & .308 (.037) & \textbf{.199} (.031) & .206 (.024) \\
 & 50 & \textbf{.720} (.058) & .726 (.054) & \textbf{.611} (.048) & .619 (.045) & \textbf{.445} (.041) & .455 (.036) & \textbf{.332} (.027) & .342 (.023) \\
 & 100 & \textbf{.813} (.040) & .816 (.037) & \textbf{.724} (.041) & .729 (.038) & \textbf{.566} (.032) & .575 (.029) & \textbf{.446} (.027) & .457 (.024) \\
\midrule
\multirow{4}{*}{\makecell[c]{clipped exp. \\ $k=100$}} & 10 & \textbf{.390} (.093) & .406 (.083) & \textbf{.275} (.068) & .287 (.055) & \textbf{.156} (.045) & .164 (.033) & .091 (.036) & \textbf{.091} (.025) \\
 & 20 & \textbf{.520} (.082) & .533 (.074) & \textbf{.391} (.062) & .407 (.056) & \textbf{.241} (.049) & .251 (.038) & \textbf{.154} (.034) & .159 (.025) \\
 & 50 & \textbf{.684} (.054) & .692 (.051) & \textbf{.567} (.052) & .578 (.048) & \textbf{.389} (.043) & .404 (.037) & \textbf{.278} (.032) & .290 (.027) \\
 & 100 & \textbf{.799} (.046) & .801 (.041) & \textbf{.688} (.045) & .696 (.041) & \textbf{.521} (.036) & .533 (.032) & \textbf{.393} (.030) & .407 (.026) \\
\midrule
\multirow{4}{*}{\makecell[c]{clipped exp. \\ $k=200$}} & 10 & \textbf{.384} (.085) & .396 (.076) & \textbf{.259} (.072) & .273 (.061) & \textbf{.142} (.046) & .148 (.035) & \textbf{.078} (.030) & .080 (.019) \\
 & 20 & \textbf{.506} (.086) & .520 (.078) & \textbf{.368} (.065) & .386 (.057) & \textbf{.230} (.048) & .239 (.038) & \textbf{.142} (.034) & .149 (.024) \\
 & 50 & \textbf{.676} (.063) & .683 (.059) & \textbf{.550} (.054) & .562 (.049) & \textbf{.373} (.040) & .389 (.034) & \textbf{.265} (.032) & .276 (.026) \\
 & 100 & \textbf{.798} (.054) & .798 (.048) & \textbf{.682} (.047) & .690 (.043) & \textbf{.504} (.037) & .518 (.033) & \textbf{.374} (.031) & .390 (.027) \\
\midrule
\multirow{4}{*}{\makecell[c]{GMM CDF \\ $k=20$}} & 10 & \textbf{.171} (.066) & .467 (.097) & \textbf{.075} (.027) & .273 (.068) & \textbf{.022} (.010) & .073 (.023) & \textbf{.009} (.004) & .017 (.007) \\
 & 20 & \textbf{.304} (.072) & .608 (.070) & \textbf{.164} (.041) & .456 (.065) & \textbf{.057} (.017) & .208 (.038) & \textbf{.022} (.007) & .074 (.016) \\
 & 50 & \textbf{.512} (.066) & .756 (.043) & \textbf{.348} (.052) & .647 (.048) & \textbf{.166} (.026) & .455 (.040) & \textbf{.075} (.014) & .266 (.031) \\
 & 100 & \textbf{.666} (.051) & .828 (.030) & \textbf{.515} (.047) & .755 (.034) & \textbf{.302} (.031) & .611 (.031) & \textbf{.165} (.019) & .453 (.031) \\
\midrule
\multirow{4}{*}{\makecell[c]{GMM CDF \\ $k=50$}} & 10 & \textbf{.107} (.054) & .201 (.086) & \textbf{.050} (.023) & .084 (.037) & .019 (.009) & \textbf{.017} (.008) & .009 (.004) & \textbf{.007} (.003) \\
 & 20 & \textbf{.221} (.072) & .384 (.093) & \textbf{.105} (.035) & .199 (.053) & \textbf{.039} (.013) & .054 (.018) & .019 (.006) & \textbf{.017} (.005) \\
 & 50 & \textbf{.421} (.079) & .577 (.071) & \textbf{.259} (.050) & .429 (.058) & \textbf{.104} (.019) & .199 (.033) & \textbf{.048} (.009) & .076 (.014) \\
 & 100 & \textbf{.627} (.070) & .719 (.053) & \textbf{.433} (.057) & .585 (.050) & \textbf{.212} (.027) & .370 (.035) & \textbf{.104} (.015) & .199 (.023) \\
\midrule
\multirow{4}{*}{\makecell[c]{GMM CDF \\ $k=100$}} & 10 & \textbf{.095} (.046) & .134 (.063) & \textbf{.051} (.026) & .055 (.027) & .019 (.008) & \textbf{.014} (.006) & .010 (.005) & \textbf{.006} (.003) \\
 & 20 & \textbf{.201} (.063) & .301 (.087) & \textbf{.098} (.033) & .143 (.050) & .038 (.012) & \textbf{.037} (.012) & .019 (.006) & \textbf{.014} (.004) \\
 & 50 & \textbf{.421} (.081) & .522 (.071) & \textbf{.244} (.049) & .347 (.063) & \textbf{.098} (.020) & .141 (.029) & \textbf{.049} (.011) & .052 (.012) \\
 & 100 & \textbf{.623} (.068) & .670 (.054) & \textbf{.425} (.052) & .522 (.049) & \textbf{.199} (.032) & .288 (.037) & \textbf{.100} (.015) & .142 (.021) \\
\midrule
\multirow{4}{*}{\makecell[c]{GMM CDF \\ $k=200$}} & 10 & \textbf{.096} (.047) & .128 (.064) & .051 (.024) & \textbf{.048} (.023) & .021 (.010) & \textbf{.014} (.007) & .009 (.004) & \textbf{.007} (.003) \\
 & 20 & \textbf{.192} (.060) & .277 (.082) & \textbf{.095} (.033) & .130 (.041) & .039 (.013) & \textbf{.033} (.011) & .020 (.006) & \textbf{.014} (.005) \\
 & 50 & \textbf{.415} (.077) & .499 (.071) & \textbf{.245} (.050) & .331 (.063) & \textbf{.099} (.020) & .131 (.026) & .048 (.010) & \textbf{.046} (.010) \\
 & 100 & \textbf{.629} (.076) & .659 (.059) & \textbf{.426} (.061) & .506 (.056) & \textbf{.196} (.031) & .273 (.037) & \textbf{.098} (.015) & .125 (.019) \\
\midrule
\bottomrule
\end{tabular}
}
\label{tab:poisson_table_sparse}
\end{table}

\subsection{MLE stalling under saturated initialization}
\label{app:clipped_init}

The clipped-exponential link provides a simple example where likelihood-based optimization can become completely flat. In the main synthetic experiments of Section~\ref{subsec:synthetic}, the default initialization does not necessarily place the iterates inside the saturated region, and therefore the clipped-exponential experiment does not always show complete stalling at early iterations. To directly isolate this failure mode, we consider an additional stress test for the Poisson model with
\[
    g^{-1}(z)=\min\{\exp(z),C\}, \qquad C=2.
\]
We use \(d=20\) non-intercept features and \(N=400\) samples. The true parameter has zero intercept and dense unit-norm slope,
\[
    \bbeta^\star = (0, d^{-1/2},\ldots,d^{-1/2}) \in \mathbb{R}^{d+1}.
\]
The covariates are drawn as \(\bx^i\sim \mathcal{N}(0,I_d)\), and we augment them with an intercept coordinate. This intercept coordinate is used only to construct a saturated initialization. Specifically, both VI and MLE are initialized at
\[
    \bbeta^0 = (2.0,\, 0.25\,\bbeta^\star_{2:d+1}),
\]
which pushes the initial predictors \(z_i^0=\langle \bbeta^0,\tilde \bx^i\rangle\) beyond the clipping threshold \(\log C\) for essentially all samples.

This setting exposes a qualitative difference between the VI field and the MLE gradient. For the clipped-exponential link, \( (g^{-1})'(z)=0 \) whenever \(z> \log C\). Hence, when the iterates are initialized inside the saturated region, the MLE gradient
\[
    \nabla \mathcal{L}_N(\bbeta)
    =
    \frac{1}{N}\sum_{i=1}^N
    \big(g^{-1}(\tilde \bx^{i\top}\bbeta)-y^i\big)
    \frac{(g^{-1})'(\tilde \bx^{i\top}\bbeta)}{g^{-1}(\tilde \bx^{i\top}\bbeta)}
    \tilde \bx^i
\]
vanishes. In contrast, the VI field
\[
    V_N(\bbeta)
    =
    \frac{1}{N}\sum_{i=1}^N
    \big(g^{-1}(\tilde \bx^{i\top}\bbeta)-y^i\big)\tilde \bx^i
\]
does not contain this derivative factor and can remain nonzero even in the saturated region.

Figure~\ref{fig:clipped_init} illustrates this behavior. The top-left panel confirms that the initial predictors lie in the clipped region. The top-right panel shows that MLE remains essentially at its initial point, while VI continues to reduce the squared parameter error. The bottom-left panel shows that the MLE gradient norm is zero throughout the trajectory, whereas the VI field norm remains nonzero. The bottom-right panel confirms the same qualitative behavior averaged over 200 independent replications. This experiment demonstrates that the plateau pathology is not merely a one-dimensional illustration. In a multidimensional Poisson GLM with a clipped-exponential link, MLE can become completely trapped under a saturated initialization, while VI can still move.

\begin{figure}[t]
    \centering
    \includegraphics[width=\linewidth]{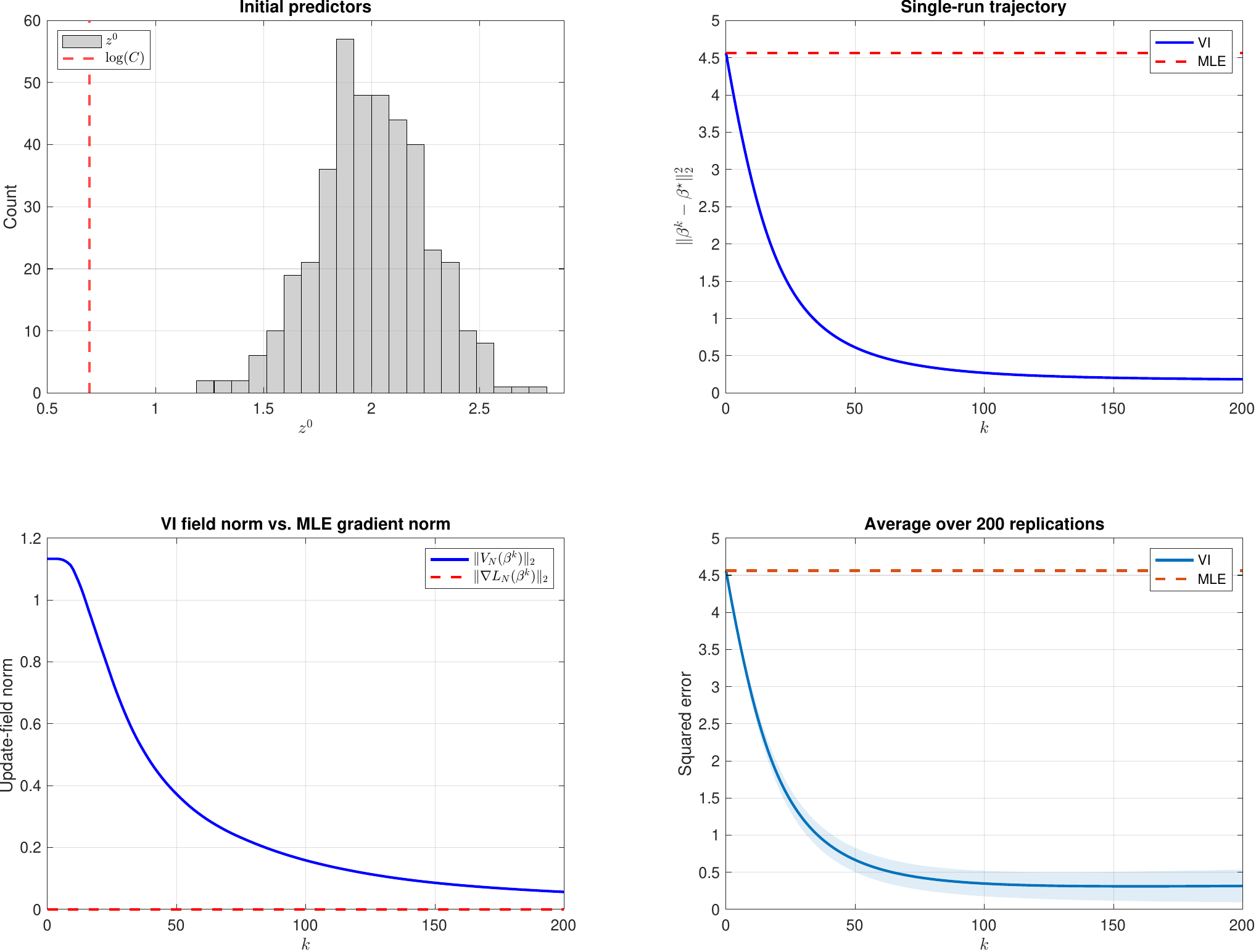}
    \caption{
    Saturated-initialization stress test for the clipped-exponential Poisson model. 
    Both VI and MLE are initialized in the clipped region. 
    Since the clipped-exponential derivative vanishes in this region, the MLE gradient is zero and the MLE trajectory remains trapped near its initial point. 
    In contrast, the VI field remains nonzero and the VI iterate continues to reduce the squared parameter error. 
    The averaged trajectory over 200 replications shows the same qualitative behavior.
    }
    \label{fig:clipped_init}
\end{figure}

\subsection{Synthetic GAM reconstruction}\label{app:GAM_exp}
\begin{figure}[!t]
    \centering
    \begin{subfigure}[t]{0.49\linewidth}
        \centering
        \includegraphics[width=\linewidth]{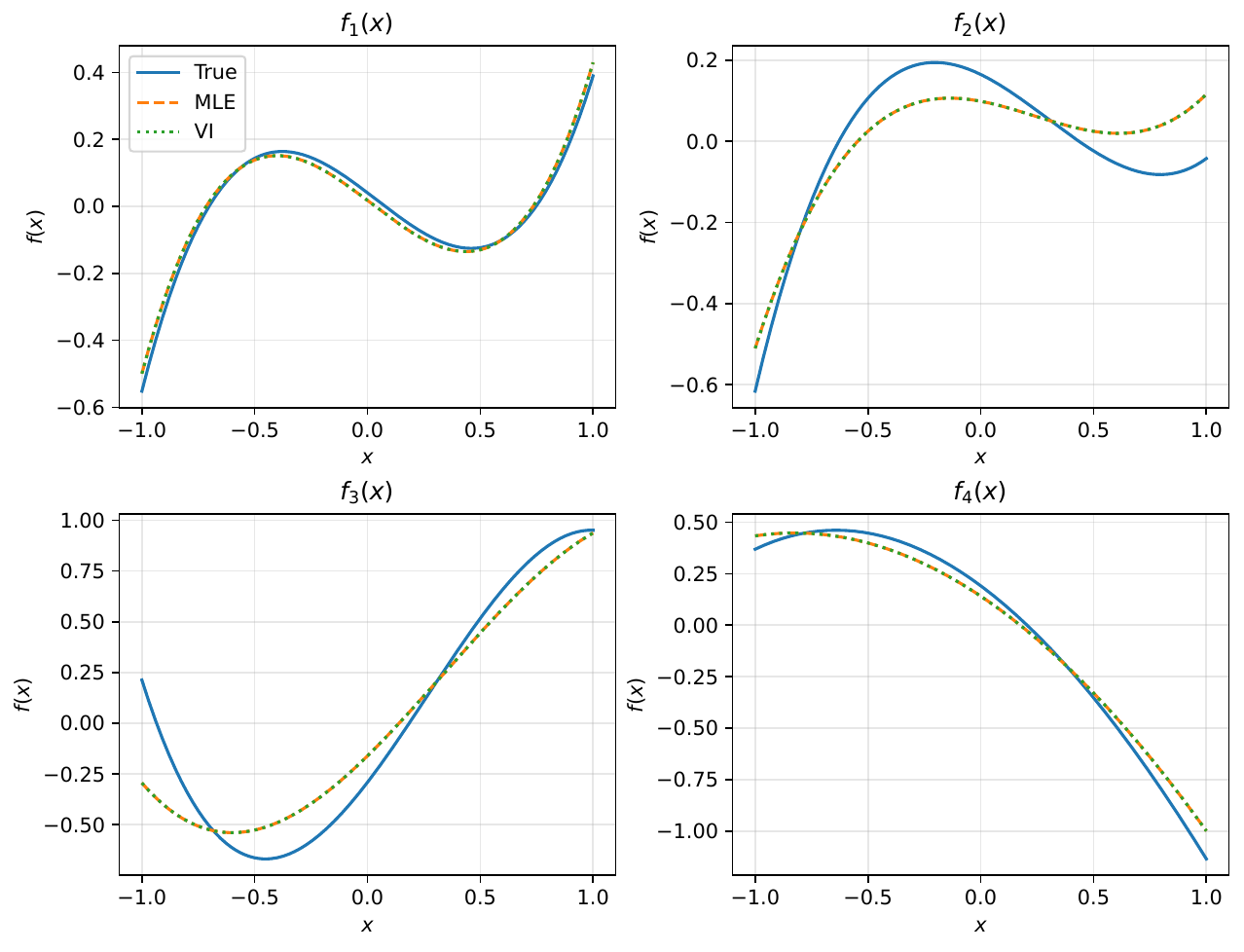}
        \caption{log}
    \end{subfigure}\hfill
    \begin{subfigure}[t]{0.49\linewidth}
        \centering
        \includegraphics[width=\linewidth]{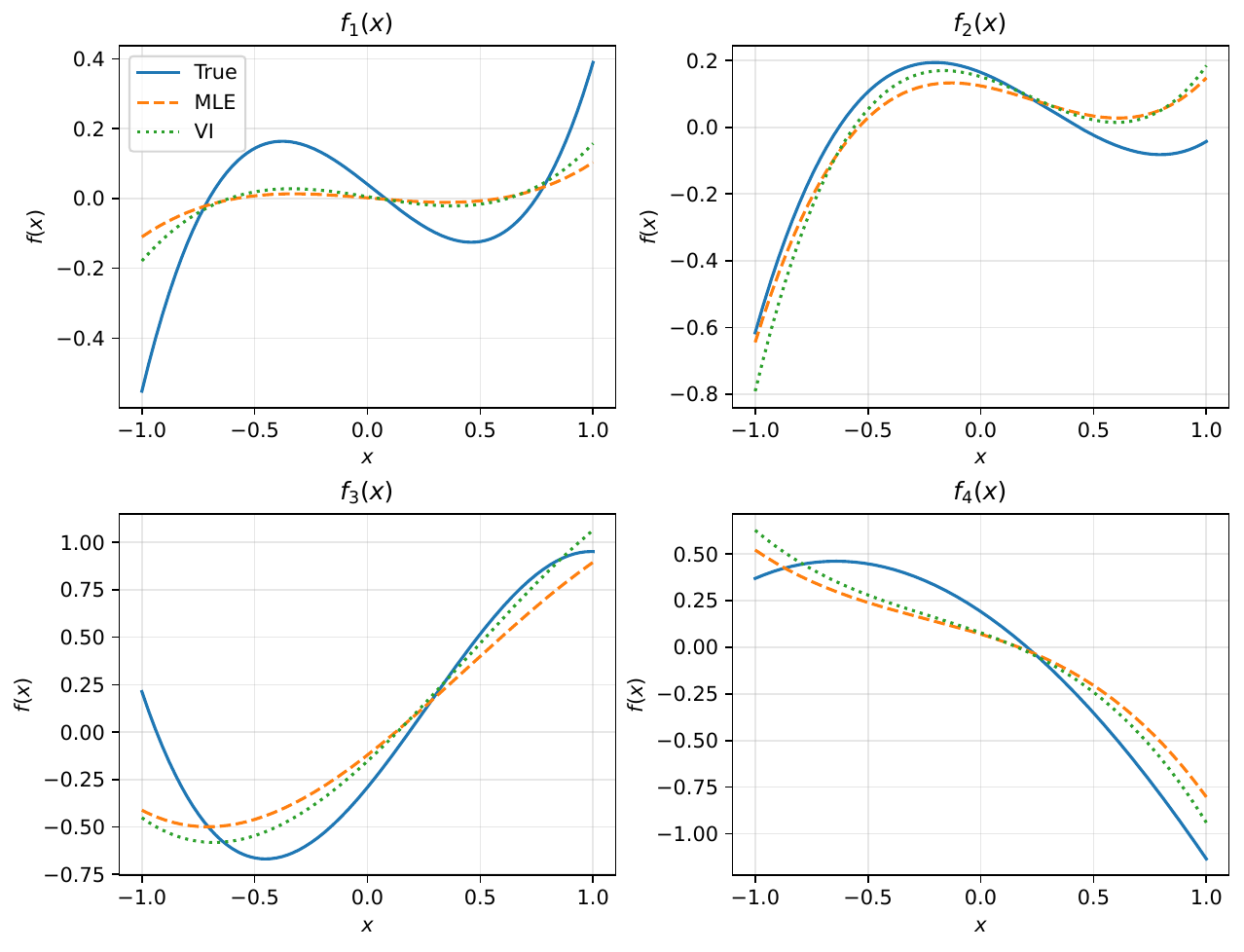}
        \caption{softplus}
    \end{subfigure}\hfill
    \begin{subfigure}[t]{0.49\linewidth}
        \centering
        \includegraphics[width=\linewidth]{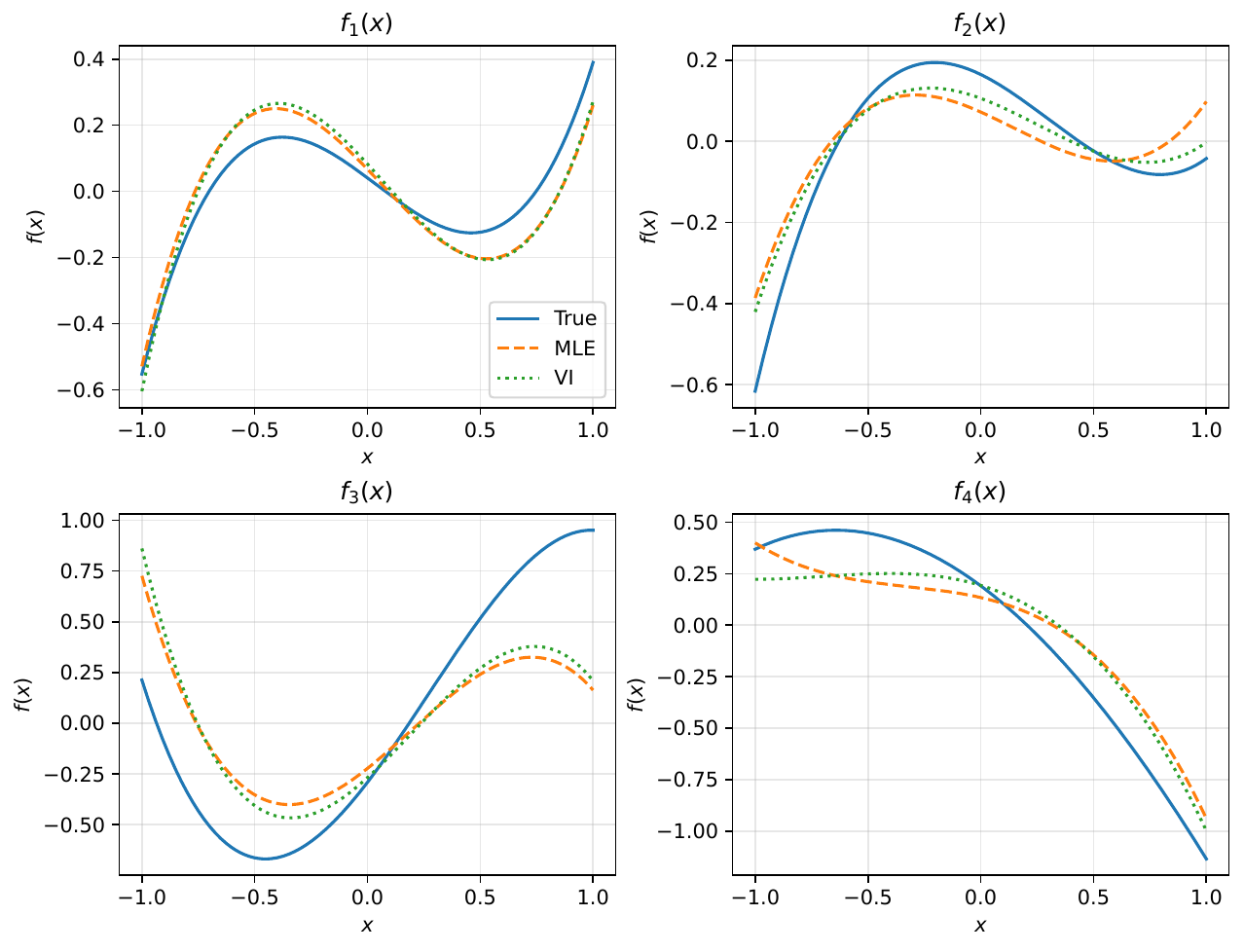}
        \caption{clipped exp.}
    \end{subfigure}\hfill
    \begin{subfigure}[t]{0.49\linewidth}
        \centering
        \includegraphics[width=\linewidth]{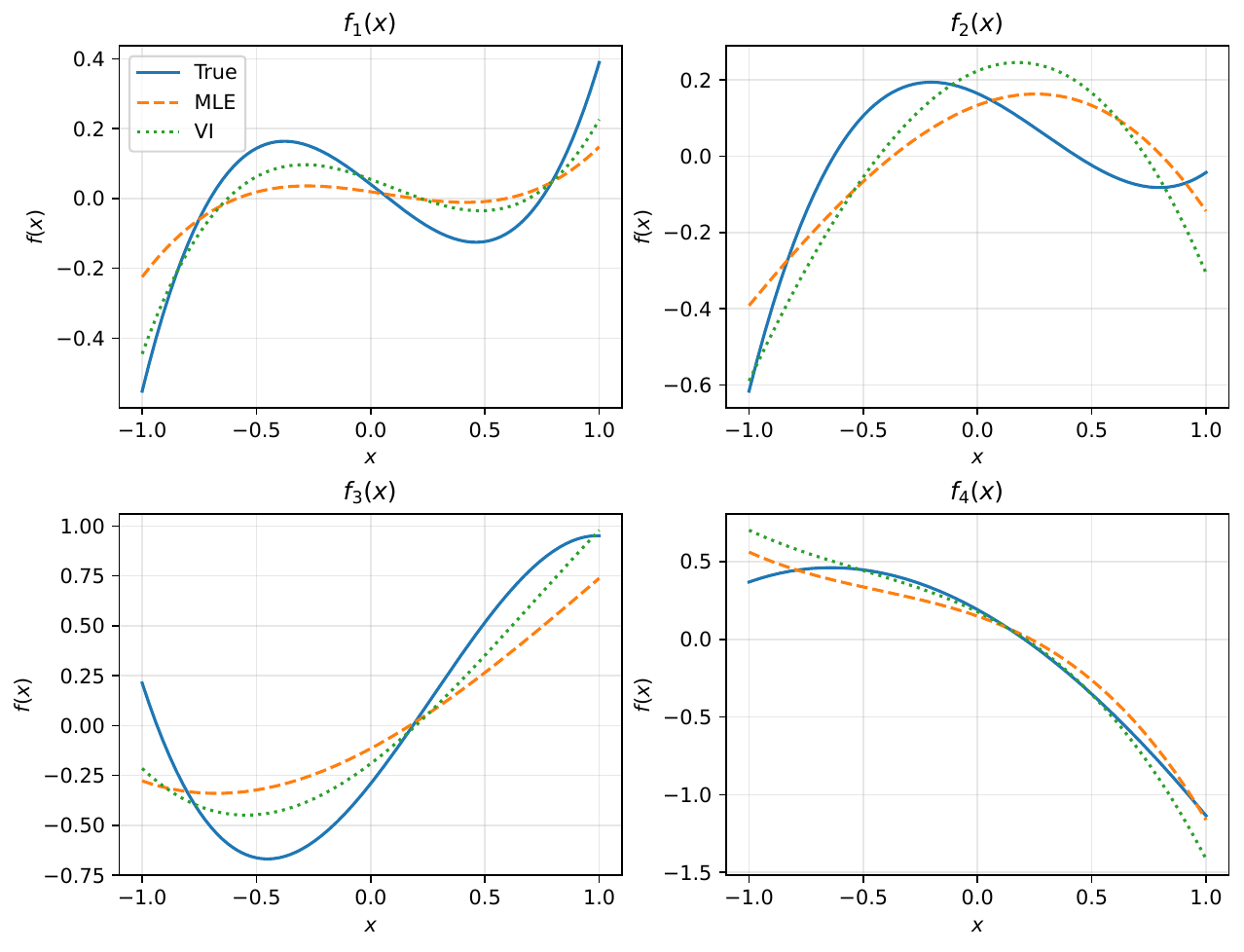}
        \caption{GMM CDF}
    \end{subfigure}
    \caption{Reconstruction of additive components in GAMs under Poisson distribution.}
    \label{fig:recon_gam_poisson}
\end{figure}

\begin{figure}[!t]
    \centering
    \begin{subfigure}[t]{0.49\linewidth}
        \centering
        \includegraphics[width=\linewidth]{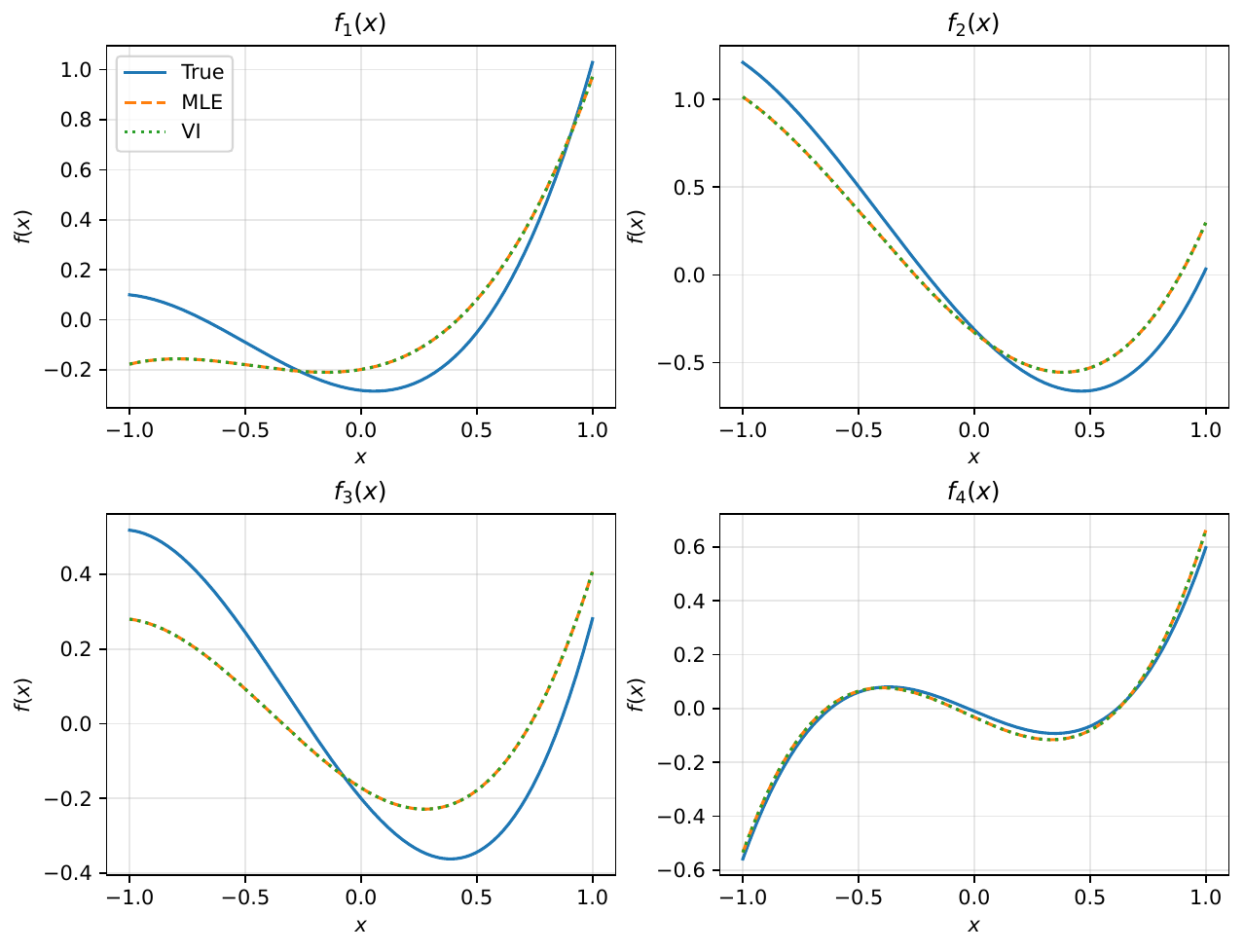}
        \caption{logistic}
    \end{subfigure}\hfill
    \begin{subfigure}[t]{0.49\linewidth}
        \centering
        \includegraphics[width=\linewidth]{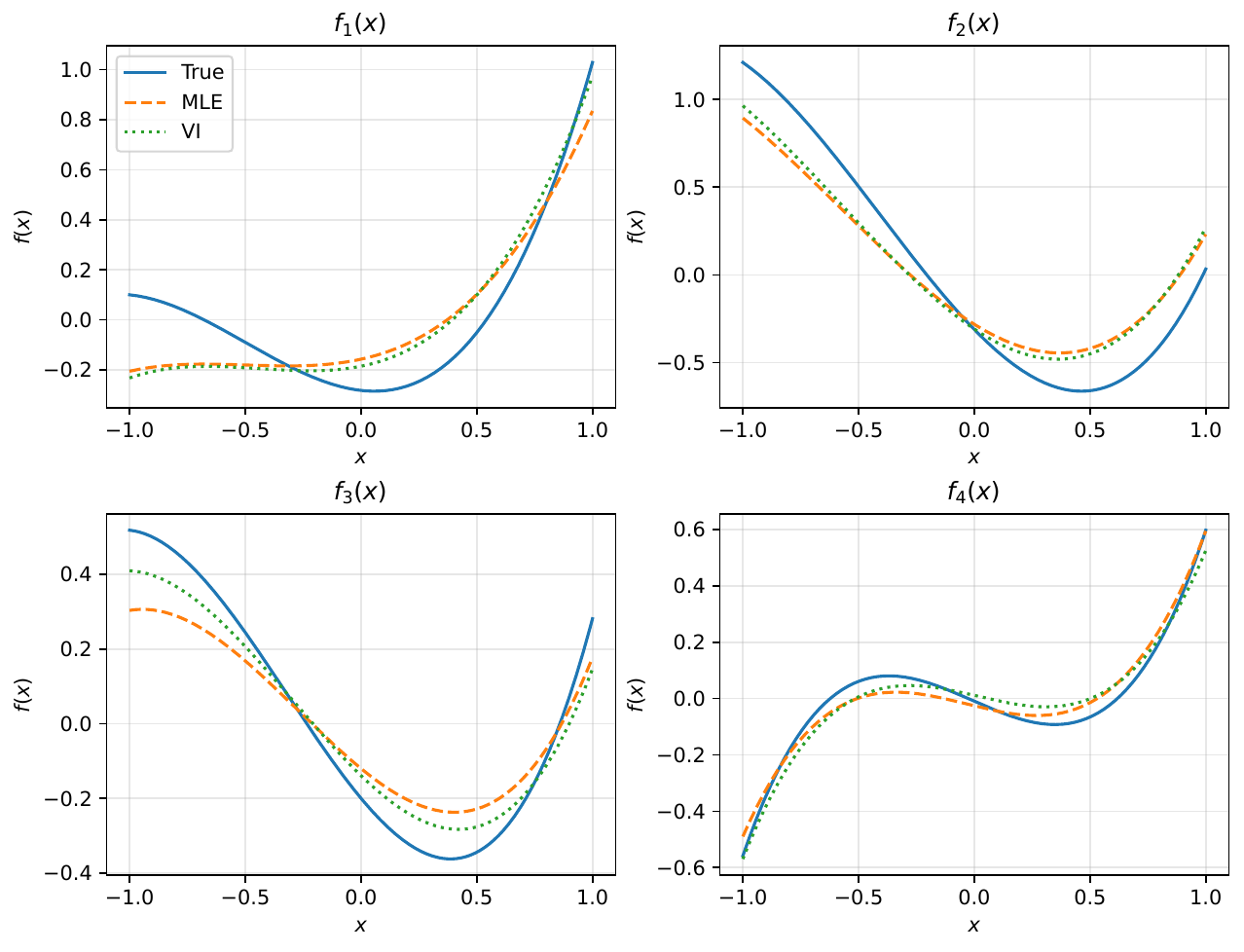}
        \caption{arctangent}
    \end{subfigure}\hfill
    \caption{Reconstruction of additive components in GAMs under Bernoulli distribution.}
    \label{fig:recon_gam_bernoulli}
\end{figure}

\begin{figure}[!t]
    \centering
    \begin{subfigure}[t]{0.49\linewidth}
        \centering
        \includegraphics[width=\linewidth]{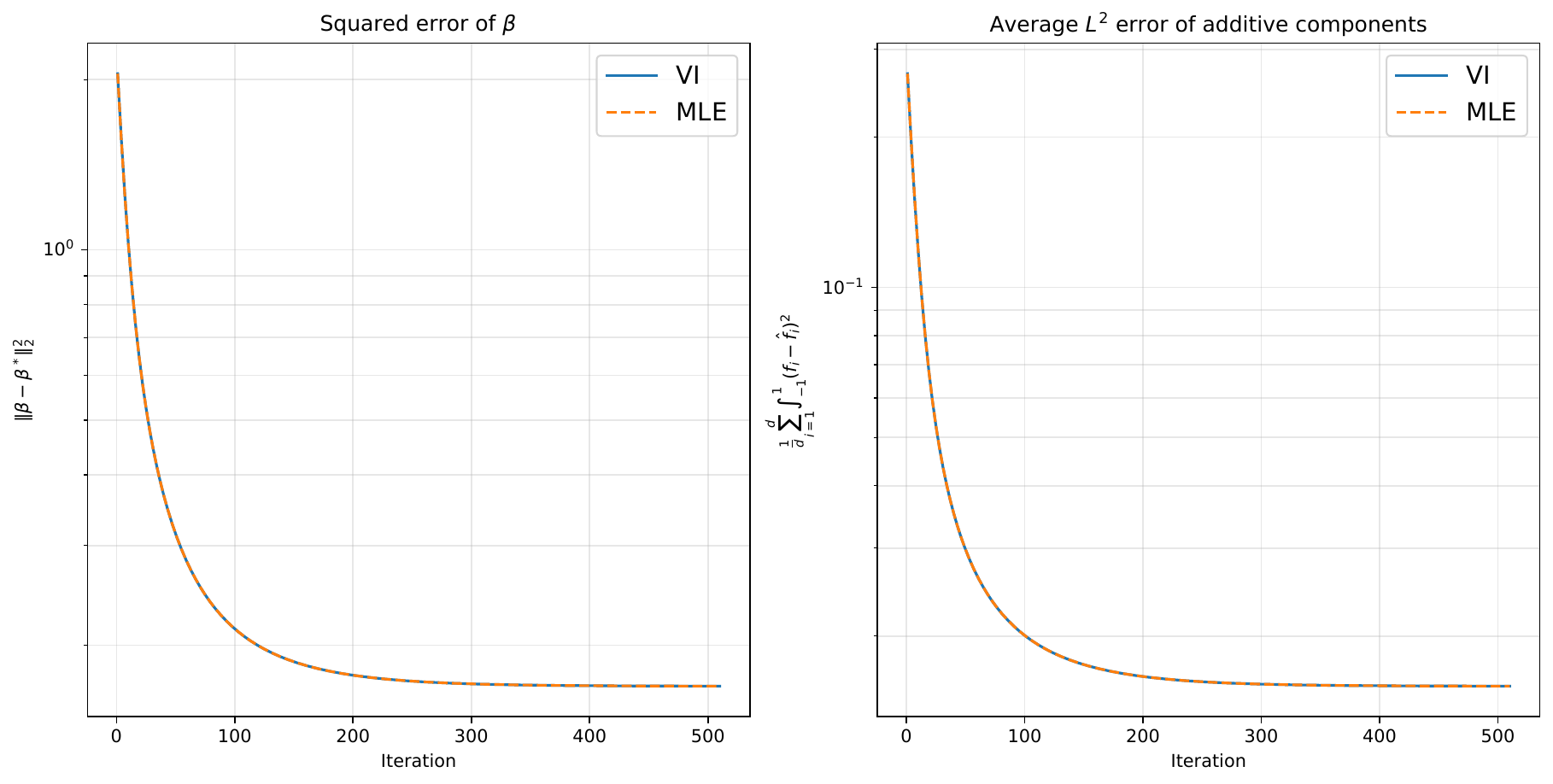}
        \caption{log}
    \end{subfigure}\hfill
    \begin{subfigure}[t]{0.49\linewidth}
        \centering
        \includegraphics[width=\linewidth]{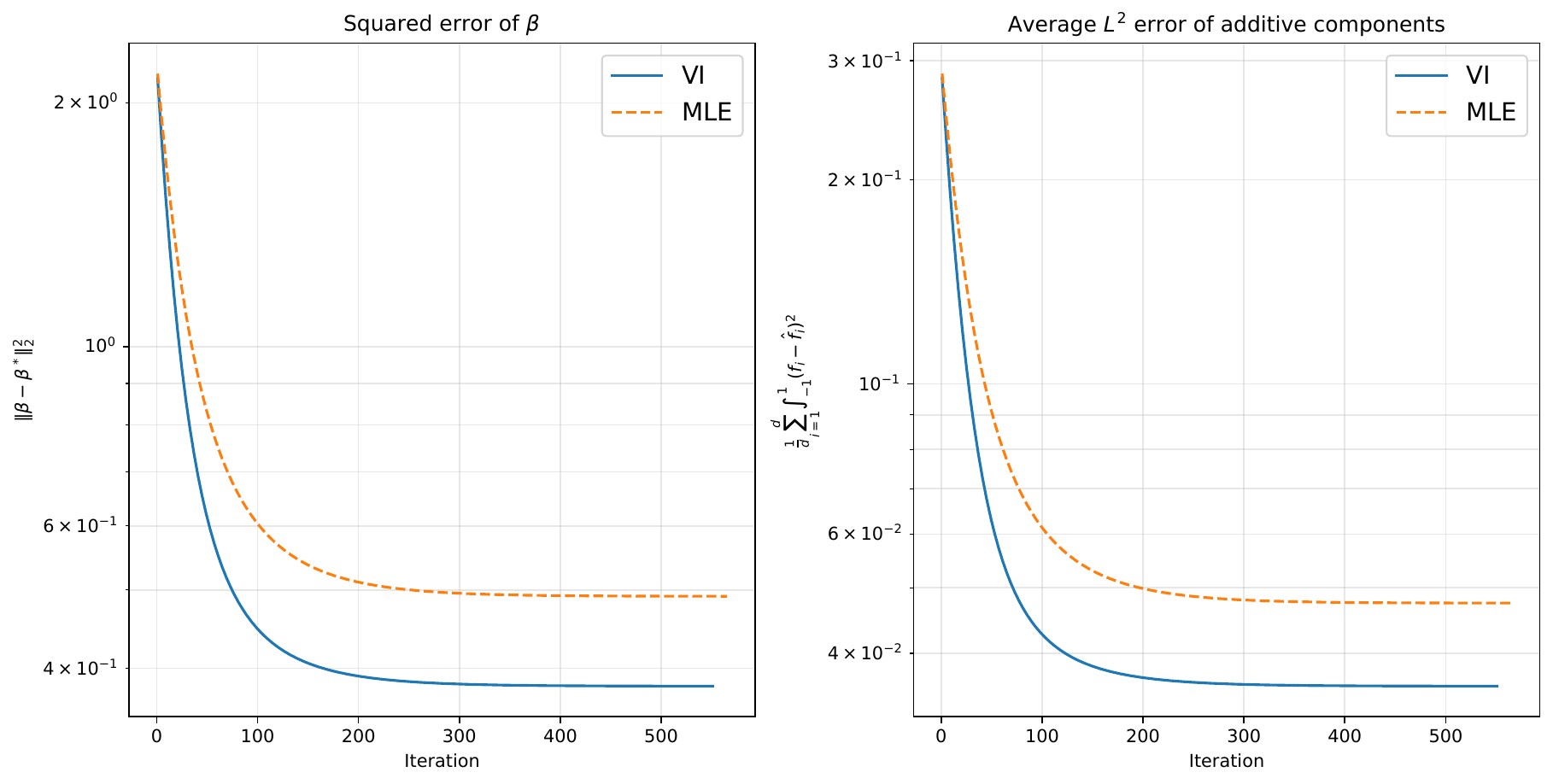}
        \caption{softplus}
    \end{subfigure}\hfill
    \begin{subfigure}[t]{0.49\linewidth}
        \centering
        \includegraphics[width=\linewidth]{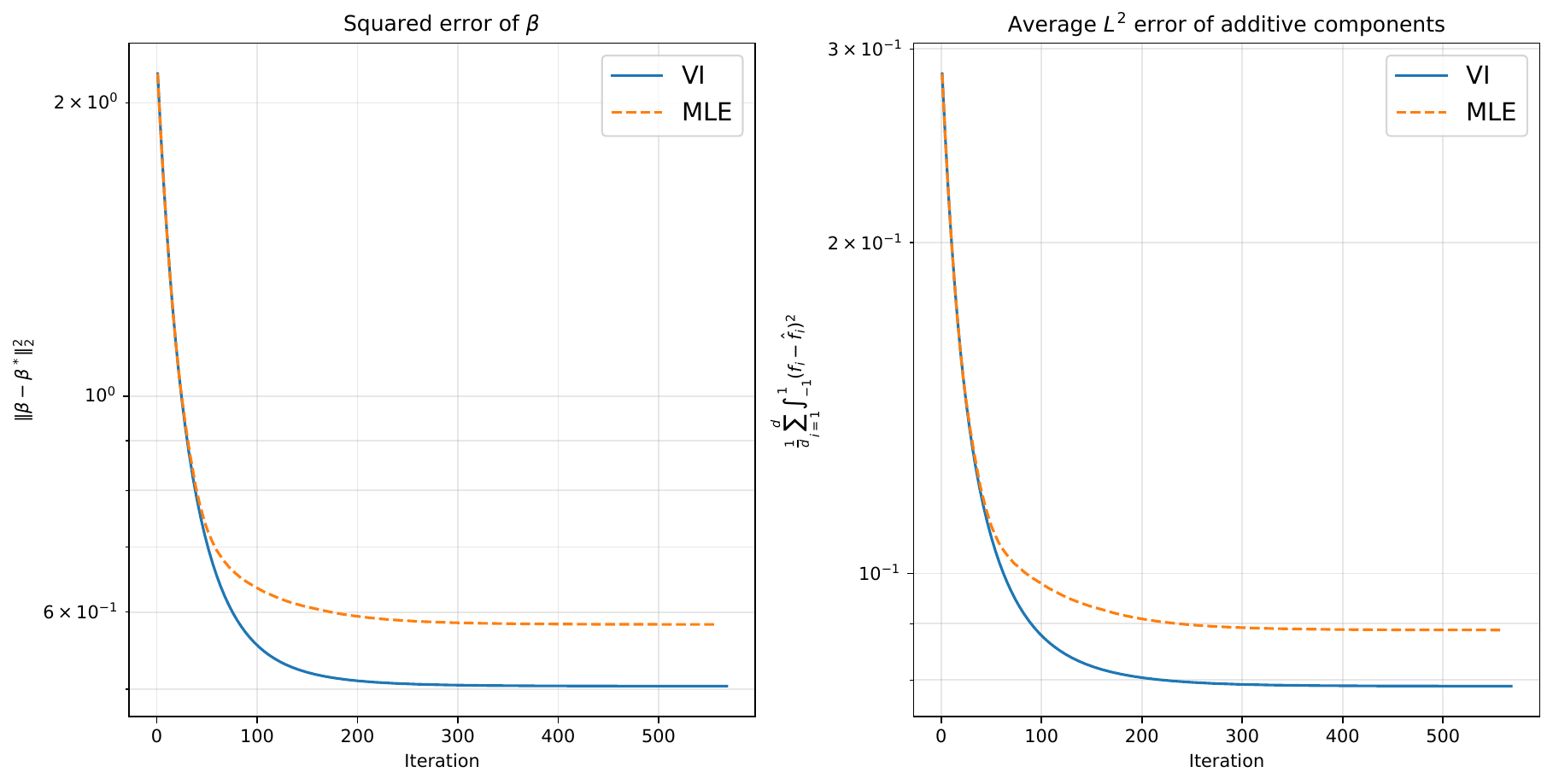}
        \caption{clipped exp.}
    \end{subfigure}\hfill
    \begin{subfigure}[t]{0.49\linewidth}
        \centering
        \includegraphics[width=\linewidth]{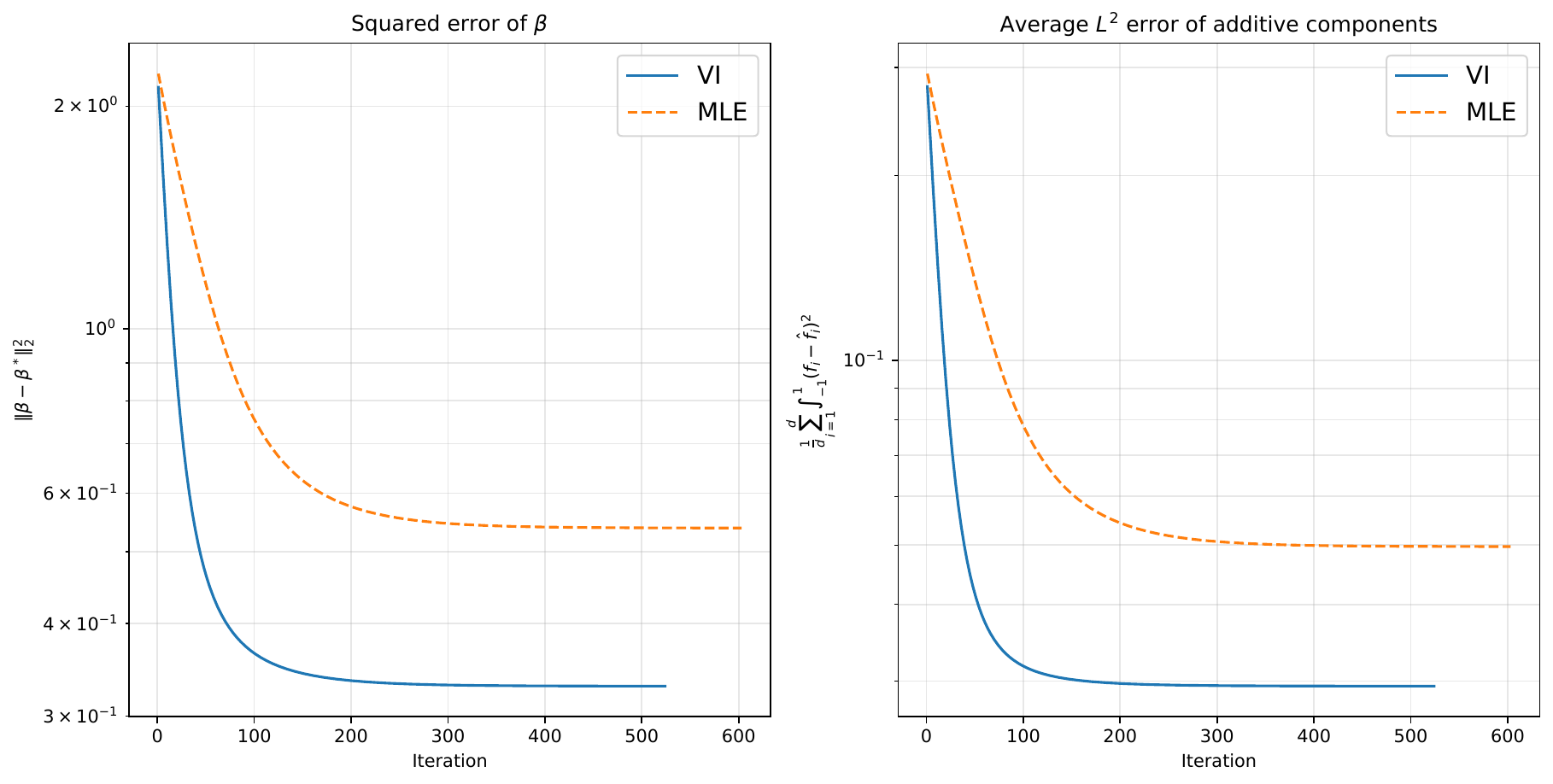}
        \caption{GMM CDF}
    \end{subfigure}
    \caption{Squared error of parameter estimates and the average $L^2$ error of the additive components (Poisson response).}
    \label{fig:traj_gam_poisson}
\end{figure}

\begin{figure}[!t]
    \centering
    \begin{subfigure}[t]{0.49\linewidth}
        \centering
        \includegraphics[width=\linewidth]{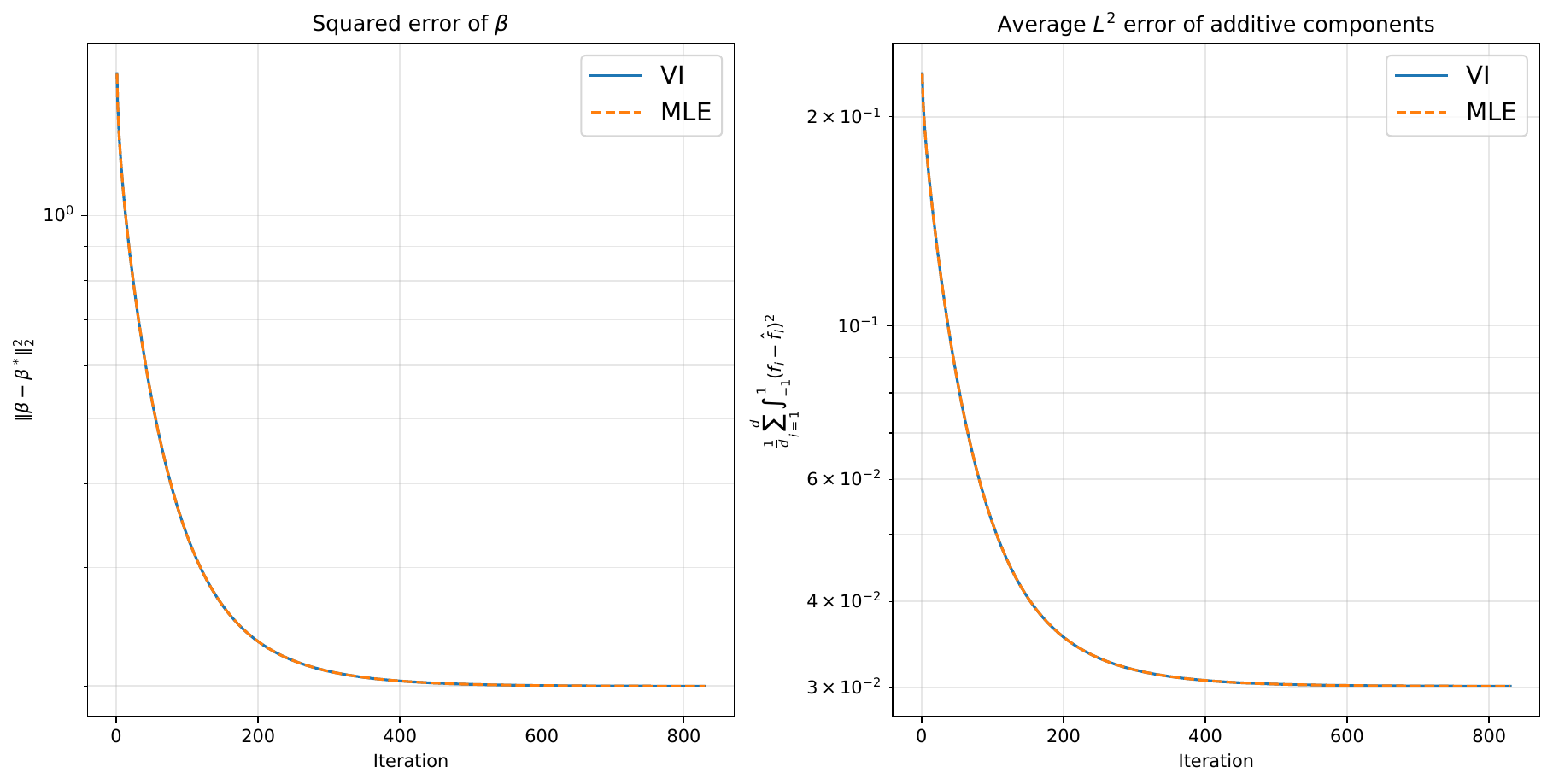}
        \caption{logistic}
    \end{subfigure}\hfill
    \begin{subfigure}[t]{0.49\linewidth}
        \centering
        \includegraphics[width=\linewidth]{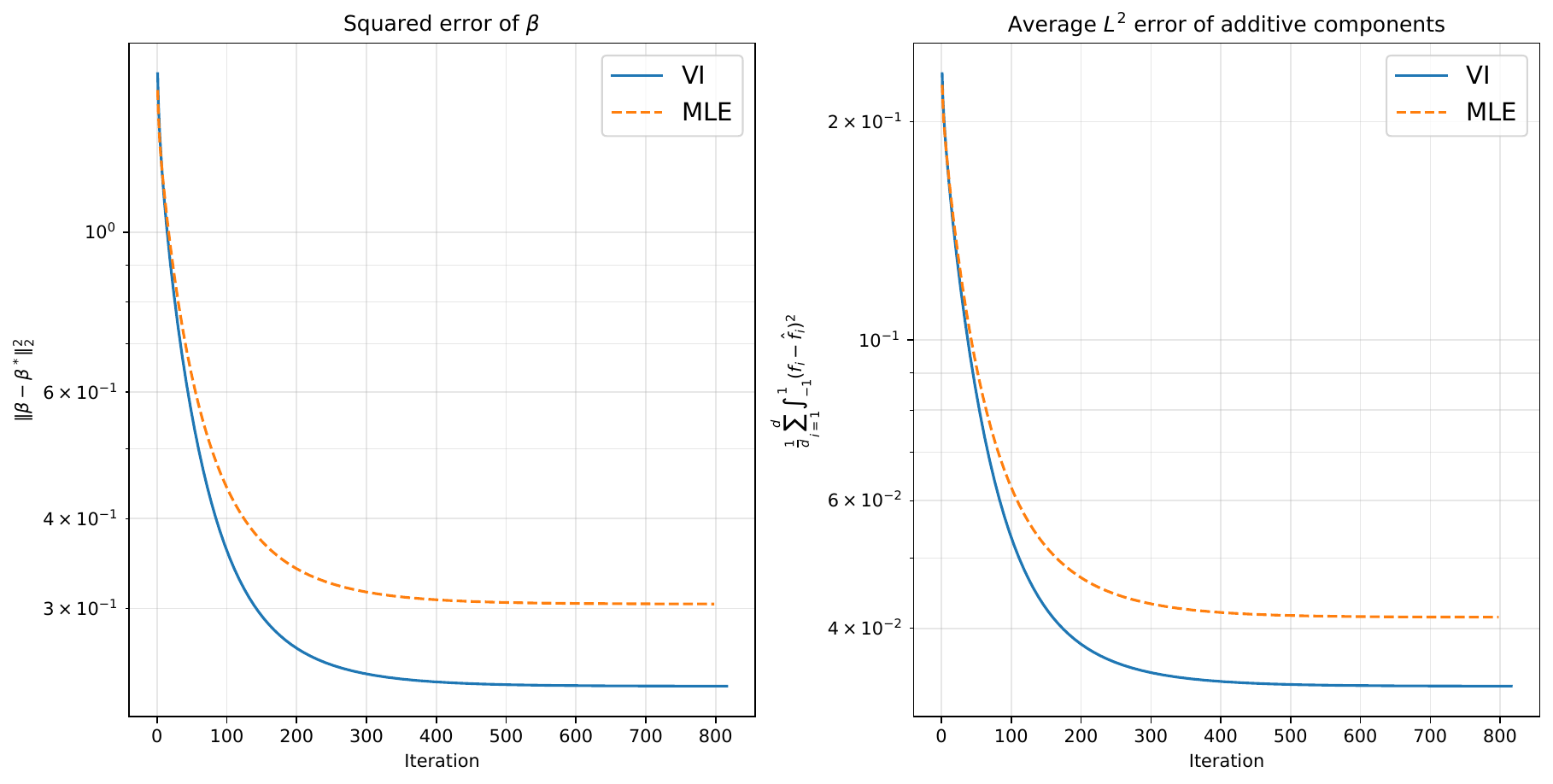}
        \caption{arctangent}
    \end{subfigure}\hfill
    \caption{Squared error of parameter estimates and the average $L^2$ error of the additive components (Bernoulli response).}
    \label{fig:traj_gam_bernoulli}
\end{figure}

We construct the GAM reconstruction experiments considered in
Section~\ref{sec:gam}.
We consider the Poisson and Bernoulli response models.
Each additive component $f_j$ is represented using a truncated
Legendre basis of degree $K$, with the coefficient of
$P_0\equiv1$ fixed to zero.
This gives the reference-interval centering convention
$\int_{-1}^{1}f_j(u)\,du=0$, rather than the pointwise constraint
$f_j(0)=0$.
We use $d=4$ additive components and $K=3$ basis functions per
component, with covariates independently sampled from $U[-1,1]$.
Each $\beta_{jk}$ is drawn from a standard normal distribution,
and the coefficient vector is then scaled so that
$\|\bbeta\|_2=1.5$.

For Poisson models, we consider the canonical log link as well as softplus, clipped exponential, and Gaussian-mixture CDF links. For Bernoulli models, we consider the canonical logistic link together with the non-canonical arctangent link, whose inverse is given by
\[
    g^{-1}(z) = \frac{1}{2} + \frac{1}{\pi} \arctan(z).
\]

Figures~\ref{fig:recon_gam_poisson} and~\ref{fig:recon_gam_bernoulli} compare the recovered additive components with the ground truths, while Figures~\ref{fig:traj_gam_poisson} and~\ref{fig:traj_gam_bernoulli} report the squared error of the parameter estimates and the average $L^2$ errors of the additive components over iterations for the Poisson and Bernoulli models, respectively. For non-canonical links in both Poisson and Bernoulli GAMs, the VI estimator often yields slightly smoother optimization trajectories and modestly improved reconstruction of the additive components compared to MLE, although the final errors are typically of similar magnitude. These experiments illustrate that the proposed VI framework extends naturally to GAMs through basis representations and remains numerically stable in additive models beyond linear predictors.

\subsection{Real-data bike-sharing benchmark}\label{app:realdata}
We provide additional details for the real-data experiment in Section~\ref{subsec:realdata}. We compare the log, softplus, and clipped exponential links
\[
    g^{-1}(z)=\exp(z),\quad
    g^{-1}(z)=\log(1+\exp(z)),\quad
    g^{-1}(z)=\min\{\exp(z),C\},
\]
where the cap \(C\) for the clipped exponential link is selected by validation Poisson deviance.

The chronological train/validation/test split contains 10,427/3,475/3,477 observations. For the fixed-budget softplus fits, the step size and checkpoint are selected by validation deviance from 250-step trajectories. Mean Poisson deviance uses the predicted-mean floor $\max\{\mu_i, 10^{-12}\}$.
These data and metric conventions also apply to the converged comparisons in Appendix~\ref{app:converged}.

\begin{figure}[!t]
    \centering
    \begin{subfigure}[t]{0.45\linewidth}
        \centering
        \includegraphics[width=\linewidth]{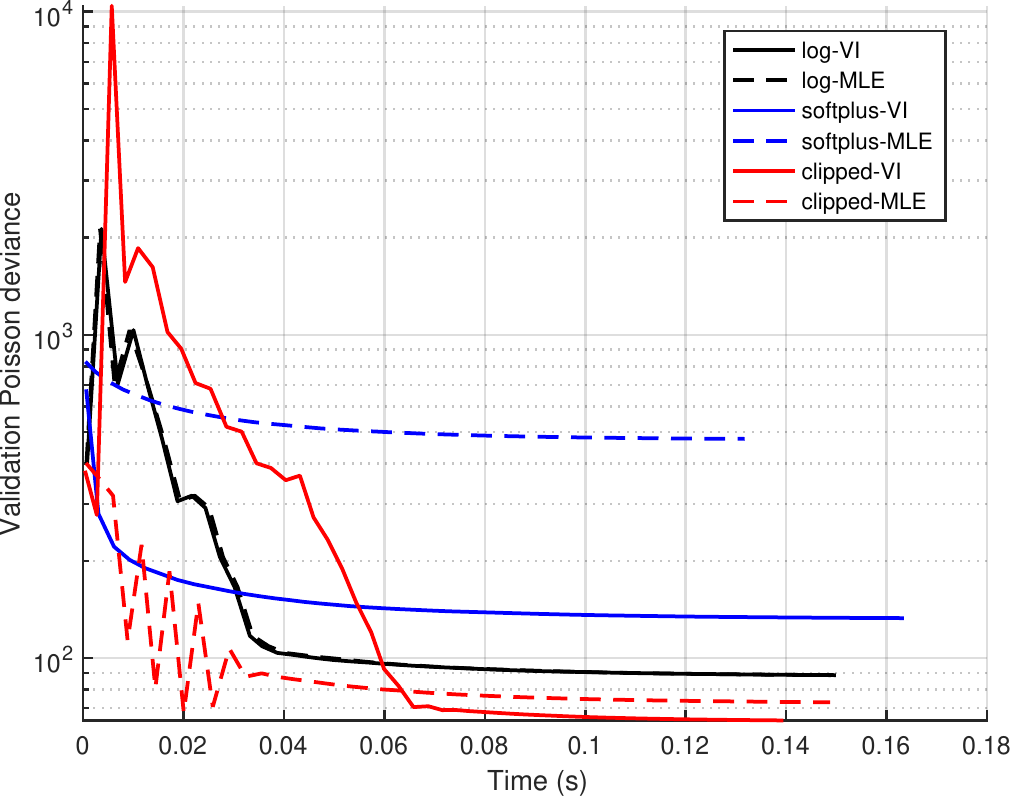}
        \caption{GLM}
    \end{subfigure}\hfill
    \begin{subfigure}[t]{0.45\linewidth}
        \centering
        \includegraphics[width=\linewidth]{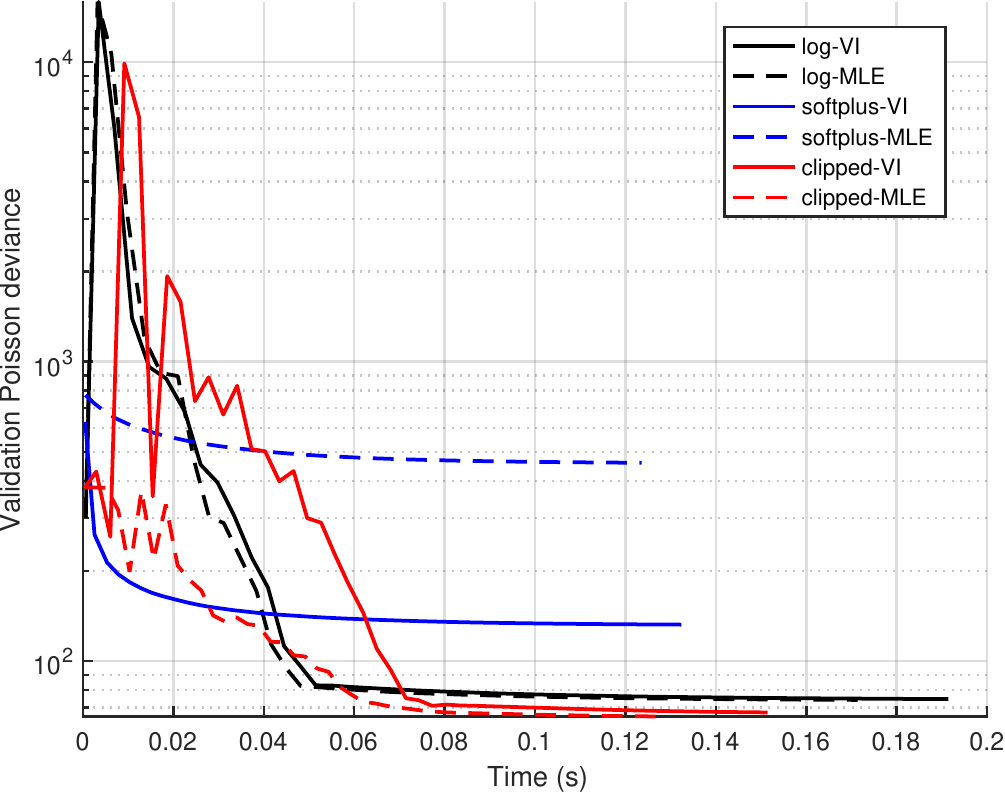}
        \caption{GAM, $K=3$}
    \end{subfigure}
    \caption{Runtime versus validation Poisson deviance.}
    \label{fig:real_runtime}
\end{figure}

\begin{table}[!t]
\centering
\caption{Test metrics on the bike-sharing dataset.}
\small
\begin{tabular}{llcccccc}
\toprule
& & \multicolumn{2}{c}{Test Dev.} & \multicolumn{2}{c}{RMSE} & \multicolumn{2}{c}{MAE} \\
Model & Link & VI & MLE & VI & MLE & VI & MLE  \\
\midrule
GLM & log & {59.43} & 59.43 & {125.30} & 125.30 & {92.45} & 92.45 \\
GLM & softplus & \textbf{133.99} & 453.85 & \textbf{187.78} & 280.89 & \textbf{135.61} & 201.89  \\
GLM & clipped-exp & \textbf{51.90} & 59.84 & \textbf{128.16} & 128.17 & \textbf{88.65} & 94.05  \\
\midrule
GAM, $K=3$ & log & {55.57} & 56.61 & {125.65} & 126.88 & {88.97} & 89.82  \\
GAM, $K=3$ & softplus & \textbf{129.48} & 428.88 & \textbf{185.03} & 277.65 & \textbf{132.33} & 199.37 \\
GAM, $K=3$ & clipped-exp & 55.32 & \textbf{50.15} & 133.50 & \textbf{125.49} & 93.37 & \textbf{86.64} \\
\bottomrule
\end{tabular}
\label{tab:realdata_main}
\end{table}

\begin{figure}[!t]
    \centering
    \includegraphics[width=0.65\linewidth]{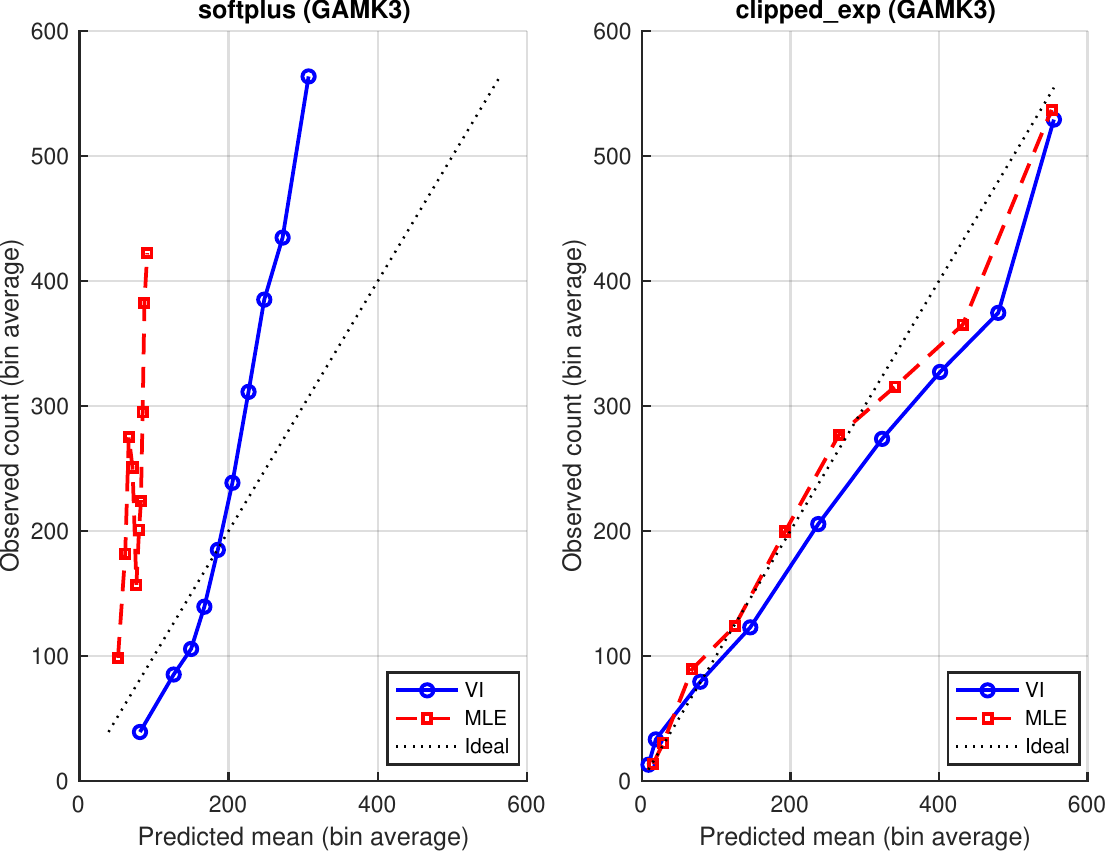}
    \caption{Test calibration plots (GAM, $K=3$). The diagonal line denotes perfect calibration.}
    \label{fig:real_calibration_gam}
\end{figure}

Figure~\ref{fig:real_runtime} shows runtime versus validation Poisson deviance for the GLM and GAM designs. Under the log link, VI and MLE have nearly identical trajectories.  Under the softplus link, VI reaches a much lower validation deviance than MLE at comparable runtime. The clipped exponential link is more mixed.

Table~\ref{tab:realdata_main} reports the test Poisson deviance, RMSE, and MAE. The softplus link shows the clearest advantage for VI in both GLM and GAM designs. For the clipped exponential link, VI is better in the GLM design, while MLE is slightly better in the GAM design.

Figure~\ref{fig:real_calibration_gam} reports test-set calibration for the GAM design under the softplus and clipped exponential links. For each bin, the $x$-axis shows the average predicted mean and the $y$-axis shows the average observed count. The softplus case shows a more stable calibration pattern for VI, whereas the clipped exponential case is more comparable between the two methods.

\subsection{Numerically converged softplus comparisons}\label{app:converged}

We compare fixed-budget and numerically converged estimators under the softplus link, for both simulation setting of Section~\ref{subsec:synthetic} and bike-sharing experiment of Section~\ref{subsec:realdata}. The simulation uses the dense, no-intercept model of Section~\ref{subsec:synthetic} with $d=20$, $N=400$, and 1,000 independent paired repetitions. Within each repetition, all four methods are fitted to the same generated sample. The bike-sharing data and evaluation metrics follow Appendix~\ref{app:realdata}.

For the numerically converged fits, both unpenalized fits start at zero and are accepted only when
\[
    \frac{\|F_N(\hat{\bbeta})\|_\infty}{1+\|\hat{\bbeta}\|_\infty} \le 10^{-8},
    \qquad F_N = V_N\ \text{(VI)},
    \quad F_N = \nabla\mathcal{L}_N\ \text{(MLE)},
\]
where both fields are sample averages. Synthetic VI uses SciPy's hybrid root solver with an analytic Jacobian. Bike VI uses damped Newton updates with an analytic Jacobian and a squared-residual line search. MLE uses unconstrained SciPy L-BFGS-B with an analytic gradient.

Table~\ref{tab:converged_softplus} shows that the large fixed-budget MLE disadvantage disappears at convergence in the correctly specified simulation, where MLE attains slightly lower parameter MSE than VI. This is consistent with the asymptotic efficiency comparison under correct specification in Appendix~\ref{app:eff}. On bike sharing, converged MLE substantially narrows the gap, but converged VI retains lower test deviance despite MLE's lower training NLL.

\begin{table}[!htbp]
\centering
\caption{Fixed-budget and numerically converged softplus comparisons. For simulation, entries are mean squared parameter errors (SDs). Training NLL for bike sharing is the mean of $\mu_i - y_i \log\mu_i$, omitting $\log(y_i!)$. Fixed budgets are $k=200$ for simulation and $k=250$ for bike sharing.}
\label{tab:converged_softplus}
\small
\setlength{\tabcolsep}{4pt}
\begin{tabular}{@{}lcccc@{}}
\toprule
& \multicolumn{2}{c}{Fixed iteration budget}
& \multicolumn{2}{c}{Numerically converged} \\
\cmidrule(lr){2-3}\cmidrule(l){4-5}
Setting / metric & VI & MLE-GD & VI root
& \makecell{MLE\\(L-BFGS-B)} \\
\midrule
Simulation: parameter MSE
& 0.241 (0.054) & 0.343 (0.051) & 0.175 (0.056) & \textbf{0.170} (0.056) \\
\midrule
Bike GLM: test deviance & 133.99 & 453.85 & \textbf{76.63} & 100.46 \\
Bike GAM: test deviance & 129.48 & 428.88 & \textbf{74.67} & 96.52 \\
Bike GLM: training NLL & $-603.88$ & $-541.76$ & $-528.13$ & $-617.20$ \\
Bike GAM: training NLL & $-600.48$ & $-547.67$ & $-524.09$ & $-617.90$ \\
\bottomrule
\end{tabular}
\end{table}



\end{document}